\documentclass[11pt,letterpaper]{article}
\usepackage[margin=1in]{geometry}

\usepackage{amsmath,amsthm,amssymb,amsfonts}
\usepackage{xspace}
\usepackage{url}
\usepackage{hyperref}
\usepackage[dvipsnames]{xcolor}
\usepackage{algorithm}
\usepackage{algpseudocode}
\usepackage{appendix}
\usepackage{babel}
\usepackage[capitalise,noabbrev]{cleveref}
\usepackage{enumitem}

\usepackage{multirow}
\usepackage{tablefootnote}
\usepackage{afterpage}
\usepackage{lipsum}
\usepackage{booktabs}
\usepackage{thm-restate}
\usepackage{graphicx}
\usepackage{subcaption}
\usepackage[normalem]{ulem}
\usepackage{dsfont}
\usepackage{thm-restate}

\hypersetup{colorlinks=true,allcolors=blue}

\theoremstyle{plain}
\newtheorem{theorem}{Theorem}[section]
\newtheorem{lemma}[theorem]{Lemma}
\newtheorem{claim}[theorem]{Claim}

\newtheorem{fact}[theorem]{Fact}

\theoremstyle{definition}
\newtheorem{definition}[theorem]{Definition}
\newtheorem{example}[theorem]{Example}

\theoremstyle{remark}
\newtheorem{remark}[theorem]{Remark}

\newcommand{\ProblemName}[1]{\textsc{#1}}
\newcommand{\kMedian}{\ProblemName{$k$-Median}\xspace}

\DeclareMathOperator{\E}{\mathbb{E}}
\DeclareMathOperator{\Var}{Var}
\newcommand{\Z}{\mathbb{Z}}
\newcommand{\R}{\mathbb{R}}
\def\RR{{\mathbb{R}}}

\DeclareMathOperator{\poly}{poly}
\DeclareMathOperator{\polylog}{polylog}
\DeclareMathOperator{\cost}{cost}
\DeclareMathOperator{\OPT}{OPT}

\DeclareMathOperator{\dist}{dist}
\DeclareMathOperator{\Diam}{diam}
\DeclareMathOperator{\sub}{sub}
\DeclareMathOperator{\seg}{seg}
\DeclareMathOperator{\NIL}{\bot}

\DeclareMathOperator{\AB}{JK}
\DeclareMathOperator{\PA}{\pi_J}
\DeclareMathOperator{\PB}{\pi_K}
\DeclareMathOperator{\AG}{AG}

\newcommand{\indic}[1]{I(#1)}

\newcommand{\FMP}{\ensuremath{F^{\mathrm{MP}}}\xspace}

\newcommand{\hatQ}{\ensuremath{\widehat{Q}}\xspace}
\newcommand{\calL}{\ensuremath{\mathcal{L}}\xspace}
\newcommand{\wopen}{\ensuremath{\mathfrak{f}}\xspace}
\newcommand{\hash}{\ensuremath{\varphi}\xspace}

\def\eps{\epsilon}
\let\epsilon\varepsilon

\def\compactify{\itemsep=0pt \topsep=0pt \partopsep=0pt \parsep=0pt}

\title{Streaming Facility Location in High Dimension\\ via Geometric Hashing\thanks{An extended abstract of this paper appeared in FOCS 2022 \cite{CJKVY22}. 
    The current version contains a new hashing scheme, 
    which improves the one-pass approximation factor from $O(d^{1.5})$ to $O(d/\log d)$
    and yields a new $O(1/\eps)$-approximation with sublinear space $O(n^{\eps})$.
  } 
}

\author{Artur Czumaj\thanks{Research partially supported by the Centre for Discrete Mathematics and its Applications (DIMAP), by a Weizmann-UK Making Connections Grant, by an IBM Faculty Award, and by EPSRC award EP/V01305X/1.
    Email: \texttt{A.Czumaj@warwick.ac.uk}
  }\\
    University of Warwick
    \and
    Arnold Filtser\thanks{Research supported by the Israel Science Foundation grant \#1042/22).
      Email: \texttt{arnold.filtser@biu.ac.il}
    }\\
    Bar Ilan University
    \and
    Shaofeng H.-C. Jiang\thanks{Research partially supported by a national key R\&D program of China No.\ 2021YFA1000900,
    a startup fund from Peking University, and the Advanced Institute of Information Technology, Peking University.
    Email: \texttt{shaofeng.jiang@pku.edu.cn}
  }\\
  Peking University
  \and
  Robert Krauthgamer\thanks{Work partially supported by ONR Award N00014-18-1-2364,
    by the Israel Science Foundation grant \#1086/18,
    by a Weizmann-UK Making Connections Grant,
    by a Minerva Foundation grant,
    and the Weizmann Data Science Research Center. 
    Email: \texttt{robert.krauthgamer@weizmann.ac.il}
  }\\
  Weizmann Institute of Science
  \and
  Pavel Vesel\'y\thanks{Partially supported 
        by GA \v{C}R project 22-22997S
        and by Center for Foundations of Modern Computer Science (Charles University project UNCE/SCI/004). 
Email: \texttt{vesely@iuuk.mff.cuni.cz}}\\
  Charles University
  \and Mingwei Yang\thanks{Email: \texttt{yangmingwei@pku.edu.cn}}\\
  Peking University
}

\date{}

\begin{document}
\maketitle
\begin{abstract}
In \emph{Euclidean Uniform Facility Location},
the input is a set of clients in $\RR^d$ 
and the goal is to place facilities to serve them,
so as to minimize the total cost of opening facilities plus connecting the clients. 
We study the classical setting of dynamic geometric streams,
where the clients are presented as a sequence of insertions and deletions
of points in the grid $\{1,\ldots,\Delta\}^d$,
and we focus on the \emph{high-dimensional regime},
where the algorithm's space complexity 
must be polynomial (and certainly not exponential) in $d\cdot\log\Delta$.

We present a new algorithmic framework,
based on importance sampling from the stream, 
for $O(1)$-approximation of the optimal cost
using only $\poly(d\cdot \log\Delta)$ space.
This framework is easy to implement in two passes,
one for sampling points and the other for estimating their contribution.
Over random-order streams, we can extend this to a one-pass algorithm
by using the two halves of the stream separately.
Our main result, for arbitrary-order streams,
computes $O(d / \log d)$-approximation in one pass
by using the new framework but combining the two passes differently.
This improves upon previous algorithms (with space polynomial in $d$) 
which guarantee only $O(d\cdot \log^2 \Delta)$-approximation,
and therefore our algorithms 
are the first to avoid the $O(\log\Delta)$-factor in approximation
that is inherent to the widely-used quadtree decomposition.
Our improvement is achieved by employing a geometric hashing scheme
that maps points in $\RR^d$ into buckets of bounded diameter, 
with the key property that every point set of small-enough diameter
is hashed into only a few buckets.
By applying an alternative bound for this hashing,
we also obtain an $O(1 / \eps)$-approximation in one pass, using larger but still sublinear space $O(n^{\eps})$ where $n$ is the number of clients.

Finally, we complement our results with a proof that
every $1.085$-approximation streaming algorithm
requires space exponential in $\poly(d\cdot \log\Delta)$,
even for insertion-only streams.
\end{abstract}

\section{Introduction}
\label{sec:intro}

\emph{Facility Location} is a classical problem
in combinatorial optimization and operations research,
and models a scenario where one wishes to find a placement of facilities
that optimizes the total service cost for a given set of customers.
The cost has two different parts:
opening a facility at a location incurs a so-called \emph{opening cost},
and serving each customer incurs a so-called \emph{connection cost}.
The goal is to minimize the sum of both costs.
Typical examples of applications include placement of servers in a network
and location planning.

\paragraph{Euclidean Uniform Facility Location.}

In the \emph{Uniform Facility Location (UFL)} problem,
all possible facilities have the same opening cost $\wopen>0$.
This is essentially a clustering problem; it is similar to $k$-median,
except that the number of clusters is not prescribed in advance,
but rather optimized by adding to the objective a regularization term
(proportional to the number of clusters).

We consider the Euclidean version of UFL,
where the data points and facilities all lie in $\RR^d$.
Formally, given as input a dataset $P\subseteq \RR^d$ and $\wopen>0$,
the goal is to \emph{open} a set $F\subseteq \RR^d$ of \emph{facilities},
so as to minimize the total \emph{connection cost} (the cost of connecting
each data point to its nearest facility according to the Euclidean distance)
plus the total \emph{opening cost}
(opening each facility costs~$\wopen$).
That is, the goal is to find $F\subseteq \RR^d$ that minimizes
\[
    \cost(P, F) := \sum_{p \in P}{\dist(p, F)} + \wopen \cdot |F|
    \enspace,
\]
where we denote $\dist(p,F) := \min_{q \in F} \dist(p,q)$
and $\dist(p,q) := \|p - q\|_2$ is the Euclidean distance.
This problem has been studied extensively for decades
in many algorithmic settings,
including offline, online, dynamic, sublinear-time, streaming,
and so forth;
see \Cref{sec:related} for an overview and references.

\paragraph{Streaming Setting.}

Euclidean UFL has also played a special role
in the study of algorithms for dynamic geometric streams.
Indeed, the seminal paper by Indyk~\cite{Indyk04}, which introduced this model,
considered UFL as one of its benchmark problems,
together with the minimum spanning tree, minimum-weight (bichromatic) matching,
and $k$-median problems.

In the setting of \emph{dynamic geometric streams},
the input $P$ is presented as a stream of insertions and deletions
of points from $[\Delta]^d := \{1, 2, \dots, \Delta\}^d$.
The algorithms should be \emph{space-efficient},
ideally using space that is polylogarithmic in $\Delta^d$,
which is a natural benchmark because it is polynomial
in the representation of a single point using $O(d\cdot\log\Delta)$ bits.
Due to this constraint, we cannot store the
actual solution which could take space $\Omega(\Delta^d)$,
hence the algorithms only approximate the optimal \emph{cost} $\OPT$, which in case of UFL is $\OPT := \min_{F \subseteq \RR^d} \cost(P,F)$.
We say that a randomized algorithm achieves \emph{$\alpha$-approximation},
for $\alpha\ge1$,
if it outputs an estimate $E\ge0$ that with probability at least $2/3$
satisfies $\OPT \le E \le \alpha \cdot \OPT$.
(In many cases, this success probability can be amplified using standard methods.)

In this dynamic geometric streaming model, algorithms generally exhibit a dichotomy, as their space complexity is
either polynomial or exponential in the dimension $d$.
Obviously, algorithms whose space complexity is exponential in $d$
are applicable only in a low-dimensional setting,
say $d=O(1)$ or $d=O(\log\log\Delta)$.
This class contains many algorithms that achieve an $O(1)$-approximation,
or even $(1+\epsilon)$-approximation for arbitrary fixed $\epsilon>0$,
from facility location \cite{LammersenS08,CzumajLMS13}
to minimum spanning tree \cite{FIS08},
and several other basic geometric problems \cite{DBLP:conf/stoc/FrahlingS05, DBLP:conf/focs/AndoniBIW09,CJKV22}. The list would extend even further
if we included algorithms for the insertion-only model.

One would clearly prefer algorithms whose space complexity is polynomial in $d$.
Such algorithms are known for facility location and several other problems,
but unfortunately their approximation ratio is usually high,
for example $O(d\cdot\log\Delta)$ or even higher;
see, e.g., \cite{Indyk04,DBLP:conf/soda/AndoniIK08,CJLW22}.
A major open problem in the area is to break this barrier
and significantly improve the approximation factor, say to $O(1)$ if not $1+\epsilon$,
to match the ratios obtainable in low dimension,
while using space polynomial in $d\cdot\log\Delta$.
This question was partially resolved for
$k$-median~\cite{DBLP:conf/icml/BravermanFLSY17} and
$k$-means~\cite{DBLP:journals/corr/abs-1802-00459},
by designing coresets,
to achieve $(1+\epsilon)$-approximation using space $k\cdot \poly(d\cdot\log\Delta)$ for a fixed $\epsilon>0$;
however, these bounds are low-space only for small values of $k$.

As for UFL, the state-of-the-art streaming algorithms achieve only:
$O(d\cdot\log^2\Delta)$-approximation using space $\poly(d\cdot\log\Delta)$ \cite{Indyk04};
$O(1)$-approximation for constant dimension $d$ \cite{LammersenS08}
(which seems to generalize to $2^{O(d)}$-approximation using $2^{O(d)}$-space); and
$(1+\epsilon)$-approximation for $d=2$ \cite{CzumajLMS13}
(even if this approach could be extended to general $d$, its space bound seems to be $(\log\Delta)^{\Omega(d)}$).
Since one can expect many applications of UFL (or clustering in general) to require high dimension~$d$, the above bounds are unsatisfactory for either having large approximation ratio or using prohibitively large space.

The main barrier for closing this gap is the lack of suitable techniques for high-dimensional spaces.
In particular, the quadtree subdivision is a natural geometric decomposition technique used in almost all previous geometric streaming algorithms,
allowing for a lot of strong results in low dimension~\cite{DBLP:conf/stoc/FrahlingS05,DBLP:conf/focs/AndoniBIW09,FIS08,CzumajLMS13,CJKV22}. However, a quadtree in high dimension, as used by Indyk~\cite{Indyk04} and more recent follow-up papers~\cite{DBLP:conf/soda/AndoniIK08,CJLW22},
gives a tree embedding, which is very useful algorithmically,
but necessarily distorts distances by an $\Omega(d\cdot\log\Delta)$-factor.
It seems that new techniques must be developed
in order to get better approximation ratio in high dimension,
but surprisingly, very little, if at all, is known beyond the quadtree techniques.
This technical barrier is not specific to UFL,
but rather applies to many geometric problems.
\subsection{Results}
\label{subsec:our-results}

We break this barrier of existing techniques,
and devise a geometric \emph{importance-sampling} approach
that is based on a \emph{geometric hashing scheme}.
Using this technique, we obtain improved approximation for UFL
in high-dimensional streams, i.e., using $\poly(d\cdot\log\Delta)$ space.
We present our new streaming bounds for UFL below,
and defer discussion of our technical contributions to \cref{sec:tech}.
Throughout, we assume that all data points in $P$ are distinct,
and that the streaming algorithm knows in advance
the parameters $d$, $\Delta$ and $\wopen$.
These are relatively standard assumptions in the model of dynamic geometric streams,
and we omit them from the statements below.\footnote{It suffices to know
  an $O(1)$-approximation to $\log\Delta$ and to $\wopen$ for all our results.
  The assumption about distinct points was made also in~\cite{Indyk04},
  and can sometimes be removed easily,
  e.g., in insertion-only streams using a random perturbation.
}

We begin by presenting a \emph{two-pass} streaming algorithm
that computes an $O(1)$-approximation to the cost of UFL.
The underlying idea is to apply in the first pass the abovementioned importance-sampling method to select a sample of points from $P$,
and then in the second pass to compute the contribution to the objective of each sampled point
(for details see \Cref{sec:tech}).
This result demonstrates the significant advantage of our technique
over the quadtree/tree-embedding based methods,
which do not yield an $O(1)$-approximation, even in two passes.
This limitation of previous techniques can be seen also in
recent work~\cite{CJLW22} that achieves
$O(\log n)$-approximation for bichromatic matching (earth mover distance)
in two passes,
where $n = |P|$ is the number of input points
(their approach extends to one pass, albeit with an additive error).

\begin{restatable}[\textbf{\emph{Two-Pass Algorithm}}]{theorem}{TwoPassAlg}
	\label{thm:two_pass-intro}\label{thm:two_pass}
	There is a two-pass randomized algorithm that
	computes an $O(1)$-approximation
	to Uniform Facility Location of an input $P\subseteq [\Delta]^d$
	presented as a dynamic geometric stream,
	using $\poly(d\cdot\log\Delta)$ bits of space.
\end{restatable}

We can extend the approach from \cref{thm:two_pass-intro} to \emph{random-order insertion-only streams}
(where the stream is a uniformly random permutation of $P$ and we assume $|P|$ is given),
by using the first half of the stream to simulate the first pass of the algorithm from \cref{thm:two_pass-intro},
and then using the second half to approximate the estimator used in the second pass of the algorithm.
There is clearly correlation between the two halves of the stream,
and neither half represents the full input faithfully,
but we can circumvent these obstacles
and obtain a one-pass $O(1)$-approximation algorithm, as stated below.
We are not aware of previous geometric streaming algorithms in high dimension,
particularly those based on quadtrees,
that obtain $O(1)$-approximation in random-order streams
(apart from those for $k$-median~\cite{DBLP:conf/icml/BravermanFLSY17} and $k$-means~\cite{DBLP:journals/corr/abs-1802-00459}, though these bounds are low-space only for small $k$).

\begin{restatable}[\textbf{\emph{Random-Order Algorithm}}]{theorem}{RandomOrderAlg}
\label{thm:random_order-intro}\label{thm:random_order}
There is a one-pass randomized algorithm that
computes an $O(1)$-approximation
to Uniform Facility Location of an input $P\subseteq [\Delta]^d$
presented as a random-order insertion-only stream,
using $\poly(d\cdot\log\Delta)$ bits of space.
\end{restatable}

Finally, to deal with arbitrary-order dynamic streams in one pass,
we have to significantly extend our methods. Specifically, our previous scheme of first obtaining samples
and then computing their estimations cannot work,
and we need our importance sampling to provide additional structure.
This more involved step crucially relies on
a powerful property in our geometric hashing scheme,
namely, that every point set of small-enough diameter
is hashed to only $\poly(d)$ distinct buckets.
However, this also introduces an $O(d / \log d)$-factor in the approximation ratio,
which corresponds to a key performance parameter of the hashing.

\begin{theorem}[\textbf{\emph{One-Pass Algorithm}}, see \cref{thm:1p} and \Cref{remark:ratio}]
\label{thm:1p-intro}
There is a one-pass randomized algorithm that
computes an $O(d/ \log d)$-approximation
to Uniform Facility Location of an input $P\subseteq [\Delta]^d$
presented as a dynamic geometric stream,
using $\poly(d\cdot\log\Delta)$ bits of space.
\end{theorem}

While the approximation ratio of $O(d / \log d)$ may look prohibitively large,
it is useful to note that many geometric problems, including facility location, are reducible to the case $d = O(\log n)$, where $n = |P|$, by using standard techniques relying on the Johnson–Lindenstrauss lemma.
Furthermore, this improves the previously best streaming algorithm for UFL with $\poly(d\cdot\log\Delta)$ space, by Indyk~\cite{Indyk04}, which has an approximation factor of $O(d\cdot\log^2\Delta)$.

Other streaming algorithms for UFL in the literature focus solely
on the case of constant dimension $d$,
and for general $d$ seem to use space exponential in $d$.
In particular, the algorithm of Lammersen and Sohler~\cite{LammersenS08}
seems to generalizes to $2^{O(d)}$-approximation using $2^{O(d)} \cdot \polylog\Delta$ space (worse approximation ratio than our algorithm),
and Czumaj et al.~\cite{CzumajLMS13} designed a $(1+\epsilon)$-approximation for UFL in the Euclidean plane $d = 2$, 
which for a fixed $\epsilon>0$ seems to require space $(\log \Delta)^{\Omega(d)}$ even if it can be extended to dimension $d$,
meaning that the space would be superlogarithmic for $d = \omega(1)$.

Finally, we remark that \Cref{thm:1p} offers a more general space-approximation tradeoff
that arises from the parameters of the geometric hashing.
In particular, it implies a one-pass $O(1 / \eps)$-approximation
using $\tilde{O}(n^{\eps})$ space (see \Cref{remark:ratio}),
which is the first one-pass algorithm for UFL
that achieves $O(1)$-approximation using \emph{sublinear} space.
This $n^{\eps}$ space bound is generally considered to be weaker 
than the $\poly(d \log \Delta)$ that we aim for,
however this regime seems to be effective for achieving $O(1)$-approximation,
e.g., a recent one-pass algorithm for minimum spanning tree
achieves $\poly(1/\eps)$-approximation using space $O(n^{\eps})$ ~\cite{CCJLW}.

The algorithmic results above are complemented by the following lower bound,
proved by a reduction from the one-way communication complexity
of the Boolean Hidden Matching problem.

\begin{theorem}[\textbf{\emph{Streaming Lower Bound}}; see \cref{thm:lb-of-ufl-in-log-space}]
\label{thm-lower-bound-intro}
Let $d\ge 1$. Every one-pass randomized algorithm that
approximates Uniform Facility Location within ratio better than $1.085$
on an insertion-only stream of points from $[\Delta]^d$ for $\Delta = 2^{O(d)}$,
requires space $2^{\poly(d\cdot \log\Delta)}$.
\end{theorem}

This lower bound excludes a streaming $(1+\epsilon)$-approximation
with $f(\epsilon)\cdot \poly(d\cdot \log\Delta)$ space,
even for insertion-only streams.
In terms of $n = |P|$, the space lower bound is $\Omega(\sqrt{n})$,
in a setting where $d = \Theta(\log n)$ and $\Delta = \poly(n)$.

\subsection{Technical Overview}
\label{sec:tech}
For simplicity, in this section we consider the insertion-only setting,
and assume that the instance is scaled
so that the opening cost is $\wopen = 1$.
Then the input $P\subseteq\RR^d$ of size $n = |P|$ is not restricted to a discrete grid.
Our overall strategy is to design a streaming implementation of
an estimator known to achieve an $O(1)$-approximation.
This estimator was proposed in~\cite{BadoiuCIS05}
in the context of sublinear-time algorithms
(based on an offline $O(1)$-approximation algorithm in~\cite{MettuP03}).
The idea is to associate to every data point $p \in P$
a value $r_p\in[1/n,1]$ (formally defined in Definition~\ref{def:rp}),
that satisfies two key properties (Fact~\ref{fact:rp}):
First, the sum of all $r_p$'s gives an $O(1)$-approximation to UFL,
i.e., $\sum_{p \in P}r_p = \Theta(\OPT)$.
Second, $r_p$ is roughly the inverse of the number of points inside
the ball $B(p, r_p)$ centered at $p$ with radius $r_p$,
i.e., $|P\cap B(p, r_p)| = \Theta(1/r_p)$.
It follows that $r_p$ can be estimated (see Fact~\ref{fact:rp_estimation})
by just counting the number of points in the balls $B(p, 2^{-j})$ for $j=1,\ldots,\log_2 n$,
which is easy if $p$ is known at the beginning of the stream.
However, if $p$ is given as a \emph{query} at the end of the stream,
then any finite approximation requires $\Omega(n)$ space,
by a reduction from the communication complexity of indexing
(see Appendix~\ref{sec:LBs-index}).

\subsubsection{Importance Sampling via Geometric Hashing}
\label{subsubsec:importance-sampling-and-geometric-hashing}

\paragraph{Na\"ive Approach: Uniform Sampling.}
Consider initially making two passes over the stream,
the first one samples a few points,
and the second pass estimates the $r_p$ value for each sampled point $p$.
Since all $r_p \in (0, 1]$,
one immediate idea is to perform \emph{uniform sampling},
and argue using Chernoff bounds that the resulting scaled estimate is likely to be
$\Theta(\sum_{p \in P}r_p) = \Theta(\OPT)$.
However, to obtain decent concentration,
one needs the expectation $\sum_{p \in P} r_p$ to be large enough,
which need not hold.
Indeed, consider Example~\ref{exp:hard} below,
where
uniform sampling needs to draw $\Omega(\sqrt{n})$ samples
to have a decent chance to see even one point $p$ with $r_p=1$,
which is necessary for obtaining a nontrivial approximation.

\begin{example}  \label{exp:hard}
Let $P = P_1 \cup P_2$,
where $P_1$ consists of $\sqrt{n}$ points,
whose pairwise distances are at least $1$,
and $P_2$ consists of $n - |P_1| = \Theta(n)$ points
whose pairwise distances are (approximately) $1/n$;
these two sets are at distance at least $1$ from each other.
(A realization of this point set is possible in dimension $\Theta(\log n)$.)
One can easily verify from the definition that
all $x\in P_1$ have $r_x=1$, and all $x\in P_2$ have $r_x = \Theta(1/n)$,
thus $\OPT = \Theta(\sum_{p \in P} r_p) = \Theta(\sqrt{n})$.
\end{example}

To bypass the limitation of uniform sampling,
we can employ \emph{importance sampling}:
sample one point $p^*\in P$,
such that each $p\in P$ is picked with probability roughly proportional to $r_p$,
and construct an unbiased estimator $\widehat{Z} = r_{p^*} / \Pr[p^*]$. By a standard analysis, such an estimator has low variance,
and thus averaging a few independent samples yields an accurate estimate.
To get some intuition, in Example~\ref{exp:hard},
when sampling proportionally to $r_p$,
the total sampling probability of the points $p$ with $r_p=1$
is far larger than that of the remaining points.

However, this importance sampling idea is difficult to implement in streaming.
At first glance, it is a chicken-and-egg problem:
importance sampling aims to estimate $\sum_{p} r_p$,
but it requires knowing the $r_p$ values.
The crux is that a coarse estimation for $r_p$ suffices for importance sampling,
but as noted above, computing $r_p$ for point queries with any finite ratio
requires $\Omega(n)$ space.
Moreover, even if the $r_p$ values of the sampled points could be estimated at the end of the stream,
how would a streaming algorithm draw samples proportionally to these estimates?

\paragraph{New Idea: Geometric Importance Sampling.}
Instead of trying to estimate the value of $r_p$,
we implement importance sampling indirectly (in Theorem~\ref{thm:impt}),
using a geometric hashing scheme $\hash:\RR^d\to\RR^d$
that ``isolates'' points with large $r_p$.
As usual in hashing, the codomain of $\hash$ is somewhat arbitrary
(e.g., in applications it could be $[\Delta]^d$),
and a \emph{bucket} refers to a preimage $\hash^{-1}(z)$,
i.e., the set of points mapped to the same image $z$.
Ideally, we would like the aforementioned points (with large $r_p$)
to each get its own bucket,
and the others (small $r_p$) to collide, say, into one bucket per cluster,
and thus take up only a few buckets.
Given such a hashing scheme, we apply it to all the points in $P$
and pick a non-empty bucket at random.
This is equivalent to sampling uniformly from the hash values (of all points in $P$),
and is easily implemented in streaming using a well-known tool
called the $\ell_0$-sampler; see e.g.~\cite{CormodeF14}
(this tool produces a uniform sample from the \emph{distinct elements} of a stream,
and applying it here will sample uniformly from the distinct hash values).

The geometric hashing scheme,
when combined with subsampling, guarantees that with high probability,
\begin{enumerate}[label=\alph*)] \compactify
\item the number of non-empty buckets is bounded by
  $\poly(d\cdot\log\Delta) \cdot \OPT$,
  and
\item every bucket with at least one point of large $r_p$
  contains at most $\poly(d\cdot\log\Delta)$ points.
\end{enumerate} We may assume that points of large $r_p$
constitute a significant fraction of $\sum_p r_p = \Theta(\OPT)$
(because we can anyway neglect a subset of points whose contribution is low),
and employ the following \emph{two-level uniform sampling}:
First sample uniformly a non-empty bucket of the hashing scheme,
and then sample uniformly a point from that bucket.
Now, the probability of sampling a point $p$ of large $r_p$
is at least $1/\poly(d\cdot\log \Delta)$ (this holds whenever $r_p > 1/\poly(d\cdot\log \Delta)$).
Thus, by taking $\poly(d\cdot\log \Delta)$ samples
we are likely to hit at least one point of large $r_p$,
which in fact leads to a robust estimator.
The aforementioned two-level uniform sampling is implemented
by extending a standard construction of the $\ell_0$-sampler (in \Cref{lemma:2dl0}),
inspired by a different extension in~\cite{FIS08}.

This implementation of importance sampling bypasses the straightforward approach
of first estimating the desired values (in our case $r_p$)
and then sampling accordingly,
as done previously in some fast algorithms,
e.g., for counting~\cite{KL83} and for geometric problems \cite{Indyk07},
and in ``data-compression'' algorithms,
e.g., constructing graph sparsifiers \cite{BK15, SS11}
and geometric coresets \cite{FeldmanL11, FeldmanSS20}.
Previously, such a non-straightforward implementation of importance sampling
was employed in streaming algorithms for matrices \cite{LW16a,BKKS20},
for the same reason that the needed values are hard to compute in a streaming fashion.

\paragraph{Consistent Geometric Hashing with Bounded Gap.}
We now elaborate on the geometric hashing,
which plays a central role in our importance-sampling algorithm
and is formally defined as follows.
Throughout, $\Diam(S)$ denotes the diameter of $S\subseteq\RR^d$.

\begin{restatable}[Consistent Hashing]{definition}{ConsistentHashing}
\label{def:decomp}
A mapping $\hash:\RR^d\to\RR^d$ is called a
\emph{$\Gamma$-gap $\Lambda$-consistent hash} with diameter bound $\ell>0$,
or simply \emph{$(\Gamma, \Lambda)$-hash},\footnote{While the parameter $\ell$ is important when applying the hashing,
  when constructing the hashing one can assume by scaling that $\ell=1$.
}
if it satisfies:
\begin{enumerate} \compactify
\item Diameter:
  for every image $z\in \hash(\RR^d)$, we have $\Diam(\hash^{-1}(z)) \leq \ell$; and
\item Consistency:
  for every $S\subseteq\RR^d$ with $\Diam(S) \leq \ell / \Gamma$,
  we have $|\hash(S)| \leq \Lambda$.
\end{enumerate}
\end{restatable}
Since we use the hash for designing streaming algorithms, we also require that storing a description of $\hash$ only uses some small space $s$, and given a point $p\in \RR^d$, there is an $s$-space algorithm that returns $\hash(p)$. In our applications, we need $s = \poly(d)$.

Intuitively, the first condition (diameter) requires that further apart points
are never hashed (mapped) to the same bucket, and
the second one (consistency) requires that highly-clustered points,
even if their number is very large, are hashed to only a few different buckets.

We present this definition with a general parameter $\Lambda$,
but in many of our applications we require $\Lambda:=\poly(d)$, which is sufficient as
our algorithms have their approximation ratios independent of $\Lambda$
and space only polynomially depending on $\Lambda$.
However, the ``gap'' parameter $\Gamma$, equal to the ratio between the diameter bound $\ell$
and the consistency diameter $\ell/\Gamma$, goes into the approximation factor of our
one-pass streaming algorithm (Theorem~\ref{thm:1p-intro}), besides affecting the space complexity polynomially.
Indeed, this dependence of $\Gamma$ in the ratio also allows us to use alternative tradeoffs of $\Gamma$ and $\Lambda$, particularly $\Gamma = O(1 / \eps)$ and $\Lambda = \tilde{O}(n^{\eps})$ to achieve $O(1 / \eps)$-approximation using space $\tilde{O}(n^{\eps})$.

\paragraph{Comparison to Related Geometric Decompositions.}
Our definition of consistent hashing is essentially equivalent to
  the notion of sparse partitions that was introduced by
  Jia, Lin, Noubir, Rajaraman, and Sundaram~\cite{DBLP:conf/stoc/JiaLNRS05}.
Their definition concerns a \emph{partition} of $\mathbb{R}^d$,
  if we view each part in that partition as a bucket in a hashing,
  then their definition is the same as \cref{def:decomp}.
  However, in our setting it is more natural to think of a hash function,
  because the part/bucket that contains each input point $x \in \mathbb{R}^d$
  must be computed \emph{using a small amount of memory}, and a mere partition of $\mathbb{R}^d$ may not suffice.
Jia et al.~\cite{DBLP:conf/stoc/JiaLNRS05} primarily focused on general metric spaces and for Euclidean spaces only give
  a construction with consistency $\Lambda=2^d$ and gap $\Gamma=\Theta(\sqrt{d})$(by a straightforward partition into hypercubes),
  which is thus not useful in our context.
  Notably, Filtser~\cite{DBLP:conf/icalp/Filtser20}
  designed a sparse partition with consistency $\Lambda=\poly(d)$
  and gap $\Gamma=O(d / \log d)$,
  however, in this partition the description of a part takes $\Omega(2^d)$ bits,
  and thus does not directly imply
  a hash function that can be evaluated on a point in small space.

\paragraph{New Constructions for Consistent Hashing}
  We provide a new construction (based on the ball carving approach of Andoni and Indyk \cite{AI08}) that achieves the same consistency and gap bounds as in~\cite{DBLP:conf/icalp/Filtser20} (even by a constant factor better)
  while the hash function can be stored in $\poly(d)$ space and the hash value for any point in $\mathbb{R}^d$ can be evaluated using space $\poly(d)$ (\Cref{thm:EucSparsePartition})\footnote{Filtser's construction in~\cite{DBLP:conf/icalp/Filtser20} is in fact more general in that it works for any metric space with a bounded doubling dimension and provides strong diameter bounds, whereas our construction in \Cref{thm:EucSparsePartition} focuses specifically on Euclidean spaces and gives weak diameter bounds, which is sufficient for our purpose.}.
  Namely, we also get a general smooth tradeoff between $\Gamma$ and $\Lambda$, namely, $\Lambda = \exp(O(d/ \Gamma)) \cdot \poly(d)$,
  and this tradeoff was shown to be nearly tight~\cite{DBLP:conf/icalp/Filtser20} (see \Cref{remark:hash_tight}).

However, this nearly tight tradeoff is achieved at the cost of a large running time of $\exp(d)$.
  Hence, we also give another construction with parameters $\Gamma = O(d^{1.5})$ and $\Lambda = \poly(d)$
  such that the hash value can be evaluated using space \emph{and time} $\poly(d)$ (\Cref{thm:exist_decomp}).
  This more time-efficient construction,
  although with a worse gap-consistency tradeoff, already suffices for many of our applications, particularly the two-pass (\Cref{thm:two_pass-intro}) and random-order (\Cref{thm:random_order-intro}) results,
  and thus leads to a $\poly(d \log \Delta)$ running time in addition to small $\poly(d \log \Delta)$ space.

In fact, the same gap bound of $\Gamma = O(d^{1.5})$
	was recently obtained by Dunkelman et al.~\cite{DBLP:conf/icalp/DunkelmanGKKRST21}.
	It is not stated there explicitly (but can be verified by inspecting their analysis)
	because their construction is designed for a different notion, called \emph{consistent rounding},
	which is incomparable to our \Cref{def:decomp} primarily because
	they require each bucket to have a bounded volume (instead of diameter),
	and their guarantee on the number of intersections is also slightly different.
	Their construction is somewhat similar to the one in \Cref{thm:exist_decomp} except that it works top-down,
	whereas ours works bottom-up.

We note that the second condition in \cref{def:decomp}
does not directly follow from many known methods in the literature.
Indeed, if we partition the space $\RR^d$ using standard methods,
like hypercube subdivisions as in a quadtree (see, e.g.,~\cite{Arora98}),
it is hard to avoid clusters from intersecting $2^{\Omega(d)}$ buckets,
instead of only $\poly(d)$.
Other geometric decompositions,
such as padded decomposition~\cite{LS93,Bartal96,CCGGP98},
and Locality-Sensitive Hashing~\cite{IM98},
aim for different guarantees that are not directly comparable to those of \cref{def:decomp};
see \Cref{sec:related} for a broader comparison.

\medskip

To intuitively see why this geometric hashing helps,
consider the instance in Example~\ref{exp:hard},
and let us focus on estimating the number of points with $r_p = 1$ (for which uniform sampling does not work).
In Example~\ref{exp:hard}, points are grouped into natural clusters:
the whole $P_2$ forms a cluster that consists of $O(n)$ points each with $r_p = O(1 / n)$,
and each point in $P_1$ forms a singleton cluster whose $r_p = 1$.
We construct a $(\Gamma, \Lambda)$-hash with diameter bound $\ell = 1/2$.
By the second guarantee, every small ball of radius $O(1 / n) \ll \ell / \Gamma $
is mapped to $\Lambda$ points/buckets,
even if the ball originally could have $\Omega(n)$ points.
Hence, after the hashing, the entire cluster $P_2$ with $O(n)$ points gets mapped to $\Lambda$
points, while points in $P_1$ are preserved because of the diameter bound of the buckets. Hence, applying the two-level uniform sampling, we hit one point in $P_1$
with at least $\Lambda$ samples on average.
Note that we can achieve $\Lambda = \poly(d)$ in our hashing bounds (\Cref{thm:EucSparsePartition,thm:exist_decomp}),
hence we only need to draw $\poly(d)$ samples which we can afford.

Finally, we note that the actual implementation of this whole idea
is more involved and requires additional steps.
For instance, the nice cluster structure in Example~\ref{exp:hard}
might not be present in a general input,
and our analysis needs to explicitly define a clustering
where a cluster containing $p$ has diameter roughly $r_p$.
Another issue is that our overview focused on $r_p=1$.
For general $r_p$, we use a subsampling at rate $2^{-i}$
to ``reduce'' the case $r_p=2^{-i}$ to the case $r_p=1$,
which is conceptually similar to the subsampling used in the construction of $\ell_0$-samplers.
At the end, this algorithm implements our desired importance sampling task,
namely, it uses space $\poly(d\cdot\log\Delta)$
and produces a sample $p^*\in P$ such that
every $p\in P$ is picked with probability proportional to at least $r_p/\poly(d\cdot\log\Delta)$.

\subsubsection{Streaming Implementations}
\label{subsubsec:streaming-implementations}

The above-mentioned geometric importance sampling (\Cref{thm:impt})
can be implemented in one pass using small space.
However, it only returns a set of \emph{samples} $S\subseteq P$,
and to actually estimate $\sum_{p \in P}r_p$
one still needs to estimate the value $r_p$ for each $p\in S$.
This limitation is a consequence of our importance-sampling approach,
which bypasses estimating the $r_p$ values on purpose.

\paragraph{Two Passes and Random-Order Streams.}
Our two-pass streaming algorithm (in \Cref{thm:two_pass}) is quite simple:
the first pass computes a sample $S\subseteq P$
using importance sampling (\Cref{thm:impt}),
and the second pass estimates $r_p$ for each $p\in S$
using Fact~\ref{fact:rp_estimation} and straightforward counting.
Its space complexity is $\poly(d\cdot\log\Delta)$ per point $p\in S$,
and we need $|S|=\poly(d\cdot\log\Delta)$.

A similar approach can be applied also in the random-order model
(i.e., the stream is a uniformly random permutation of $P$):
the first half of the stream is used to generate a sample $S$, and the second half is used to estimate the $r_p$ value
of the points $p\in S$ sampled in the first half.
However, more technical steps are needed in the analysis,
due to the correlation between the two halves of the stream,
and the fact that a random half does not represent the full stream accurately.

\paragraph{One-pass Implementation.}
The one-pass setting is significantly more difficult.
We estimate $\sum_{p \in P} r_p$,
by partitioning $P$ into levels $i=1,\ldots,d\cdot\log_2\Delta$,
namely, we let $P_i := \{p\in P: r_p\in (2^{-i}, 2^{-i+1}]\}$
and $W_i := \sum_{p \in P_i} r_p$.
We build an estimator for $W_i$ separately for each $i$
(recall that in this section we assume $\wopen = 1$).

For simplicity, we focus here on $i=0$,
so we now only care about points $p\in P_0$, meaning that $r_p = \Theta(1)$.
Then $W_0 = \sum_{p \in P_0}{r_p} = \Theta(|P_0|)$,
and it suffices to estimate $|P_0|$.
To this end, it is natural to use the estimator
$1 / \Pr[x] \cdot \indic{x \in P_0}$,
where $x$ is a sample generated by our importance sampling.
However, the indicator $\indic{x \in P_0}$ turns out to be very sensitive,
and it is difficult to estimate it in one pass even within a constant factor
(e.g., distinguish between $r_x \geq 1/2$ and $r_x\leq 1/10$).
Hence, we have to design a more relaxed tester and analyze how this affects the approximation ratio.

To implement this tester,
we construct an $(\Gamma, \Lambda)$-hash with diameter bound $\ell = 1/10$.
Recalling that, by \Cref{fact:rp}, the $r_p$ values satisfy that
$|P\cap B(p, r_p/2)| \le 2 /r_p$
we observe that if a hashing bucket contains a point with $r_p \ge 1/2$,
then the number of points in this bucket is bounded by $O(1)$.
This means that a bucket either consists only of points with large $r_p$ values,
or no such point at all.
Using this observation,
we maintain a counter for the number of points mapped to every bucket,
and when a point is sampled, we retrieve also the counter for its bucket,
and if the counter is small, we use it as a proxy for the event that
the bucket consists only of points with large $r_p$ values,
including in particular the sampled point.

However, a subtle technical issue is that some point $y$ with $r_y \ll 1$
can possibly lie on the ``boundary'' of the bucket,
and then the number of points in that bucket is small,
while $P\cap B(y,r_y)$ contains many nearby points that lie in other buckets.
Hence, we need to count the number of points in a slightly enlarged region,
i.e., for a bucket $Q \subseteq \mathbb{R}^d$ we need to
count the points in $P$ that fall inside the $r_y$-neighborhood of $Q$,
denoted here as $B(Q, r_y)$
and defined as the set of points in $\RR^d$ at distance at most $r_y$ from $Q$.
Hence, if we maintain the counter for $P\cap B(Q, \beta)$ for some $\beta > 0$
then the information of the counter suffices for rejecting $y$'s whose $r_y \leq O(\beta)$.

Now, suppose we are to maintain the counter for some $\beta$.
To implement this, whenever we see a data point $x$ arrives,
we should increase the counter for all buckets $Q$ such that $x \in B(Q, \beta)$
(a similar trick appears e.g.\ in~\cite{FIS08}).
However, this becomes challenging in the streaming setting,
primarily due to the fact that the number of buckets $Q$ such that $x \in B(Q, \beta)$ can be huge (e.g., $2^d$),
and more importantly, most of them may be ``fake'', in that they do not contain any data point.
Consequently, these fake nonempty buckets can enlarge the support of the $\ell_0$-sampler significantly, which makes it difficult to obtain a uniform sample.

To tackle this challenge, we make use of the guarantee from the geometric hashing,
that any subset of $\mathbb{R}^d$ with diameter less than $\epsilon = O(1 / \Gamma)$
is mapped to $\Lambda$ buckets (see Definition~\ref{def:decomp}).
Hence, if we choose $\beta = \epsilon / 2$, the effect of enlarging the
buckets by an additive $\beta = \epsilon / 2$ is essentially
making every data point $x$ a ball $B(x, \epsilon / 2)$, and add the image
of this ball to the buckets/counters.
The $\Lambda$ intersection bound ensures that the number of fake buckets
is still well bounded, and this $\Lambda$ factor goes into the space.

Finally, the approximate tester is off by a factor of $O(\Gamma)$,
which is the gap of the geometric hashing,
and we show that this factor goes into the approximation ratio,
which will thus be $O(\Gamma)$.

     \subsection{Related Work}
\label{sec:related}

\paragraph{Facility Location Problem.}
The facility location problem is one of the fundamental problems in operations research and combinatorial optimization, and has received extensive studies in the past.
In the offline setting with uniform opening costs,
the facility location problem has been proved to be NP-hard \cite{DBLP:journals/siamcomp/MegiddoS84}, and is hard to approximate within factor $1.463$~\cite{GuhaK99} (unless $\mathbf{NP} \subseteq \mathbf{DTIME}\left[n^{O(\log \log n)}\right]$).
For the upper bounds,
the state of the art for general metrics is a $1.488$-approximation
by Li~\cite{DBLP:journals/iandc/Li13},
and PTAS's are known for special metric spaces,
specifically,
doubling metrics~\cite{DBLP:journals/jacm/Cohen-AddadFS21},
minor-free graphs~\cite{DBLP:journals/siamcomp/Cohen-AddadKM19},
and near-linear time PTAS's for (bounded-dimensional) $\mathbb{R}^d$~\cite{DBLP:journals/siamcomp/KolliopoulosR07} 
and planar graphs~\cite{Cohen-AddadPP19}.
For the online setting,
Meyerson~\cite{Meyerson01} gave an $O(\log n)$-competitive algorithm,
and Fotakis~\cite{Fotakis08} proved that it has ratio $\Theta(\log n / \log \log n)$
and that this ratio is asymptotically optimal.

\paragraph{Geometric Decomposition.}
Geometric decomposition is a topic that was studied extensively,
with many different definitions, even beyond $\RR^d$,
that are motivated by numerous applications.
For brevity, we only mention a few that are closer to our work.
One basic genre, often called \emph{space partitioning},
refers to a partition of $\RR^d$,
perhaps using a variant of the standard grid (quadtree) partition,
e.g.~\cite{AMNSW98}.
Sometimes it is convenient to use multiple space partitions,
or a probability distribution over space partitions
(e.g., a few shifts or a random shift of the grid partition \cite{FK97,Arora98,CHJ20,FL22,Fil22}). 
Another standard requirement is that every part in the partition
has a bounded diameter
(e.g., padded and separating decomposition \cite{LS93,Bartal96,CCGGP98,Fil19}),
or alternatively that every part has a bounded volume
(e.g., \cite{DBLP:conf/icalp/DunkelmanGKKRST21, KORW08}).
The above examples ask that nearby points lie in the same part,
but another type of decomposition, called Locality-Sensitive Hashing (LSH)
\cite{IM98,AI08}, only asks that close-by points fall in the same part
with noticeably higher probability than far-away points.

\paragraph{Two-pass and Random-order Streaming Algorithms.}
Besides the most studied streaming model of one-pass algorithms over an arbitrary (non-random) order streams,
algorithms requiring a few passes or assuming that the stream order is random have received significant attention as well.
For graph streams, two-pass streaming algorithm have been designed,
for example, for graph spanners~\cite{DBLP:conf/podc/KapralovW14,DBLP:conf/soda/FiltserKN21},
maximum matching~\cite{DBLP:conf/approx/KonradN21}, and triangle counting~\cite{DBLP:journals/tcs/CormodeJ17}.
Apart from graph streams, other examples of two-pass algorithms include those for matrix norm estimation~\cite{DBLP:conf/icml/BravermanCKLWY18}, set cover~\cite{DBLP:conf/pods/Assadi17},
and geometric earth mover distance~\cite{CJLW22}.
The one-pass random-order (insertion-only) setting was studied for problems including matching~\cite{DBLP:conf/soda/KapralovKS14,DBLP:conf/icalp/AssadiB21},
quantile estimation~\cite{DBLP:journals/siamcomp/GuhaM09}, graph connected components and minimum spanning tree~\cite{DBLP:conf/soda/PengS18},
and frequency moment estimation~\cite{DBLP:conf/icalp/WoodruffZ21}.

     \section{Preliminaries}
\label{sec:prelim}

\paragraph{Notation.}
We use the usual notation $[n] := \{1, \ldots, n\}$,
and for a function $\varphi : X \to Y$ and $y \in Y$,
we denote $\varphi^{-1}(y) := \{ x \in X : \varphi(x) = y \}$.
The \emph{$d$-dimensional ball} centered at $x \in \RR^d$ with radius $r \geq 0$ 
is defined as $B(x, r) := \{ y \in \RR^d : \dist(x, y) \leq r\}$.

\paragraph{Definitions and Facts from Mettu-Plaxton (MP) Algorithm.}
We will need some machinery from the MP algorithm~\cite{MettuP03,BadoiuCIS05}.
We first introduce the definition of $r_p$,
and then recall useful facts,
particularly that it suffices to approximate $\sum_{p \in P}{r_p}$,
because it $O(1)$-approximates $\OPT$.

\begin{definition}[\cite{MettuP03}]
    \label{def:rp}
    For every $p \in P$, let $r_p$ be the number such that
    \begin{equation} \label{eq:rp}
        \sum_{x \in P\cap B(p, r_p)} \big( r_p - \dist(p, x) \big) = \wopen.
    \end{equation}
\end{definition}
It is easy to see that $r_p$ is well-defined
and $\frac{\wopen}{|P|} \leq r_p \leq \wopen$.
Indeed, using the notation $z^+=\max(z,0)$,
we can write the left-hand side of~\eqref{eq:rp} as
$\sum_{x \in P} (r_p - \dist(p, x))^+$,
which is easily seen to be non-decreasing with $r_p$. 
For illustration, suppose $p$ is one of $k$ points
whose pairwise distances are all equal to $a \in (0, \wopen)$, 
and all other points are at distance at least $\wopen$ from $p$; 
then $r_p = \Theta(a+\wopen/k)$.

\begin{fact}[Lemmas 1 and 2 in \cite{BadoiuCIS05}]
\label{fact:rp}
The following holds.
\begin{itemize}
\item For every $p \in P$,
  it holds that $|P\cap B(p, r_p)|\ge \wopen / r_p$ and $|P\cap B(p, r_p/2)| \le 2\wopen /r_p$.
\item $\sum_{p \in P}{r_p} = \Theta(\OPT)$.
\end{itemize}

\end{fact}
We assume without loss of generality (w.l.o.g.) that $\Delta$ is a power of two.
Let $L := d \cdot \log_2\Delta \geq \log_2 |P|$ as $|P|\le \Delta^d$ by the 
assumption that points are distinct.
The first point of Fact~\ref{fact:rp} implies that the $r_p$ value can be approximated within a constant factor by counting the number of points in balls of geometrically increasing radii.
\begin{fact}\label{fact:rp_estimation}
Let $j_0$ be the maximum $j\in\{0,\ldots,L\}$
such that $|P\cap B(p, 2^{-j} \wopen)|\ge 2^j$.
Then 
\begin{equation}\label{eq:r_p_approx_def_using_balls}
  r_p \in ( 2^{-j_0-1} \wopen, 2^{-j_0+1} \wopen].
\end{equation}
Moreover, there is a one-pass deterministic streaming algorithm that given the opening cost $\wopen > 0$ and a point $p$ in 
advance of the data set $P$ presented as a dynamic stream,
returns an estimate $\hat{r}_p$ such that 
$r_p \le \hat{r}_p\le O(r_p)$ using space of $O(L^2)$ bits.
\end{fact}

\begin{proof}By Fact~\ref{fact:rp},
if $r_p\leq 2^{-j_0-1} \wopen$ then we arrive at the contradiction
$|P\cap B(p,2^{-j_0-1} \wopen)| \geq |P\cap B(p,r_p)| \geq \wopen/r_p \geq 2^{j_0+1}$. 
And if $r_p>2^{-j_0+1} \wopen$ then we arrive at the contradiction
$|P\cap B(p,2^{-j_0} \wopen)| \leq |P\cap B(p,r_p/2)| \leq 2\wopen / r_p < 2^{j_0}$. 
To implement the estimation in dynamic streams,
we count the number of points in each of the balls $B(p, 2^{-j} \wopen)$ for $j = 0,\ldots,L$
using an $L$-bit counter.
\end{proof}

     \section{Importance Sampling via Geometric Hashing}
\label{sec:mp_sample}

In this section, we develop a streaming algorithm for importance sampling on $P$,
where the probability to report each point $x\in P$
is (at least) proportional to its contribution to $\sum_{x \in P} r_x$.
Similarly to other streaming algorithms for sampling (e.g., $\ell_p$-samplers),
our algorithm might fail with a tiny but non-zero probability, 
in our case $q_{\text{fail}} = 1/\poly(\Delta^d)$, 
and the analysis can effectively ignore these events by a union bound.\footnote{Failure in Theorem~\ref{thm:impt} or Lemma~\ref{lemma:2dl0} means that the algorithm may behave arbitrarily,
  e.g., not return anything or even return a point outside $P$,
  and it is not easy to verify if the point is in $P$.
  Formally, having failure probability $q_{\text{fail}}$
  means that the total variation distance
  between the algorithm's output distribution and desired distribution 
  (e.g., uniform over a certain set in the case of $\ell_0$-sampler)
  is at most $q_{\text{fail}}$. 
}
While the algorithm's goal is to sample from $P$,
we also allow it to return $\NIL$, which is \emph{not} considered a failure,
as long as it returns points from $P$ with sufficiently large probability.\footnote{For example, an acceptable output distribution may be
$\NIL$ with probability $\frac12$,
and every $x\in P$ with probability $\frac12 r_x/\sum_{y\in P}r_y$. 
}
The output $\NIL$ is useful in the algorithm's design,
as it can replace the use of a fixed point from $P$, 
and also handle properly the corner case $P=\emptyset$.

\begin{theorem}
\label{thm:impt}
There is a one-pass randomized algorithm that, given $P \subseteq [\Delta]^d$
presented as a dynamic geometric stream,
samples a random point $p^*\in P \cup \{\NIL\}$ such that
\[
  \forall x \in P, \qquad
  \Pr[p^* = x] \geq \Omega\left(\frac{1}{\poly(d\cdot\log\Delta)}\right) \cdot \frac{r_x}{\sum_{y \in P}r_y},
\]
and also reports a $2$-approximation $\widehat{\Pr}[p^*]$ for the probability of sampling this point, 
i.e., $\Pr[p^* = x] \leq \widehat{\Pr}[p^*] \leq 2 \Pr[p^* = x]$.
This algorithm uses $\poly(d\cdot\log\Delta)$ bits of space,
and fails with probability at most $1/\poly(\Delta^d)$.
\end{theorem}

\begin{proof}[Proof of Theorem~\ref{thm:impt}]
We first provide an algorithm that samples points $x\in P$
at a given level $i\in \{1,\ldots,L\}$,
which refers to points with $r_x$ value roughly $2^{-i} \wopen$.
We present it as an offline algorithm in Algorithm~\ref{alg:impt_i}, 
and discuss below how to implement it as a streaming algorithm.
The main guarantee about its output is given in Lemma~\ref{lemma:impt_i} below,
whose proof appears in Section~\ref{sec:proof_lemma_impt_i}. 
We remark that the algorithm returns $\NIL$ in case $\sub_i(P)$ is empty, i.e., no point survives the subsampling.
(We will see later that Theorem~\ref{thm:impt} follows
by simply executing a streaming implementation of this algorithm with a random level $i$.)

\begin{lemma}
\label{lemma:impt_i}
Algorithm~\ref{alg:impt_i} returns a random point $p^*\in P \cup \{\NIL\}$
such that
\[
  \forall x \in P_i, \qquad
  \Pr[p^* = x] \geq \Omega\left(\frac{1}{\poly(d\cdot\log\Delta)}\right) \cdot 2^{-i} \cdot \frac{\wopen}{\OPT},
\]
where $P_i := \{ x \in P : 2^{-i} \wopen < r_x \leq 2^{-i + 1} \wopen \}$.
\end{lemma}

  \begin{algorithm}[ht]
    \caption{Importance sampling for single level $i$}
    \label{alg:impt_i}
    \begin{algorithmic}[1]
      \State let $\ell\gets 0.1\cdot 2^{-i} \wopen$,
      let $\hash_i$ be a $(\poly(d), \poly(d))$-hash of $\RR^d$ with diameter bound $\ell$
      \Comment{use \Cref{thm:exist_decomp}}
      \label{line:impt_i-rounding}
      \State subsample $P$ with rate $2^{-i}$ and let $\sub_i(P)$ denote the subsampled subset of $P$ \label{line:impt_i-subsample}
      \State  sample uniformly $a \in \hash_i(\sub_i(P))$,
      then uniformly $p^* \in \hash_i^{-1}(a) \cap \sub_i(P)$
\label{line:sample}
      \State return $p^*$ \Comment{if such $p^*$ does not exist (e.g., if $\sub_i(P) = \emptyset$), return $\NIL$.}
      \label{line:impt_i-return}
    \end{algorithmic}
  \end{algorithm}

\paragraph{Streaming Implementation of Algorithm~\ref{alg:impt_i}.}
Line~\ref{line:impt_i-rounding} uses \Cref{thm:exist_decomp}
to get a data-oblivious function $\hash_i$, hence this step
can be executed before the algorithm starts to process the stream.

\subparagraph{Subsampling in Dynamic Streams.}
Line~\ref{line:impt_i-subsample} performs subsampling with rate $2^{-i}$,
that is, each point in $P$ is independently sampled with probability $2^{-i}$.
If the stream is insertion-only we just sample each newly added point
independently with probability $2^{-i}$.
However, in dynamic streams we need consistency between insertions and deletions of a point,
and we thus apply a random hash function $h: [\Delta]^d \rightarrow \{0, 1\}$
such that for every point $p$ we have $\Pr[h(p) = 1] = 2^{-i}$.
We draw this hash function at the beginning of the stream,
and then for each insertion/deletion of a point $p$,
we evaluate $h(p)$ to determine whether $p$ is subsampled.
    
    In the analysis, we assume that these subsampling events are independent for all points,
    i.e., that $h$ is fully random. To deal with the fact that storing such a hash function takes $\Delta^d$ bits, we use Nisan's pseudorandom
    generator (PRG)~\cite{Nisan92} that has the following guarantee:
    For any parameters $R$ and $S$, given a seed of $\Omega(S\cdot \log R)$ truly random bits,
    the PRG generates $R$ bits that cannot be distinguished from truly random bits
    by any algorithm running in space $S$. Naturally, we use this PRG
    with $S$ being the space cost of our algorithm and $R = i\cdot \Delta^d$ (which is
    the number of independent fair coin flips needed to generate $h$).

\subparagraph{Two-level Uniform Distinct Sampling.}
To implement the final sampling step of Algorithm~\ref{alg:impt_i}
(in line~\ref{line:sample}),
we present in Lemma~\ref{lemma:2dl0} a two-level $\ell_0$-sampler,
which is more convenient to describe as sampling from a frequency matrix.
The proof of this lemma, provided in Appendix~\ref{sec:proof_lemma_2dl0},
is an extension of a standard $\ell_0$-sampler (from a frequency vector); see e.g.~\cite{CormodeF14}.
We are not aware of such a sampler in the literature,
although similar extensions were devised before, e.g.~in~\cite{FIS08}.\footnote{The notion of \emph{$\ell_p$-sampling with meta-data}, 
  which was recently introduced in~\cite{CJLW22},
  sounds related but is quite different,
  as each index $i$ arrives with an associated value $\lambda_i$;
  in fact, their approach builds on Precision Sampling~\cite{DBLP:conf/focs/AndoniKO11}
  and is applicable only for $p>0$.
The use of a two-level structure and its representation as a matrix
  were introduced in~\cite{CW05,JW09} as cascaded aggregates/norms, 
  however their algorithms estimate these norms,
  not sampling an index by the norm. 
A sampler for cascaded $\ell_{p,2}$-norm was designed in~\cite{BKKS20},
  building on properties of the Gaussian distribution and $\ell_p$-samplers. 
}

\begin{restatable}[Two-Level $\ell_0$-Sampler]{lemma}{lemmaTwoLevelEllZero}
\label{lemma:2dl0}
There is a randomized algorithm,
that given as input a matrix $M\in\RR^{m \times n}$,
with $m\le n$ and integer entries bounded by $\poly(n)$,
that is presented as a stream of additive entry-wise updates,
returns an index-pair $(i,j)$ of $M$, 
where $i$ is chosen uniformly at random (u.a.r.) from the non-zero rows,
and then $j$ is chosen u.a.r.\ from the non-zero columns in that row $i$.
The algorithm uses space $\poly(\log n)$, 
fails with probability at most $1 / \poly(n)$,
and can further report the corresponding row-sum $\sum_{j'} M_{i,j'}$. 
\end{restatable}

It is straightforward to implement
line~\ref{line:sample} of Algorithm~\ref{alg:impt_i} using this sampler.
Simply convert the updates to $P$, on the fly, 
into updates to a frequency matrix $M$,
whose rows correspond to all hash buckets (images of $\hash_i$)
and columns correspond to all grid points ($[\Delta]^d$).
This is clearly a huge matrix, but it is not maintained explicitly.
The reported row-sum $\sum_{j'} M_{i, j'}$ corresponds to
the number of points in $P$ that are hashed (mapped)
to the bucket returned by the sampler.
Hence, we can implement Algorithm~\ref{alg:impt_i} 
in one pass over a dynamic geometric stream. 
The success probability depends on the two-level $\ell_0$-sampler,
and is thus $1-1/\poly(n) \geq 1-1/\poly(\Delta^d)$,
and assuming success,
the output distribution is as described in Lemma~\ref{lemma:impt_i}.

We can now complete the proof of Theorem~\ref{thm:impt}.
The algorithm draws uniformly at random a level $i^*\in\{1,\ldots,L\}$
and executes Algorithm~\ref{alg:impt_i} for this level $i^*$.
Now consider a point $x\in P$,
and let $j$ be the level for which $x \in P_{j}$,
i.e., $2^{-j} \wopen < r_x \le 2^{-j+1} \wopen$.
Then by Lemma~\ref{lemma:impt_i}, the probability to sample this point $x$ is
\begin{equation}  \label{eq:RandomLevel2}
  \Pr[p^* = x] \geq
  \Pr[i^* = j]\cdot \frac{1}{\poly(d\cdot\log\Delta)} \cdot 2^{-j}  \cdot \frac{\wopen}{\OPT}
  \geq \frac{1}{L}\cdot \frac{1}{\poly(d\cdot\log\Delta)} \cdot \Omega\left( \frac{r_x}{\sum_{y \in P}r_y} \right) .
\end{equation}
Recall that the algorithm needs to report also an estimate $\widehat{\Pr}[p^*]$
for the probability of sampling the specific point $p^*$ that is reported. 
Given the randomly chosen level $i^*$ (which might differ from the level $j$ of $x$),
for the algorithm to pick $x$,
it must first subsample $x$, which happens with probability $2^{-i^*}$,
then pick the bucket of $x$ under $\hash_{i^*}$,
while there are $|\hash_{i^*}(\sub_{i^*}(P))|$ non-empty buckets, 
and finally, it has to pick this point $x$ from its bucket,
which contains $|\hash_{i^*}^{-1}(\hash_{i^*}(x)) \cap \sub_{i^*}(P)|$ subsampled points. Thus, 
\[
  \Pr\big[ p^*=x \mid i^* \big]
  = \frac{2^{-i^*}}{|\hash_{i^*}(\sub_{i^*}(P))| \cdot |\hash_{i^*}^{-1}(\hash_{i^*}(x)) \cap \sub_{i^*}(P)|}\ .
\]
Furthermore, the algorithm can accurately estimate all these quantities;
indeed, the bucket size is known from the two-level $\ell_0$-sampler
(recall that Lemma~\ref{lemma:2dl0} reports also the corresponding row-sum),
and to estimate the number of non-empty buckets the algorithm can run
in parallel a standard streaming algorithm
for counting distinct elements (see e.g.~\cite{DBLP:conf/pods/KaneNW10}).
\end{proof}

\subsection{Proof of Lemma~\ref{lemma:impt_i}}
\label{sec:proof_lemma_impt_i}

\paragraph{Subsampling.}
The first step is to subsample every point in $P$ independently with probability $1 / 2^i$.
For every subset $S \subseteq P$, let $\sub_i(S) \subseteq S$ be the random subset after the subsampling.
The following describes several standard facts about the subsampling.

\begin{fact}
    \label{fact:card_cluster}
    $\forall S \subseteq \mathbb{R}^d$ and $t \geq 2$,
    the following holds.
    \begin{itemize}
        \item If $|S| \geq 2^i$, then 
        $\Pr[|\sub_i(S)| \geq t \cdot |S| \cdot 2^{-i}] \leq  \exp(-\Theta(t))$.
        \item If $|S| \leq 2^i$, then
        $\Pr[|\sub_i(S)| \geq t] \leq \exp(-\Theta(t))$.
    \end{itemize}
\end{fact}
\begin{proof}
    For every $u \in S$, let $X_u \in \{0, 1\}$ be the indicator random variable
    such that $X_u = 1 $ if and only if $u \in \sub_i(S)$, so $\Pr[X_u = 1] = 2^{-i}$ for every $u \in S$.
    Then $|\sub_i(S)| = \sum_{u \in S}{X_u}$,
    and $\E[|\sub_i(S)|] = |S| \cdot 2^{-i}$.
    Let $\mu := \E[\sub_i(S)]$.

    \begin{itemize}
        \item If $|S| \geq 2^i$, then $\mu \geq 1$. By Chernoff bound,
        \[
            \begin{aligned}
            \Pr[|\sub_i(S)| \geq t \cdot |S| \cdot 2^{-i}] 
            &= \Pr[|\sub_i(S)| - \mu \geq (t - 1) \cdot \mu] \\
            & \leq \exp(-\Theta(t)\cdot \mu)
            \leq \exp(-\Theta(t)).
            \end{aligned}
        \]
        \item If $|S| \leq 2^i$, then $\mu \leq 1$. By Chernoff bound,
        \[
            \begin{aligned}
            \Pr[|\sub_i(S)| \geq t] 
            &= \Pr[|\sub_i(S)| - \mu \geq (t/ \mu - 1) \cdot \mu] \\
            & \leq \exp(-\Theta(t)).
            \end{aligned}
        \]
\end{itemize}
\end{proof}

\paragraph{Geometric Hashing.}
We consider bounded consistent hashing schemes that do not map small ``clusters'' 
of points into too many buckets.
We restate the definition below,
and we prove the existence of such hashing schemes, with different tradeoff for the parameters,
in \Cref{thm:exist_decomp,thm:EucSparsePartition}.
As mentioned in \Cref{sec:intro},
Jia et al.~\cite{DBLP:conf/stoc/JiaLNRS05}
introduced (and further studied in~\cite{DBLP:conf/icalp/Filtser20}) an essentially equivalent notion called ``sparse partitions'',
although we further require that evaluating the hash function at a point is space-efficient.

\ConsistentHashing*

\paragraph{Sampling on $\hash_i(\sub_i(P))$.}

Suppose we find a $(\Gamma, \Lambda)$-hash $\hash_i$ with diameter bound $\ell = 2^{-i} \wopen / 10$ such that $\Gamma = \poly(d)$ and $\Lambda = \poly(d)$ (by using \Cref{thm:exist_decomp}),
and define $\epsilon := \ell / \Gamma$ which is the magnitude of the consistency guarantee.
Then $\hash_i(\sub_i(P))$ essentially maps points in $\sub_i(P)$ into buckets,
and our plan is to sample from these buckets.
Next, we wish to upper bound $|\hash_i(\sub_i(P))|$, which is the support of sampling, in terms of $\OPT$.
Since the guarantee on $\hash_i$ in \Cref{def:decomp} 
is about clusters/subsets, we need to first define a clustering of the point set
(Lemma~\ref{lemma:extend_MP_cluster}) such that the number of points
in each cluster $C$ is upper bounded by $O(\wopen/\Diam(C))$.
Then, in Lemma~\ref{lemma:fsubi_ub},
we use the guarantee of the geometric hashing on
the clusters resulting from Lemma~\ref{lemma:extend_MP_cluster}
to bound $|\hash_i(\sub_i(P))|$.
In general, there are two types of clusters according to the diameter:
i) ``small'' with diameter at most $\epsilon = \ell / \Gamma$, for which we use the consistency guarantee of our geometric hashing, i.e., the second point
of Definition~\ref{def:decomp} (note that for small clusters, $\wopen/\Diam(C)$ is not a useful bound on $|C|$),
and ii) ``large'', for which $O(\wopen / \Diam(C))$ is not too large
and the subsampling leaves only $\poly(d\cdot\log\Delta)$ points for each ``large'' cluster with high probability.
(In the lemma below, if $\Diam(C) = 0$, the bound $O(\wopen / \Diam(C))$ is defined to be $\Delta^d$.)

\begin{lemma}[Extended MP-clustering]
    \label{lemma:extend_MP_cluster}
    There exists a partition $\mathcal{C}$ of $P$
    such that $\wopen \cdot |\mathcal{C}|  \leq O(d\cdot\log\Delta) \cdot \OPT$
    and for every $C \in \mathcal{C}$, $|C| \leq O(\wopen / \Diam(C))$.
\end{lemma}
\begin{proof}
    The following algorithm from~\cite{MettuP03}, called MP algorithm, finds a $3$-approximation for UFL.
    \begin{enumerate}
        \item List $P$ in non-decreasing order of $r_{p}$.
        \item Examine $p \in P$ in order, and
        if there is no open facility in $B(p, 2r_p)$, then
        open the facility at $p$.
    \end{enumerate}

    Denote the set of the facilities opened by MP algorithm as $\FMP$. We use the following steps to construct a partition of $P$.
    \begin{enumerate}
        \item For every $p \in P$, assign it to the nearest
        point in $\FMP$.
        \item For every $p \in \FMP$, let $C(p) \subseteq P$
        be the set of points that are assigned to $p$.
        \item For every $p \in \FMP$ and every $j = 1,\dots, L$,
        let $C(p)_j := C(P) \cap P_j$.
        \item \label{step:divide} For every $p \in \FMP$ and every $j$,
        arbitrarily divide $C(p)_j$ into subsets of
        size $2^j$, possibly with a unique subset that has size $< 2^j$,
        and include these subsets into $\mathcal{C}$.
    \end{enumerate}
    
    Now we show that $\mathcal{C}$ is the collection satisfying Lemma~\ref{lemma:extend_MP_cluster}.
    Clearly, $\mathcal{C}$ covers $P$, since every point in $p$ is assigned to some point in $\FMP$.

    To upper bound the number of points in each cluster, suppose $C \in \mathcal{C}$ is created by dividing $C(p)_j$ for some $p$ and $j$.
    Then $\forall q \in C$, $p \in B(q, 2r_q)$.
    To see this, suppose for the contrary that
    $p \notin B(q, 2r_q)$. By the construction of $C(p)$,
    $p \in \FMP$ is the closest to $q$,
    hence $p \notin B(q, 2r_q)$ implies that no point in $\FMP$
    belongs to $B(q, 2r_q)$.
    However, by the MP algorithm, this means $q$ should have been added to $\FMP$,
    which is a contradiction.
Hence for every $q_1, q_2 \in C \subseteq P_j$,
        \[\dist(q_1, q_2) \leq \dist(q_1, p) + \dist(p, q_2) \leq O(2^{-j} \wopen) + O(2^{-j} \wopen)
        = O(2^{-j} \wopen),\]
    which implies that $\Diam(C) = O(2^{-j} \wopen)$.
    By Step~\ref{step:divide} of the construction, we know that $|C| \leq 2^j$. Therefore, when $\Diam(C) > 0$, $|C| \leq O(\wopen / \Diam(C))$ holds.

    Finally, we bound the number of clusters. For every $p \in \FMP$ and every $j$, the number of subsets that 
    we obtain from the division is at most $\left\lceil|C(p)_j| / 2^j\right\rceil \le |C(p)_j| / 2^j + 1$.
    Summing over $p$ and $j$, 
    \[
        \begin{aligned}
        \wopen \cdot |\mathcal{C}|
        \leq \wopen \cdot \sum_{p \in \FMP}{\sum_{j=1}^L}{1 + |C(p)_j| / 2^j}
        &\leq O(d\cdot\log\Delta) \cdot |\FMP| \cdot \wopen +
            O(1) \cdot \sum_{p \in \FMP}{\sum_{j}}{\sum_{q \in C(p)_j}{r_q}} \\
        &\leq O(d\cdot\log\Delta) \cdot |\FMP| \cdot \wopen + O(\OPT) \\
        &\leq O(d\cdot\log\Delta) \cdot \OPT),
        \end{aligned}
    \]
    where the second inequality uses that $C(p)_j \subseteq P_j$
    and $\sum_{p \in P}{r_p} = \Theta(1) \cdot \OPT$.
    This completes the proof of Lemma~\ref{lemma:extend_MP_cluster}.
\end{proof}

\begin{lemma}
    \label{lemma:fsubi_ub}
    With probability at least $1 - 1 / \poly(\Delta^d)$,
    $ \wopen  \cdot |\hash_i(\sub_i(P))| \leq \poly(d\cdot\log\Delta) \cdot (\Gamma  + \Lambda) \cdot  \OPT$.
\end{lemma}
\begin{proof}
    Let $\mathcal{C}$ be the collection of subsets guaranteed by Lemma~\ref{lemma:extend_MP_cluster}.
    Then
    \[
        |\hash_i(\sub_i(P))| \leq \sum_{C \in \mathcal{C}}{ |\hash_i(\sub_i(C)| }
        \leq |\mathcal{C}| \cdot \max_{C \in \mathcal{C}} |\hash_i(\sub_i(C))|.
    \]
Since $\wopen \cdot |\mathcal{C}| \leq O(d\cdot\log\Delta) \cdot \OPT$,
    it suffices to prove that with probability at least $1 - 1 / \poly(\Delta^d)$,
    $|\hash_i(\sub_i(C))| \leq \poly (d\cdot\log\Delta) \cdot (\Gamma + \Lambda)$ for every $C \in \mathcal{C}$. 
    
    We wish to bound $|\hash_i(\sub_i(C))|$ for $C$ with ``small'' diameter
    and ``large'' diameter separately.
    For $C$ with $\Diam(C) \leq \epsilon := \Theta(\wopen / (2^i \Gamma))$,
    by Definition~\ref{def:decomp},
    we have $|\hash_i(C)| \leq \Lambda$, which implies that
    \[
        |\hash_i(\sub_i(C))| \leq |\hash_i(C)|
        \leq \Lambda \,.
\]
For $C$ with $\Diam(C) > \epsilon$,
    by Lemma~\ref{lemma:extend_MP_cluster}, we have
    $|C| \leq O(\wopen / \Diam(C)) \leq O(2^i \cdot \Gamma)$.
    By Fact~\ref{fact:card_cluster}, with probability at least $1 - 1 / \poly(\Delta^d)$,
    $|\sub_i(C)| \leq O(d\cdot\log\Delta) \cdot \Gamma$.
    Finally, applying the union bound to all subsets $C \in \mathcal{C}$
    with $\Diam(C) > \epsilon$ concludes the proof of Lemma~\ref{lemma:fsubi_ub}.
\end{proof}

The next lemma states that if the close neighborhood (at distance $r_p$)
of any point $p$ in $P_i$ does not contain too many points after subsampling $P$,
then for any subsampled point $p\in P_i$
there are not too many subsampled points mapped by $\hash_i$ into the same bucket as $p$.

\begin{lemma}
    \label{lemma:inverse_ub}
With probability at least $1 - 1 / \poly(\Delta^d)$,
    for every point $p$ such that $r_p \geq 2^{-i} \wopen$,
    $|\hash_i^{-1}(\hash_i(p)) \cap \sub_i(P)| \leq O(d\cdot\log\Delta)$.
\end{lemma}
\begin{proof}
    Observe that if $r_p \geq 2^{-i} \wopen$, then
    $|\hash_i^{-1}(\hash_i(p)) \cap P| \leq O(2^i)$.
    This is a consequence of having $\ell = 2^{-i} \wopen / 10$, and the fact that
    $|B(p, r_p / 2)\cap P| \leq O(2^i)$ (by Fact~\ref{fact:rp}).
    Hence, Lemma~\ref{lemma:inverse_ub} follows
    by applying Fact~\ref{fact:card_cluster} on the set $\hash_i^{-1}(\hash_i(p)) \cap P$.
\end{proof}

\begin{proof}[Proof of Lemma~\ref{lemma:impt_i}]
    If $P_i = \emptyset$, then Algorithm~\ref{alg:impt_i} always returns an arbitrary point of $P$ or $\NIL$,
    and the guarantee of the lemma trivially holds.
    Now suppose $P_i \neq \emptyset$ and fix $x \in P_i$.
    Let $P' = P \setminus \{x\}$.
    Let $\mathcal{E}$ be the event that the following happens (over the randomness of $\sub_i$).
    \begin{enumerate}
        \item $\wopen \cdot |\hash_i(\sub_i(P'))| \leq \poly(d\cdot\log\Delta) \cdot \OPT$.
        \item for every point $p$ with $r_p \geq 2^{-i} \wopen$,
        $|\hash_i^{-1}(\hash_i(p)) \cap \sub_i(P')| \leq O(d\cdot\log\Delta)$.
    \end{enumerate}
    By Lemma~\ref{lemma:fsubi_ub} (with the parameter $\Gamma = \poly(d)$ and $\Lambda = \poly(d)$ by \Cref{thm:exist_decomp}) and Lemma~\ref{lemma:inverse_ub},
    $\Pr[\mathcal{E}] \geq 1 - 1 / \poly(\Delta^d)$.
    Moreover, $\mathcal{E}$ implies that
    \[
        \wopen \cdot |\hash_i(\sub_i(P))| \leq \wopen \cdot |\hash_i(\sub_i(P'))| + \wopen
        \leq \poly(d\cdot\log\Delta) \cdot \OPT,
    \]
    and for every point $p$ with $r_p \geq 2^{-i} \wopen$,
    \[
        |\hash_i^{-1}(\hash_i(p)) \cap \sub_i(P)|
        \leq |\hash_i^{-1}(\hash_i(p)) \cap \sub_i(P')| + 1
        \leq O(d\cdot\log\Delta).
    \]
    Note that $\mathcal{E}$ is independent to whether or not $x$ survives the subsampling.
    Let $A_x := \hash_i^{-1}(\hash_i(x)) \cap \sub_i(P)$ be the set of points
    that lie in the same bucket of $\hash_i$ as $x$.
    Suppose Algorithm~\ref{alg:impt_i} returns a random point $p^*$.
    Then a point $x \in P_i$ is sampled with probability
    \[
        \begin{aligned}
        \Pr[p^* = x]
        &\geq \Pr[p^* = x \mid \mathcal{E}] \cdot \Pr[\mathcal{E}] \\
        &\geq \Omega(1) \cdot 2^{-i} \cdot \frac{1}{|\hash_i(\sub_i(P))| \cdot |A_x|} \\
        &\geq \Omega(1) \cdot \frac{1}{\poly(d\cdot\log\Delta)} \cdot 2^{-i} \cdot \frac{\wopen}{\OPT}.
        \end{aligned}
    \]
    This completes the proof of Lemma~\ref{lemma:impt_i}.
\end{proof}

     \section{Streaming Algorithms}
\label{sec:streaming}

\subsection{A Two-Pass $O(1)$-Approximation in Dynamic Streams}
\label{sec:two_pass}
\TwoPassAlg*

\begin{proof}[Proof sketch.]
Observe that Theorem~\ref{thm:impt} samples every point $p^* \in P$
with probability proportional to $r_{p^*}$ up to a $\poly(d\cdot\log\Delta)$ factor.
Hence, the estimator $\widehat{Z} = 1 / \widehat{\Pr}[p^*] \cdot r_{p^*}$
has expectation $\E[\widehat{Z}] = \Theta(1) \cdot \sum_{p \in P} r_p$,
noting that $\widehat{\Pr}[p^*]$ returned by Theorem~\ref{thm:impt}
is a $2$-approximation to the true sampling probability $\Pr[p^*]$.
(If the algorithm of Theorem~\ref{thm:impt} returns $\NIL$, we simply set $\widehat{Z} = 0$.)
Moreover, a standard calculation of the variance of an importance sampler
shows that $\Var[\widehat{Z}] \leq \poly(d\cdot\log\Delta) \cdot \E^2[\widehat{Z}]$.
Therefore, averaging over $\poly(d\cdot\log\Delta)$ independent samples of estimator $\widehat{Z}$,
one obtains an $O(1)$-approximate estimator with constant probability
(by a standard application of Chebyshev's inequality).

Hence, the complete algorithm is straightforward: We draw $\poly(d\cdot\log\Delta)$ independent samples $S$
using Theorem~\ref{thm:impt} in the first pass,
and in the second pass, we estimate the $r_p$ value for every sampled $p \in S$ using Fact~\ref{fact:rp_estimation}.
\end{proof}

\subsection{A One-pass $O(1)$-Approximation in Random-order Streams}
\label{sec:random_order}
\RandomOrderAlg*

Recall that in the proof of Theorem~\ref{thm:two_pass},
we use the importance-sampler of Theorem~\ref{thm:impt} in the first pass,
while in the second pass we estimate the $r_p$ value for the sampled points $p$.
We plan to adapt this two-pass algorithm to the random-order setting in a natural way.
Roughly, we run the sampler in the first half of the stream, and we estimate the $r_p$ values using the second half of the stream.
We formalize this in Algorithm~\ref{alg:random_order}.

\begin{algorithm}[ht]
	\caption{Importance sampling in a random-order stream}
	\label{alg:random_order}
	\begin{algorithmic}[1]
		\Procedure{One-Sample}{}
			\State let $J$ denote the first half of the stream,
            and $K$ the second half
			\State obtain a sample $x^J \in J$ using Theorem~\ref{thm:impt}
            on $J$ together with the corresponding estimated probability $\widehat{\Pr}[x^J]$
			\State obtain an estimate $\widehat{r}^K_{x^J}$ of $r^K_{x^J}$ using Fact~\ref{fact:rp_estimation}
			on $K\cup \{x\}$ 
\label{line:rp_estimation-random_order}
			\State return $1 / \widehat{\Pr}[x^J] \cdot \widehat{r}^K_{x^J}$
			\Comment{if $x^J = \NIL$, return $0$}
		\EndProcedure
		\Procedure{Main}{}
			\State $m \gets \poly(d\cdot\log\Delta)$
			\For{$i = 1, \ldots, m$}
				\State run \textsc{One-Sample} to obtain $\widehat{Z}_i$
			\EndFor
			\State return $\widehat{Z} = \frac{1}{m} \cdot \sum_{i}^{m}{\widehat{Z}_i}$
		\EndProcedure
	\end{algorithmic}
\end{algorithm}

\paragraph{Analysis.}
We  give a constant probability upper and lower bounds on the estimator $\widehat{Z}$ returned by 
Algorithm~\ref{alg:random_order} separately.
Let $\OPT^J$ and $\OPT^K$ denote the optimal UFL cost for point sets $J$ and $K$, respectively;
recall that $\OPT$ is the optimal cost for the whole point set $P$.
Moreover, for any point $x$ we define $r^J_x$ and $r^K_x$ similarly by restricting Definition~\ref{def:rp} to sets $J$ and $K$, respectively, with one adjustment: 
When computing $r^K_x$ (resp. $r^J_x$)
for $x\notin K$ (resp. $x \notin J$), we still take $x$ into account.
This is to ensure $r^K_x > 0$ is well defined (since otherwise $B(x, \wopen)\cap K$ may be empty).

Before we proceed, we use the following interpretation of the randomness in the entire algorithm.
\begin{itemize}
	\item We interpret the random permutation of the stream as first
	randomly partitioning $P$ into the two halves $J$ and $K$, and then applying a uniformly random permutation
	on each of $J$ and $K$ independently.
	Let $\AB$ denote the randomness of the partition,
	and let $\PA$ and $\PB$ denote the randomness of the permutations of $J$ and $K$, respectively.
	\item We denote the randomness of the algorithm $\AG$.
	Note that Algorithm~\ref{alg:random_order} uses randomness independent to the random permutation of the stream, so $\AG$ is independent to $\AB, \PA, \PB$.
	Since Fact~\ref{fact:rp_estimation} gives a deterministic estimation
	of the $r_p$ values, the randomness of $\AG$ comes solely from
	the importance-sampler (Theorem~\ref{thm:impt}) running on $J$.
	In our analysis, we condition on the success of all $m$ instances of Theorem~\ref{thm:impt}, which happens with probability $1 - 1 / \poly(\Delta^d)$ by the union bound.
	Therefore, we have that for every point $x^J$
	sampled by Theorem~\ref{thm:impt},
	$\Pr[x^J]\le \widehat{\Pr}[x^J] \leq 2\Pr[x^J] $. (Here, we write the notation $\Pr[p^* = x^J]$ from Theorem~\ref{thm:impt} as $\Pr[x^J]$ for short.)
	
\end{itemize}

\paragraph{Upper Bound.}
We upper bound the expectation of $\widehat{Z}$ by $O(\OPT)$ in Lemma~\ref{lemma:random_order_ub}.
Then a standard application of Markov's inequality implies
\[
	\Pr[\widehat{Z} \leq O(\OPT)] \geq \frac{5}{6}\,.
\]
\begin{lemma}
	\label{lemma:random_order_ub}
	Let $Z = 1 / \widehat{\Pr}[x^J] \cdot \widehat{r}_{x^J}^K$ be an estimate returned by \textsc{One-Sample}.
	Then $\E[ Z ] \leq O(1) \cdot \OPT$.
\end{lemma}
\begin{proof}
	Using that $\widehat{r}^K_{x^J}\le O(r^K_{x^J})$ and
	$\widehat{\Pr}[x] \ge \Pr[x]$ 
    (since we condition on the success of Theorem~\ref{thm:impt}), we bound
    \[
        \begin{aligned}
		\E[Z]
		&= \E_{\AB}\left[ \E_{\AG}[ 1 / \widehat{\Pr}[x^J] \cdot \widehat{r}^K_{x^J} \mid \AB, \PA, \PB] \right] \\
		&\leq O(1)\cdot \E_{\AB}\left[ \E_{\AG}[ 1 / \Pr[x^J] \cdot r^K_{x^J} \mid \AB, \PA, \PB] \right] \\
		&= O(1)\cdot \E_{\AB}\left[ \sum_{x \in J} r_x^K \right] \\
		&\leq O(1)\cdot \E_{\AB}\left[ \sum_{x \in P} r_x^K \right] \\
		&= O(1)\cdot \sum_{x \in P} \E_{\AB}\left[ r_x^K \right]\,.
        \end{aligned}
    \]
	It is sufficient to upper bound $\E_{\AB}[r_x^K]$ by $O(r_x)$ for every $x \in P$,
	since $\sum_{x \in P} r_x = O(\OPT)$ by Fact~\ref{fact:rp}.
	Fix some $x \in P$.
	By Fact~\ref{fact:rp}, we have $|B(x, r_x) \cap P|\ge \wopen/r_x$.
	Let $Y = |(B(x, r_x) \cap K)\cup\{x\}|$ be the number of points in $B(x, r_x)$ that 
	are in the second half of the stream
	($x$ is always included in $Y$).
	Note that $r^K_x \le O( \wopen/Y)$ (possibly less due to points outside $B(x, r_x)$,
	which we do not take into account).
	Thus, by bounding $\E[ \wopen/Y]$ we get the desired result. In the remaining part of the proof,
	the expectation is over the randomness of $\AB$ only.
	
	To this end, for $u\in B(x, r_x) \cap P\setminus\{x\}$ let $Y_u$ be the indicator random variable 
	equal to 1 if and only if $u\in K$. We have $\E[Y_u] = 1/2$ and
	$Y = \sum_{u \in B(x, r_x) \cap P} Y_u$. Thus, by linearity of expectation 
	$\E[Y]\ge |B(x, r_x) \cap P|/2 \ge  \wopen/(2r_x)$.
	Furthermore, variables $Y_u$ are negatively associated,
    so Chernoff bound still applies~\cite{DBLP:journals/rsa/DubhashiR98}, and we have 
    \[
	 \Pr\left[Y \le  \wopen/(4r_x)\right]
	 \le \Pr\left[Y \le 0.5\E[Y]\right]
	 \le \exp(-\Theta(\E[Y]))
	 = \exp(-\Theta( \wopen/r_x))\,.
    \]
	We use this to bound the expectation of $\wopen / Y$ as follows:
    \[\begin{aligned}
	\E[ \wopen/Y] = \sum_{y \ge 1} \Pr[Y = y]\cdot \frac{ \wopen}{y}
	&\le O(r_x) + \Pr\left[Y \le  \wopen/(4r_x)\right]\cdot \sum_{1\le y\le  \wopen/(4r_x)} \frac{\wopen}{y}
	\\
	&\le O(r_x) + \wopen \cdot \exp(-\Theta( \wopen/r_x))\cdot O(\log( \wopen/r_x))
	= O(r_x)\,.
    \end{aligned}\]
	It follows that $\E_{\AB}[r_x^K] = O(r_x)$, which concludes the proof.	
\end{proof}

\paragraph{Lower Bound.}
Consider $Z = 1 / \widehat{\Pr}[x^J] \cdot \widehat{r}_{x^J}^K$ as returned by \textsc{One-Sample}.
Let
\[
	Z' := 1 / \Pr[x^J] \cdot \min\{ r^J_{x^J}, r^K_{x^J} \}\,.
\]
Clearly, $Z \geq Z'/2$ with probability $1$, since by Fact~\ref{fact:rp_estimation}, $ \widehat{r}^K_{x^J}\ge r^K_{x^J} \geq r_{x^J} $
and since we condition on the the event that $ \widehat{\Pr}[x^J] \leq 2\Pr[x^J] $.
Therefore, we focus on showing that $Z'$ is $\Omega(\OPT)$ with high constant probability, and this implies $Z$ is also $\Omega(\OPT)$ with high constant probability.

The first step is to observe that by Definition~\ref{def:rp},
\begin{equation}
	\label{eqn:min_lb}
	\forall x \in P, \qquad
	\min\{ r^J_x, r^K_x \} \geq r_x.
\end{equation}
We now show that $\OPT^J$ is not much smaller
than $\OPT$ with high probability
\begin{lemma}
	\label{lemma:opt_a}
	$\OPT^J \geq \Omega(\OPT)$ with high constant probability (over the randomness of $\PA, \PB, \AB$).
\end{lemma}
\begin{proof}
    If $\OPT \leq O(\wopen)$, then the bound holds with probability $1$,
    since any feasible solution of $J$ must open at least one facility,
    which implies $\OPT^J \geq \wopen \geq \Omega(\OPT)$.

    Now, assume $\OPT \geq \Omega(\wopen)$.
	It holds that $\OPT^J = \Theta(1) \cdot \sum_{x \in J}{r^J_x} \geq \Theta(1) \cdot \sum_{x \in J}{r_x}$ by Fact~\ref{fact:rp} and~\eqref{eqn:min_lb}.
	For $u\in P$, let $X_u$ be the indicator variable such that $X_u = 1$ if and only if $u \in J$;
	by the definition of $J$, we have $\E[X_u] = 0.5$.
	Let $X := \sum_{u \in P}{X_u \cdot r_u / \wopen}$.
	Then $\sum_{x \in J}r_x = \wopen \cdot X$, and $\wopen \cdot \E[X] = \Theta(1) \cdot \OPT$. Observe that $X_u$'s are negatively associated, and we can still apply Chernoff bounds~\cite{DBLP:journals/rsa/DubhashiR98}.
	Therefore,
    \[\begin{aligned}
		\Pr[\OPT^J \leq O(1)\cdot \OPT)]
		&\leq \Pr[\wopen \cdot X - \wopen\cdot \E[X] < - 0.5\cdot\wopen \cdot \E[X]] \\
        &\leq \exp(-\Theta(1) \cdot \E[X]) \\
        &\leq \exp(-\Theta(1) \cdot \OPT / \wopen) \\
		&\leq O(1),
    \end{aligned}\]
	where the first step holds for a small enough constant hidden in $O(1)$
	and the last step follows from the assumption that $\OPT \geq \Omega(\wopen)$.
\end{proof}
Now, we condition on the the event $\mathcal{E}$ that $\OPT^J \geq \Omega(\OPT)$,
which happens with high probability by Lemma~\ref{lemma:opt_a}.
Thus, in the remaining part we fix the randomness of $\AB, \PA, \PB$ such that event $\mathcal{E}$ happens. 
Then by~\eqref{eqn:min_lb} we get
\[
\E_{\AG}[Z'  ] \ge \E_{\AG}[ 1 / \Pr[x^J] \cdot r_{x^J} ]
	= \E_{\AG}[\OPT^J ] = \Omega(\OPT).
\]
By Theorem~\ref{thm:impt},
\[\begin{aligned}
	\E_{\AG}[Z'^2]
	&= \E_{\AG}[ 1 / (\Pr[x^J])^2 \cdot (\min\{ r^J_{x^J} , r^K_{x^J} \})^2] \\
	&= \sum_{x \in J} 1 / \Pr[x^J] \cdot (\min\{ r^J_{x} , r^K_{x} \})^2
	\\
	&\leq \poly(d\cdot\log\Delta)
		 \cdot \OPT^J \cdot \sum_{x \in J} (\min\{ r^J_{x} , r^K_{x} \})^2 / r^J_x \\
	&\leq \poly(d\cdot\log\Delta)
		 \cdot \OPT^J \cdot \sum_{x \in J} r^J_{x}  \\
	&= \poly(d\cdot\log\Delta)
		 \cdot  (\OPT^J)^2  \\
	&\leq \poly(d\cdot\log\Delta) \cdot \OPT^2
\end{aligned}\]
Hence, conditioned on $\mathcal{E}$, \textsc{One-Sample} is a low-variance sampler.
A straightforward application of Chebyshev's inequality
implies that $Z'$ is at least $\Omega(\OPT)$ with probability $5 / 6$.
As $Z \geq Z'$ with probability~1, the same holds for the estimator $\widehat{Z}$ returned by the final \textsc{Main} procedure in Algorithm~\ref{alg:random_order}.
\qed

\subsection{A One-Pass $O(\Gamma)$-Approximation in Arbitrary-Order Streams}
\label{s:1p}

\begin{theorem}
    \label{thm:1p}
    If for some $\Gamma, \Lambda > 0 $ there exists a $(\Gamma, \Lambda)$-hash
    such that the hash value for any point in $\RR^d$ can be evaluated in space $\poly(d)$, then
    there is a one-pass randomized algorithm that
    computes an $O(\Gamma)$-approximation
    to Uniform Facility Location of an input $P\subseteq [\Delta]^d$
    presented as a dynamic geometric stream,
    using $\poly(d \cdot\log\Delta) \cdot O(\Gamma^2\Lambda^2)$ bits of space.
\end{theorem}
\begin{remark}
    \label{remark:ratio}
    We obtain two space-ratio tradeoffs for \Cref{thm:1p}
    by using the hash construction in \Cref{thm:EucSparsePartition} with different parameters.
    Note that \Cref{thm:EucSparsePartition} gives a $(\Gamma, \Lambda = e^{\frac{8d}{\Gamma}} \cdot O(d \log d))$-hash, for any $\Gamma \in [8, 2d]$,
    and the hash value for a point in $\RR^d$ may be evaluated in space $\poly(d)$ (however, the time complexity is exponential in $d$).
    \begin{enumerate}
        \item Plugging in $\Gamma = \Theta(\frac{d}{ \log d})$, we get
        $\Lambda = \poly(d)$; this immediately yields \Cref{thm:1p-intro}.
        \item Alternatively, we assume the algorithm has an advance knowledge of $n$, the number of data points\footnote{
        An $O(1)$-factor estimate of $n$ is sufficient. This assumption can be removed (with a modest increase of space complexity), say,
        by trying all powers of $2$ up to $\Delta^d$, which is an upper bound on $n$.}.
        Then for every data point we apply the Johnson-Lindenstrauss transform~\cite{JL84} with distortion $0.1$ (which can be any fixed small constant) and hence target dimension $O(\log n)$,
        and then feed it into \Cref{thm:1p}.
        Now, the new dimension becomes $d = O( \log n )$, and we pick $\Gamma = O(\eps^{-1})$
        so that $\Lambda = O(n^{\eps/2} \cdot \poly(\log n)) = \tilde{O}(n^{\eps/2})$.
        This implies an $O(\eps^{-1})$-ratio and $\tilde{O}(n^{\eps})$-space tradeoff for \Cref{thm:1p}.
    \end{enumerate}
\end{remark}

\paragraph{One-pass Algorithm Outline.}
Given the successful application of the importance sampling algorithm
in two passes or in one pass over a random-order stream, a natural idea is to attempt to implement this approach in one pass over an arbitrary-order stream as well.
However, recall that a crucial step in the importance sampling is to estimate the $r_p$ value for
some $p$ sampled by the algorithm, and we managed to do so by using Fact~\ref{fact:rp_estimation}
for the random-order and two-pass settings.
However, Fact~\ref{fact:rp_estimation} does not apply to the one-pass arbitrary-order setting,
since it only works in the case when $p$
is given in advance, but in our case $p$ is given only after processing the stream.
In other words, after processing the stream one has to answer an estimate $r_p$ for a \emph{query} point $p$, with no foreknowledge of this point.
Indeed, answering this query within any constant factor $c$ requires $\Omega(n)$ space, even for 1D, which follows by a reduction from the one-way communication problem of INDEX;
see Appendix~\ref{sec:LBs-index} for details.

Hence, we turn to a more structured estimator.
In particular, to estimate $\sum_{p \in P} r_p$ (which gives an $O(1)$-approximation for Uniform Facility Location by Fact~\ref{fact:rp}),
we consider $i = 0, \ldots, L$ and let $W_i := \sum_{p \in P_i} r_p$.
We focus on estimating $W_i$ separately for each $i$.
Recall that $\sub_i(P)$ samples each point in $P$ independently with probability $2^{-i}$.
Since $W_i = \sum_{p \in P_i}{r_p} = \Theta(\wopen \cdot |P_i| / 2^i)$,
the number of points $p \in \sub_i(P)$ with $r_p \in P_i$ is a nearly-unbiased
estimator for $W_i / \wopen$.
Therefore, our plan is to estimate $|\sub_i(P_i)|$.

To this end, we use ideas from Algorithm~\ref{alg:impt_i} which defines an 
importance sampling on level~$i$.
We alternatively interpret it as sampling from $\sub_i(P)$,
and it can be seen that an estimator $Z = 1 / \Pr[x] \cdot I(x \in P_i)$
is an unbiased estimator for $|\sub_i(P_i)|$.
However, the estimator may not be well-concentrated around the expectation
when $W_i$ is too small (compared with $\OPT$).
In case this happens, we argue that the estimator is also very small
in the absolute sense (i.e., compared to $\OPT$ and not relative to the expectation), with high probability.
This turns out to be sufficient for our purpose,
since in such a case $W_i$ contributes very little to $\OPT$ anyway (which means we can ignore it),
and the small value of the estimator actually correctly reflects this.
We formalize this idea in Lemma~\ref{lemma:unbiased_est}.
Recall that $L := d \cdot \log_2\Delta$.

\begin{lemma}
    \label{lemma:unbiased_est}
    Let $\kappa$ be such that $\wopen \cdot \kappa = \Theta(\OPT / L^3)$ and suppose that $\OPT = \Omega(L^4) \cdot \wopen$.
    Then for every $0 \leq k \leq L$,
    each of the following holds with probability at least $1 - 1/\poly(\Delta^d)$.
    \begin{itemize}
        \item If $2^{k-i} \cdot W_k \geq \wopen \cdot \kappa$,
        then $\wopen \cdot |\sub_i(P_k)| = \Theta(2^{k-i}) \cdot W_k$;
        \item Otherwise, $\wopen \cdot |\sub_i(P_k)| \leq O(\OPT / L^2)$.
    \end{itemize}
\end{lemma}
\begin{proof}
    For $u \in P_k$, let $X_u \in \{0, 1\}$ be the indicator random variable
    such that $X_u = 1$ if and only if $u \in \sub_i(P_k)$,
    so $\Pr[X_u = 1] = 2^{-i}$ for every $u \in P_k$.
    Then $|\sub_i(P_k)| = \sum_{u \in P_k}{X_u}$,
    and $\E[|\sub_i(P_k)|] = |P_k| \cdot 2^{-i}$.

    Let $\mu := \E[\sub_i(P_k)]$.
By the definition of $P_k$, every $p \in P_k$ satisfies
    $2^{-k} \wopen < r_p \leq 2^{-k+1} \wopen$.
    Hence $\frac{|P_k|}{2^k} \wopen < W_k = \sum_{p \in P_k} r_p \leq \frac{|P_k|}{2^{k-1}} \wopen$,
    and this implies
    \begin{equation}
        \mu = \Theta(2^{k - i}) \cdot W_k / \wopen. \label{eqn:mu}
    \end{equation}
    \begin{itemize}
        \item If $2^{k - i} \cdot W_k \geq \wopen\cdot \kappa$, then $\mu \geq \Omega(\kappa) \geq \Omega(d\cdot\log\Delta)$ by \eqref{eqn:mu} and $\OPT = \Omega(L^4) \cdot \wopen$.
        We apply the Chernoff bound to get,
        \[
            \Pr[|\sub_i(P_k)| \in (1 \pm 0.5) \cdot \mu]
            \geq 1 - \exp(-\Omega(\kappa)) = 1 - 1 / \poly(\Delta^d).
        \]
        \item Otherwise, $\mu \leq O(\kappa)$,
        so $|P_k| \leq O(\kappa) \cdot 2^i$.
        Using that $O(d\cdot \log \Delta)\cdot \wopen \cdot \kappa = O(\OPT / L^2)$, we apply Fact~\ref{fact:card_cluster} to $P_k$ and get,
        \[
            \Pr[\wopen \cdot |\sub_i(P_k)| \leq O(\OPT / L^2)]
            \geq 1 - 1 / \poly(\Delta^d).
        \]
\end{itemize}
    This concludes the proof of Lemma~\ref{lemma:unbiased_est}.
\end{proof}

Another more severe issue is that the estimator $I(x \in P_i)$ turns out
to be very sensitive, and it seems difficult to estimate its value even
up to a constant factor (say, test if $r_x \geq 2^{-i} \wopen / 100$), in one pass.
To bypass this, we design an approximate tester for $x \in P_i$, and analyze
how such an approximation affects the ratio.

\begin{definition}[Approximate Tester]
    \label{def:approx_tester}
    We say that $\widehat{I}_{ \leq  i} : P \to \{0, 1\}$ is an $\alpha$-approximate tester
    for level $i$ if for all points $p\in P$ it holds that, 
    \begin{itemize}
        \item if $p \in P_{\leq i}$, then $\widehat{I}_{\leq i}(p) = 1$, and
        \item if $\widehat{I}_{\leq i}(p) = 1$, then $p \in P_{\leq i + \alpha}$,
    \end{itemize}
    where $P_{\leq i} = \bigcup_{j \leq i}{P_j}$.
\end{definition}
Note that $P_{\leq i}$ corresponds to points $p$ with $r_p$ value larger than $2^{-i} \wopen$,
and this approximate tester tolerates points with a not too small value of $r_p$.
This is not a problem since for $j \geq i$, $\wopen \cdot |\sub_i(P_j)| \approx W_j / 2^{i-j}$
and taking this into the estimation would result in an $O(2^\alpha)$ times more
contribution of $W_j$ overall (due to summing the geometric sequence over $i\leq j$).

To implement this tester,
we observe that by the definition of the $r_p$ value (Definition~\ref{def:rp}),
if a point $x$ has $r_x \approx 2^{-i} \wopen$, then
a point $y$ with $r_y \ll 2^{-i} \wopen$ cannot appear in the ball $B(x, r_x / 5)$,
since otherwise by Fact~\ref{fact:rp},
$|B(x, r_x / 2) \cap P| \leq O(2^{i})$ while $B(y, r_y) \subseteq B(x, r_x / 5)$
and $|B(y, r_y) \cap P| \gg 2^i$.
Therefore, we can assert that a bucket cannot contain any point $p$ with a large enough $r_p$ value
if the number of points inside is too large
(the buckets are those defined by the $(\Gamma, \Lambda)$-hash $\hash_i$, where we set the diameter bound $\ell_i = 2^{-i} \wopen / 10$).
This also means a bucket either consists of points with large $r_p$ values only,
or none at all.
Using this observation,
for each bucket we maintain a counter for points inside.
When some point is sampled, we also retrieve this counter for its bucket,
and see if it is not too large, in order to test if the bucket consists of points with large enough $r_p$ values.

However, it is possible that some point $y$ with $r_y \ll 2^{-i} \wopen$
lies inside the bucket, but the number of points inside the bucket is still small
as the majority of the large number of points in
$B(y, r_y) \cap P$ actually lie outside of the bucket.
To resolve this issue, we observe that by a similar reasoning as in the previous paragraph,
if a bucket contains a point $x$ with $r_x \approx 2^{-i}\wopen$,
no point $y$ with $r_y \ll 2^{-i} \wopen$ may be \emph{close} to the bucket with $x$.
Hence, we actually need to count the number of points in a slightly enlarged area,
i.e., for a bucket $Q \subset \mathbb{R}^d$ we need to 
count the points in the neighborhood $B(Q, r_y) \cap P$.
Hence, if we maintain the counter for $B(Q, \beta) \cap P$ for some $\beta > 0$
then we can reject $y$'s with $r_y \ll O(\beta)$.

Now, suppose we are to maintain the counter for some $\beta$.
To implement this, whenever we see a data point $x$ arrives,
we should increase the counter for all buckets $Q$ such that $x \in B(Q, \beta)$.
However, this becomes challenging in the streaming setting,
primarily due to the fact that the number of buckets $Q$ such that $x \in B(Q, \beta)$ can be huge,
and more importantly, most of them may be ``fake'', in that they do not contain any data point.
Consequently, these fake nonempty buckets can enlarge the support of the $\ell_0$-sampler significantly,
which makes it difficult to obtain a uniform sample from the non-empty buckets that are not ``fake''.

To tackle this challenge, we make use of the guarantee from the geometric hashing,
that any subset of $\mathbb{R}^d$ with diameter less than $\epsilon_i := \ell_i / \Gamma$
is mapped to $\le \Lambda$ buckets (see Definition~\ref{def:decomp}).
Hence, if we choose $\beta = \epsilon_i / 2$, the effect of enlarging the
buckets by an additive $\beta = \epsilon_i / 2$ is essentially
making every data point $x$ a ball $B(x, \epsilon_i / 2)$, and add the image
of this ball to the buckets/counters.
We show that this is sufficient to implement the tester with $\alpha = \lceil\log_2 \Gamma\rceil$.

We implement these algorithmic steps 
using an extended $\ell_0$-sampler stated in Lemma~\ref{lemma:ext_l0_sampler} (this extension is somewhat simpler than that of the two-level $\ell_0$-sampler in Lemma~\ref{lemma:2dl0}).
We show in Lemma~\ref{lemma:1pbucket} that we can still bound the number of nonempty buckets,
even including the fake ones, by $\poly(d\cdot\log\Delta)\cdot \OPT$,
which is a strengthened version of Lemma~\ref{lemma:fsubi_ub}.
This is needed to have a bounded support for the $\ell_0$-samplers.

\begin{lemma}[$\ell_0$-Sampler with Data Fields]
    \label{lemma:ext_l0_sampler}
    Let $X$ be a length-$n$ frequency vector, with the maximum frequency bounded by $\poly(n)$ and with each index also associated with a data field of 
    at most $\poly(\log n)$ bits.
    Suppose $X$ is presented in a dynamic data stream consisting of
    (positive or negative) updates to the frequency entries of $X$,
    and with updates to the associated data fields.
    Then, there exists a streaming algorithm that for every such $X$,
    returns an index $i$ of $X$, together with its frequency and data field,
    such that $i$ is chosen uniformly at random (u.a.r.) among indices with non-zero frequency or non-zero data field,
    using space $\poly(\log n)$ and with success probability at least $1 - 1 / \poly(n)$.
\end{lemma}
\begin{proof}[Proof sketch.] The lemma follows directly from the guarantees of a standard $\ell_0$-sampler (see Lemma~\ref{lemma:ell_0-sampler}) by only viewing the multiplicity of a point
    summarized by the $\ell_0$-sampler as a pair of frequency and data field.    
\end{proof}

With this, we state the complete algorithm in three parts: Algorithm~\ref{alg:1pinit} describes the initialization procedure before processing the stream and
Algorithm~\ref{alg:1pupdate} outlines the procedure to insert a new point
(one can remove a point in the same way, just by decreasing instead of increasing the
frequencies and counters).
After the streams ends, we call Algorithm~\ref{alg:1pquery} that returns an
estimate $\widehat{Z}$ for the Uniform Facility Location cost.

We note that $\{C_a^{(t)}\}_t$'s are independent counters
for the points lying in the enlarged buckets.
We show that each $\{C_a^{(t)}\}_t$ has a suitable expectation 
to yield an approximate tester, and we use $T$ of them
to implement the median trick to boost the success probability of an accurate estimation.
A few more remarks regarding our algorithm are in order:
\begin{itemize}
	\item In line~\ref{line:hash} of Algorithm~\ref{alg:1pinit} we need to initialize $T$ fully random
		hash functions $h^{(t)}_i : [\Delta]^d \to \{0, 1\}$. Each of them would require $[\Delta]^d$ bits to store,
		however, we use Nisan's pseudorandom generator (PRG)~\cite{Nisan92}
		for space bounded computation, in the same way as described in Section~\ref{sec:mp_sample}.
\item In Algorithm~\ref{alg:1pquery}, we made one simplification: 
	In line~\ref{line:counter_else} we need $|\calL^{(j)}_i|$, the size of the support of $\calL^{(j)}_i$, which we cannot obtain exactly but an $O(1)$-approximation 
	can be computed in streaming (see e.g.~\cite{DBLP:conf/pods/KaneNW10}) and
	is sufficient as it introduces only a small constant factor error in the analysis.
\end{itemize}
\begin{algorithm}[ht]
	\caption{One-pass streaming algorithm: initialization}
	\label{alg:1pinit}
	\begin{algorithmic}[1]
		\For{$i = 0, \ldots, L$}
			\State $\hash_i: \mathbb{R}^d \rightarrow \mathbb{R}^d$ be a $(\Gamma, \Lambda)$-hash 
			with diameter bound $\ell_i = 2^{-i} \wopen / 10$
            \Comment{define $\epsilon_i := \ell / \Gamma$}
			\State set $m \gets \poly(d\cdot\log\Delta) \cdot O(\Gamma^2 \Lambda^2)$
			\State set $T \gets O(d\cdot\log\Delta)$ \State for $0 \leq t \leq T$, let $h^{(t)}_i : [\Delta]^d \to \{0, 1\}$ be a fully random hash function such that
            \begin{align*}
                \forall x \in [\Delta]^d,\quad \Pr[h_i^{(t)}(x) = 1] = 2^{-i}
            \end{align*}
            \label{line:hash}
            \State for $0 \leq t \leq T$, $ S \subseteq [\Delta]^d$,
            define $\sub_i^{(t)}(S) \gets \{ x \in S : h^{(t)}_i(x) = 1 \}$
            \State let $\sub_i \gets \sub_i^{(0)}$ \State initialize $m$ extended $\ell_0$-samplers of Lemma~\ref{lemma:ext_l0_sampler} on domain $\hash_i([\Delta]^d)$
			and with $T$ counters $\{C^{(t)}\}_{t=1}^{T}$ as the data field of each element; we denote them  $\calL^{(j)}_i$ 
			for $j = 1, \dots, m$.
		\EndFor
	\end{algorithmic}
\end{algorithm}

\begin{algorithm}[ht]
	\caption{One-pass streaming algorithm: insert $p \in [\Delta]^d$}
	\label{alg:1pupdate}
	\begin{algorithmic}[1]
		\For{$i = 0, \ldots, L$}
				\label{line:subsampling}
                \For{$j = 1, \ldots, m$}
	                \If{$p$ survives subsampling $\sub_i$, i.e., $h_i^{(0)}(p) = 1$} \Comment{$\sub_i$ subsamples with rate $2^{-i}$}
                    	\State increase the frequency of $\hash_i(p)$ in $\calL^{(j)}_i$ \EndIf
\For{$u \in \hash_i(B(p, \epsilon_i / 2))$}
                        \State for $1 \leq t \leq T$, increase the counter $C^{(t)}$
                        associated with $u$ in $\calL_i^{(j)}$ by $h_i^{(t)}(p)$
                        \label{line:counter}
                        \State \Comment{we increase $C^{(t)}$ by 1 iff $h_i^{(t)}(p) = 1$}
                    \EndFor
                \EndFor
		\EndFor
	\end{algorithmic}
\end{algorithm}

\begin{algorithm}[ht]
	\caption{One-pass streaming algorithm: query procedure}
	\label{alg:1pquery}
	\begin{algorithmic}[1]
		\For{$i = 0, \ldots, L$} 
			\For{$j = 1, \ldots, m$} \Comment{recall that $m = \poly(d\cdot\log\Delta)$}
				\State query $\calL^{(j)}_i$,
                let $a$ be the sampled non-empty bucket of $\hash_i$, and let $n_a$ and $\{C_a^{(t)}\}_{t=1}^{T}$ be the frequency and counters of $a$ returned by $\calL^{(j)}_i$, respectively
                \label{line:sample_a}
                \State let $\widehat{C}_a$ be the median of
                $\{C_a^{(1)}, \ldots, C_a^{(T)} \}$ \Comment{recall that $T = O(d\cdot\log\Delta)$}
				\If{$n_a = 0$ or $\widehat{C}_a > c$ for some constant $c$ (to be determined from the analysis)}
				\label{line:counter_if}
					\State let $\widehat{Z}_i^{(j)} \gets 0$
				\Else
					\State let $\widehat{Z}_i^{(j)} \gets |\calL^{(j)}_i| \cdot n_a \cdot \wopen$
					\Comment{$|\calL^{(j)}_i|$ is the size of the support of $\calL^{(j)}_i$}
					\label{line:counter_else}
				\EndIf
			\EndFor
			\State let $\widehat{Z}_i \gets \frac{1}{m} \cdot \sum_{j=1}^{m}\widehat{Z}_i^{(j)}$
		\EndFor
		\State return $\sum_{i = 0}^{L} \widehat{Z}_i$
	\end{algorithmic}
\end{algorithm}

\paragraph{Support Size of $\ell_0$-Samplers at Level $i$.}
To analyze our algorithm, we first bound the number of buckets of $\hash_i$ that are either non-empty (i.e., have non-zero frequency) or have a non-zero counter.

\begin{lemma}
	\label{lemma:1pbucket}
	For every level $i$ and any $j = 1, \dots, m$, with probability at least $1 - 1 / \poly(\Delta^d)$,
	$\wopen \cdot |\calL^{(j)}_i| \leq \poly(d\cdot\log\Delta)  \cdot O(\Gamma \Lambda) \cdot \OPT$.
\end{lemma}
\begin{proof}
	Fix $t = 0, \dots, T$.
    Define an extended dataset $\widetilde{P}^{(t)} := \bigcup_{p \in \sub_i^{(t)}(P)}{B(p, \epsilon_i / 2)} \subseteq \mathbb{R}^d$ consisting of neighborhoods of points subsampled by $\sub_i^{(t)}$.    
    Apply Lemma~\ref{lemma:extend_MP_cluster} on $P$
    and obtain a partition $\mathcal{C}$ of $P$.
    It suffices to upper bound
    $|\hash_i(\bigcup_{C \in \mathcal{C}} B(\sub_i^{(t)}(C), \epsilon_i / 2))|$,
    since $\widetilde{P}^{(t)} \subseteq \bigcup_{C \in \mathcal{C}} B(\sub_i^{(t)}(C), \epsilon_i / 2)$.
    We show this upper bound for a fixed $t$; 
    the final upper bound on $|\mathcal{L}^{(j)}_i|$ will be only $T = O(d\cdot \log \Delta)$ times larger.
    
    Examine each $C \in \mathcal{C}$.
    If $\Diam(C) \leq \epsilon_i / 2$,
    then $\Diam(B(C, \epsilon_i / 2)) \leq  \epsilon_i$.
    Hence, by the guarantee of Definition~\ref{def:decomp},
    $|\hash_i(B(\sub_i(C), \epsilon_i / 2))| \leq \Lambda$.
    Otherwise, $\Diam(C) > \epsilon_i / 2$.
    In this case, we use the guarantee of $\mathcal{C}$
    that
    $|C| \leq O(\wopen / \Diam(C))$,
    so $|C| \leq O(\epsilon_i^{-1}\wopen) \leq O(2^i \cdot \Gamma)$.
    By Fact~\ref{fact:card_cluster},
    with probability at least $1 - 1 / \poly(\Delta^d)$,
    $|\sub_i^{(t)}(C)| \leq \poly(d\cdot\log\Delta) \cdot O(\Gamma)$.
    We apply the union bound and condition on the success of this for every $C \in \mathcal{C}$.
    The guarantee of the geometric hashing implies that for every $p\in \sub_i^{(t)}(C)$,
    $|\hash_i(B(p, \epsilon_i / 2))| \le \Lambda$.
    We thus have
    \[
        |\hash_i(B(\sub_i^{(t)}(C), \epsilon_i / 2))|
        \leq \Lambda \cdot |\sub_i^{(t)}(C)|
        \leq \poly(d\cdot\log\Delta) \cdot O(\Gamma \Lambda).
    \]

    Therefore, 
    \begin{align*}
        |\mathcal{L}^{(j)}_i|
        \leq \sum_{t = 1}^{T}\left|\hash_i\left( \bigcup_{C \in \mathcal{C}} B(\sub_i^{(t)}(C), \epsilon_i / 2) \right)\right|
        &\leq T\cdot \poly(d\cdot\log\Delta)  \cdot |\mathcal{C}| \cdot O(\Gamma \Lambda) \\
        &\leq \poly(d\cdot\log\Delta) \cdot O(\Gamma \Lambda)  \cdot \OPT / \wopen.
    \end{align*}
    This concludes the proof of Lemma~\ref{lemma:1pbucket}.
\end{proof}

\paragraph{The Accuracy Bound for the Counters.}
Recall that we maintain multiple counters $\{C_a^{(t)}\}$ for the enlarged
buckets, and we now analyze their accuracy.
The fact that we need $T = \poly(d\cdot\log\Delta)$ of them is used to ensure
a success probability $1 - 1 / \poly(\Delta^d)$.

\begin{lemma}
    \label{lemma:counter}
    Let $a \in \mathbb{R}^d$ be a bucket sampled in Algorithm~\ref{alg:1pquery},
    and let $Q = \hash_i^{-1}(a)$ be the point set corresponding to bucket $a$.
    Then with probability at least $1 - 1 / \poly(\Delta^d)$ (over the randomness of $\sub_i^{(t)}$ for $1 \leq t \leq T$),
    \begin{itemize}
        \item if $|B(Q, \epsilon_i / 2) \cap P| / 2^{i} \leq O(1)$ then $\widehat{C}_a \leq O(1)$, and
        \item otherwise $\widehat{C}_a = \Theta(1) \cdot |B(Q, \epsilon_i / 2) \cap P|/2^{i}$.
    \end{itemize}
\end{lemma}
\begin{proof}
    First consider the case of $|B(Q, \epsilon_i / 2) \cap P| \geq \Omega(2^i)$.
    Fix some $1 \leq t \leq T$.
    Observe that 
    $C_a^{(t)} = |\sub_i^{(t)}(B(Q, \epsilon_i/2)\cap P)|$.
    Hence
    $\E[C_a^{(t)}] = |B(Q, \epsilon_i / 2) \cap P| / 2^i $.
    By Markov's inequality,
    \[
        \Pr[C_a^{(t)} \geq \Omega(1) \cdot |B(Q, \epsilon_i / 2) \cap P| / 2^i]
        \leq 0.1,
    \]
    and using the case condition $|B(Q, \epsilon_i / 2) \cap P| \geq \Omega(2^i)$ and the full independence of subsampling points, by Chernoff bound
    \[
        \Pr[C_a^{(t)} \leq 0.5 \cdot |B(Q, \epsilon_i / 2) \cap P| / 2^i]
        \leq 0.1.
    \]
    Therefore,
    \[
        \Pr[C_a^{(t)} = \Theta(1) \cdot |B(Q, \epsilon_i / 2) \cap P| / 2^i]
        \geq 0.8.
    \]
    Since for $1 \leq t \leq T$, $C_a^{(t)}$'s are independent,
    then by the standard median trick (using Chernoff bound),
    we have that $\widehat{C}_a = \Theta(1) |B(Q, \epsilon_i / 2) \cap P| / 2^i$ with high probability.

    The other case of $|B(Q, \epsilon_i / 2) \cap P| \leq O(2^i)$
    can be analyzed similarly (and is simpler since we only need an upper bound).
\end{proof}
\paragraph{Analysis of Level $i$.}
We now consider a level $i = 1, \dots, L$ and bound the estimator $\widehat{Z}_{i}$ on this level.
To this end, in Lemma~\ref{lemma:1pone_i}
we first show that the counter $\widehat{C}_a$ in our algorithm
indeed serves as an adequate approximate tester,
and then we bound the error of the algorithm at level $i$,
particularly the effect of using the approximate tester according to Definition~\ref{def:approx_tester}.
Since $\epsilon_i = \ell_i / \Gamma$,
the approximate tester is off by a factor $\Gamma$ which translates into $O(\Gamma)$
in the final ratio.

\begin{lemma}
    \label{lemma:1pone_i}
    Given that $\OPT = \Omega(L^4)$,
    for every $i$, with probability at least $2/3$
    variable $\widehat{Z}_i$ in Algorithm~\ref{alg:1pquery} satisfies
    \begin{itemize}
        \item $\widehat{Z}_i \leq O(\OPT / L) + \sum_{k \leq i + \alpha} \Theta(2^{k - i}) \cdot W_k$, and
        \item either $W_i  \leq O(\OPT / L^3)$, or $\widehat{Z}_i \geq \Omega(W_i)$,
    \end{itemize}
	where $\alpha := \lceil \log_2 \Gamma \rceil$.
\end{lemma}
\begin{proof}
First, condition on the success of the $\ell_0$-samplers.
    We apply Lemma~\ref{lemma:counter}, 
    Lemma~\ref{lemma:1pbucket}, Lemma~\ref{lemma:unbiased_est},
    and the union bound to get that with probability at least $1 - 1 / \poly(\Delta^d)$, the following statements hold:
    \begin{enumerate}[label=(C\arabic*)]
        \item $\wopen \cdot |\mathcal{L}_i^{(j)}| \leq \poly (d\cdot\log\Delta) \cdot O(\Gamma \Lambda) \cdot \OPT$ for any $j = 1,\dots, m$.\label{cond1}
        \item For every bucket $a \in \hash_i([\Delta]^d)$ sampled in Algorithm~\ref{alg:1pquery},
        let $Q_a = \hash_i^{-1}(a)$,
        then if $|B(Q_a, \epsilon_i / 2) \cap P| \leq O(2^i)$,
        we have $\widehat{C}_a \leq O(1)$; otherwise, $\widehat{C}_a = \Theta(1) \cdot |B(Q, \epsilon_i/2) \cap P| / 2^i$.\label{cond2}
\item Let $\kappa$ be such that $\wopen \cdot \kappa = \Theta(\OPT / L^3)$
        as in Lemma~\ref{lemma:unbiased_est}.
        For every $0 \leq k \leq L$,
        if $2^{k - i} \cdot W_k \geq \kappa \wopen$ then $\wopen \cdot |\sub_i(P_k)| = \Theta(2^{k-i}) \cdot W_k$;
        otherwise $\wopen \cdot |\sub_i(P_k)| \leq O(\OPT / L^2)$.
        \label{cond3}
    \end{enumerate}
    By condition~\ref{cond2}, we know that if for some sampled bucket $a$,
    $\widehat{C}_a \leq c$, then $|B(Q_a, \epsilon_i/2) \cap P| / 2^i \leq O(1)$,
    which implies $|Q_a \cap P| / 2^i \leq O(1)$.
    Now consider the $n_a$ defined in Algorithm~\ref{alg:1pquery}.
    Since $n_a$ counts $\sub_i(Q_a \cap P)$, by Fact~\ref{fact:card_cluster}, 
    $n_a \leq O(d\cdot\log\Delta)$ with probability at least $1 - 1 / \poly(\Delta^d)$
    (over the randomness of $\sub_i = \sub_i^{(0)}$ which is independent of the randomness of $\widehat{C}_a$).
    Hence, we also condition on the success of this event, as follows.
    \begin{enumerate}[label=(C\arabic*)]
        \setcounter{enumi}{3}
        \item For every bucket $a \in \hash_i([\Delta]^d)$ sampled in Algorithm~\ref{alg:1pquery},
        if $\widehat{C}_a \leq c$ then $n_a \leq \poly(d\cdot\log\Delta)$.
        \label{cond4}
    \end{enumerate}
    Therefore, all the above conditions happen with probability at least $1 - 1 / \poly(\Delta^d)$, over the randomness of $\sub_i^{(t)}$'s ($0 \leq j \leq T$).

    \paragraph{Interpreting The If-else Branch as An Approximate Tester.}
    We show that Algorithm~\ref{alg:1pquery} in fact simulates an approximate tester 
    satisfying Definition~\ref{def:approx_tester}. We define this tester $\widehat{I}_{\leq i}$  as follows:
    For every bucket $a$ with non-zero frequency,
    we assign $\widehat{I}_{\leq i}(p) = 1$ for every $p$ that belongs to $a$
    if $a$ would reach line \ref{line:counter_else} (the else branch),
    i.e., if $\widehat{C}_a \le c$
    (we set $\widehat{I}_{\leq i}(p)$ for any $p$ in any non-empty bucket $a$, no matter whether it is sampled or not).
    We then set all other values $\widehat{I}_{\leq i}(p)$ to be zero.
    
    For a non-empty bucket $a$, let $Q_a := \hash_i^{-1}(a)$ be its corresponding points (that are not necessarily in $P$).
    If $a$ reaches line \ref{line:counter_else},
	then every point $p\in Q_a\cap P$ has $r_p \geq \Omega(\epsilon_i) = \Omega(2^{-i} \wopen / \Gamma)$,
    as otherwise $| B(Q_a, \epsilon_i / 2)  \cap P| \geq \Omega(2^i \cdot \Gamma)$,
    which implies $\widehat{C}_a \geq \Omega(\Gamma) > c$
    by condition~\ref{cond2},
    giving us a contradiction (we choose constant $c > 1$ based on the hidden constants in condition~\ref{cond2} so that this implication works).
    
	On the other hand, if bucket $a$ contains at least one point $p$ (i.e., $p \in Q_a$) with $r_p \geq 2^{-i}\wopen$,
	then as $\ell_i = 2^{-i} \wopen / 10$, 
	the bucket $a$ and its close neighborhood contain at most $O(2^i)$ points.
	More precisely, by triangle inequality and the fact that $\Diam(\hash_i^{-1}(a)) \leq \ell_i$,
    we have
    \[
        B(\hash_i^{-1}(a), \epsilon_i / 2)
        \subseteq B(p, \ell_i + \epsilon_i / 2)
        \subseteq B(p, r_p/2),
    \]
    hence by Fact~\ref{fact:rp}
    \[
        |B(\hash_i^{-1}(a), \epsilon_i / 2)  \cap P|
        \leq |B(p, r_p) \cap P|
        \le O(2^i).
    \]
	This implies that $\widehat{C}_a \leq c$ by condition~\ref{cond2},
	so $a$ would reach line~\ref{line:counter_else} if it gets sampled
	(this case determines the constant $c$ in line~\ref{line:counter_if} of Algorithm~\ref{alg:1pquery}).
Therefore, $\widehat{I}_{\leq p}$ is an $\alpha$-approximate tester
    for $\alpha := \lceil \log_2 \Gamma \rceil$.

    \paragraph{Expectation Analysis.}
    Now we analyze the expectation of $\widehat{Z}^{(j)}_i$ for every $j$.
    Observe that in Algorithm~\ref{alg:1pquery},
    $\widehat{Z}^{(j)}_i \neq 0$
    only if the sampled $a$ reaches line~\ref{line:counter_else}
    of Algorithm~\ref{alg:1pquery},
    and by the interpretation of the approximate tester $I_{\leq i}$,
    we have for such $a$,
    \begin{equation}
        n_a = |\{ p \in \sub_i(P) \cap \hash_i^{-1}(a) : I_{\leq i}(p) =1 \}|.
        \label{eqn:na}
    \end{equation}
    Hence, by \eqref{eqn:na}
    \[\begin{aligned}
        \E[\widehat{Z}^{(j)}_i]
        &= \wopen \sum_{a \in \calL_i^{(j)}} \frac{1}{|\calL_i^{(j)}|}
            \cdot |\calL_i^{(j)}|
            \cdot |\{ p \in \sub_i(P) \cap \hash_i^{-1}(a) : I_{\leq i}(p) =1 \}| \\
        &= \wopen \cdot |\{ p \in \sub_i(P) : \widehat{I}_{\leq i}(p) = 1 \}|.
    \end{aligned}\]
By the definition of the approximate tester (Definition~\ref{def:approx_tester}),
    if $p \in P_{\leq i}$ then $\widehat{I}_{\leq i}(p) = 1$,
    which means the bucket that contains $p$ can reach line \ref{line:counter_else}.
    Hence
    \begin{equation}
        \E[\widehat{Z}^{(j)}_i] \geq \wopen\cdot |\sub_i(P_{\leq i})| \geq \wopen\cdot |\sub_i(P_i)|.
        \label{eqn:expz_lb}
    \end{equation}
    Similarly, by the guarantee that if $\widehat{I}_{\leq i}(p) = 1$ then $p \in P_{\leq i + \alpha}$,
    we have
    \begin{equation}
        \E[\widehat{Z}^{(j)}_i]
        \leq \wopen\cdot \sum_{k \leq i + \alpha} |\sub_i(P_k)|.
        \label{eqn:expz_ub}
    \end{equation}
	Since $\widehat{Z}_i$ is the mean of $\widehat{Z}^{(j)}_i$ for $j = 1, \dots, m$,
	inequalities~\eqref{eqn:expz_lb} and~\eqref{eqn:expz_ub} hold for $\widehat{Z}_i$ as well.
	
    \paragraph{Upper Bound.}
    We show that with good probability,
    $\widehat{Z}_i \leq O(\OPT / (d\cdot\log\Delta)) + \sum_{k \leq i + \alpha} \Theta(2^{k - i}) \cdot W_k$.
    Let $K_l := \{ k \leq i + \alpha : 2^{k-i} \cdot W_k \geq \wopen \kappa \}$
    and $K_s := \{ k \leq i + \alpha : 2^{k-i} \cdot W_k < \wopen \kappa \}$.
    Then by condition~\ref{cond3},
    \[\begin{aligned}
        \wopen\cdot \sum_{k \leq i + \alpha} |\sub_i(P_k)|
        &= \wopen\cdot \sum_{k \in K_s} |\sub_i(P_k)| + \wopen\cdot \sum_{k \in K_l} |\sub_i(P_k)| \\
        &\leq |K_s| \cdot O(\OPT / (d\cdot\log\Delta)^2) + \sum_{k \in K_l} \Theta(2^{k-i}) \cdot W_k \\
        &\leq O(\OPT / (d\cdot\log\Delta)) + \sum_{k \leq i + \alpha} \Theta(2^{k-i}) \cdot W_k.
    \end{aligned}\]
    Combining with \eqref{eqn:expz_ub} and Markov's inequality,
    \[\begin{aligned}
        \Pr\left[ \widehat{Z}_i \geq \Omega(\OPT / (d\cdot\log\Delta)) + \sum_{k \leq i + \alpha} \Theta(2^{k - i}) \cdot W_k \right] 
        &\leq \Pr\left[\widehat{Z}_i \geq \Omega(\wopen) \cdot \sum_{k \leq i + \alpha}|\sub_i(P_k)|\right] \\
        &\leq 
        \Pr[\widehat{Z}_i \geq 6 \cdot \E[\widehat{Z}_i]] \leq 1 / 6.
    \end{aligned}\]

    \paragraph{Lower Bound.}
    It suffices to prove that if $W_i \geq \wopen\cdot \kappa$, then
    $\Pr[\widehat{Z}_i \leq O(W_i)] \leq 1 / 6$.
By condition \ref{cond3} and \eqref{eqn:expz_lb},
    $\E[\widehat{Z}^{(j)}_i] \geq \wopen\cdot |\sub_i(P_i)| \geq \Theta(W_i)$,
    thus also $\E[\widehat{Z}^{(j)}_i] \ge \wopen\cdot \kappa = \Theta(\OPT / L^3)$.
    By condition~\ref{cond1}, condition~\ref{cond4}
    and the fact that $\widehat{Z}_i^{(j)} \neq 0$
    only if the sampled $a$ reaches line~\ref{line:counter_else}
    of Algorithm~\ref{alg:1pquery},
with probability $1$,
    \begin{equation}
        \widehat{Z}^{(j)}_i
        \leq \wopen \cdot |\calL_i^{(j)}| \cdot n_a
        \leq \poly(d\cdot\log\Delta) \cdot O(\Gamma \Lambda) \cdot \OPT
        \leq \poly(d\cdot\log\Delta) \cdot O(\Gamma \Lambda) \cdot \E[\widehat{Z}^{(j)}_i].
    \end{equation}
    Therefore, by Hoeffding's inequality and the definition of $m$,
    \[\begin{aligned}
        \Pr[\widehat{Z}_i \leq O(W_i)]
        &\leq \Pr[\widehat{Z}_i \leq 0.5\cdot \E[\widehat{Z}_i]]
        \leq \Pr[\widehat{Z}_i - \E[\widehat{Z}_i] \leq -0.5\cdot \E[\widehat{Z}_i]] \\
        &\leq \exp\left( -\Theta(1) \cdot \frac{m^2 \E^2[\widehat{Z}_i]}{m \cdot \poly(d\cdot\log\Delta) \cdot O(\Gamma^2 \Lambda^2)
        \E^2[\widehat{Z}_i]} \right) \\
        &\leq  1 / 6.
    \end{aligned}\]
\end{proof}

\begin{proof}[Proof of Theorem~\ref{thm:1p}]

We make use of known results in a black-box way
to detect and solve the case when $\OPT \le \poly( d\cdot\log\Delta)\cdot \wopen$.
Namely, we run an $\poly(d\cdot\log\Delta)$-space, $O(d \log^2\Delta)$-approximation for UFL by~\cite{Indyk04} in parallel to our algorithm (Algorithms~\ref{alg:1pinit}, \ref{alg:1pupdate}, \ref{alg:1pquery}),
and denote the resulting estimate by $Y$.

If $Y / \wopen \geq \Omega(\poly(d\cdot\log\Delta))$,
then we know there exists an $O(1)$-approximation using at least
$\poly(d\cdot\log\Delta)$ facilities, which means $\OPT \geq \Omega(\poly(d\cdot\log\Delta))$.
Hence, in this case we obtain the assumption that we need for the analysis and we use the estimate returned by Algorithm~\ref{alg:1pquery}.
Thus, by Lemma~\ref{lemma:1pone_i} for $i = 0, \dots, L$ (with success probability amplified with the median trick, that is, by running sufficiently many independent repetitions and taking the median estimate)
we obtain an estimator $\widehat{Z}$ such that with high probability
\[\begin{aligned}
	\widehat{Z}
	&\leq \sum_{i \in [L]} O(\OPT / L) + \sum_{k \leq i + \alpha} \Theta(2^{k - i}) \cdot W_k \\
	&\leq O(\OPT) + O(2^\alpha) \cdot \sum_{i}W_i \\
	&\leq O(2^\alpha) \OPT,
\end{aligned}\]
and for the lower bound,
\[\begin{aligned}
	Z
	\geq \sum_{i \in [L]} \Omega(W_i) - O(\OPT / L^3) 
	\geq \Omega(\OPT).
\end{aligned}\]
Since $\alpha = \lceil \log_2 \Gamma \rceil$,
we conclude that we obtain an $O(\Gamma)$-approximation in this case.

Otherwise, if $W / \wopen \leq \poly(d\cdot\log\Delta)$,
we know the number of open facilities cannot be more than $\poly(d\cdot\log\Delta)$
in any constant-factor approximate solution.
To deal with this case, we run several instances of the following algorithm for $\kMedian$  in parallel to our algorithm.
Recall that in the $\kMedian$ problem, the goal is to find a subset $C \subseteq \mathbb{R}^d$
with $|C| \leq k$, such that for dataset $P$, the objective $\sum_{x \in P}{\dist(x, C)}$ is minimized.

\begin{theorem}[\cite{DBLP:conf/icml/BravermanFLSY17}]
    \label{thm:kmedian}
    There is an algorithm that for every $0 < \epsilon < 0.5$, integer $\Delta, k, d \geq 1$,
    for any set of points $P \subseteq [\Delta]^d$ presented as a dynamic stream,
    computes a solution $C \subseteq \mathbb{R}^d$ with $|C|\leq k$ to \kMedian on $P$ as well as an estimate $R$ of the objective function (the total connection cost),
    that is $(1 + \epsilon)$-approximate with constant probability, using space $O(k \epsilon^{-2}\poly(d\cdot\log\Delta))$.
\end{theorem}

We run the algorithm of Theorem~\ref{thm:kmedian} (with a constant $\epsilon$) in parallel for every $k = 1,\dots, \poly(d\cdot\log\Delta)$ so that if an optimal solution uses at most $\poly(d\cdot\log\Delta)$ facilities, we find an $O(1)$-approximate solution.

The space bound follows from the fact that the $O(d \log^2\Delta)$-approximation of~\cite{Indyk04} and \Cref{thm:kmedian}
both use at most $\poly(d\cdot\log\Delta)$ space,
and our Algorithms~\ref{alg:1pinit}, \ref{alg:1pupdate}, \ref{alg:1pquery}
use $\poly(d \cdot \log \Delta) \cdot O(\Gamma^2 \Lambda^2)$ space
and that we run no more than $\poly(d\cdot\log\Delta)$ independent instances of every algorithm.
\end{proof}
     \section{Consistent Hashing Bounds}
\label{sec:hash}

In this section, we present two constructions of consistent hashing that achieve different tradeoffs of parameters.
\Cref{thm:EucSparsePartition} gives a near-optimal tradeoff between the $\Gamma$ and $\Lambda$ parameters. While the resulting hash function needs $\poly(d)$ space to evaluate at any point in $\RR^d$, 
the running time for the hash-value evaluation is exponential in $d$.
(Since the hash function is completely oblivious of the actual data, it may be possible to compute it in advance in $\exp(d)$ time and then hardcode it into an oracle
that may be used arbitrarily many times.)
In our second result, \Cref{thm:exist_decomp}, we design a consistent hashing that can be evaluated in $\poly(d)$ space \emph{and} time,
at the cost of a weaker gap bound of $\Gamma = d^{1.5}$ (provided that $\Lambda = \poly(d)$).
We start with the statements and a brief overview of their proofs.
The detailed proofs can be found in \Cref{sec:proof_opt_hash} and \Cref{sec:proof_d15_hash}, respectively.

\begin{restatable}{theorem}{EucSparsePartition}
\label{thm:EucSparsePartition}
    For every $\Gamma \in [8, 2d]$,
	there exists a $(\Gamma, \Lambda)$-hash $\hash : \mathbb{R}^d \to \mathbb{R}^d$
    for $\Lambda = e^{\frac{8d}{\Gamma}} \cdot O(d \log d)$.
    Furthermore, $\hash$ can be described by $O(d^2\cdot \log^2 d)$ bits and one can evaluate $\hash(x)$ for input $x\in \mathbb{R}^d$
    in space $O(d^2 \log^2 d)$.
\end{restatable}
\begin{remark}
	\label{remark:hash_tight}
	According to \cite{DBLP:conf/icalp/Filtser20},
	for every $\Gamma, \Lambda > 0$, any $(\Gamma, \Lambda)$-hash must satisfy 
	$\Lambda>(1+\frac{1}{2\Gamma})^d>e^{\frac{d}{4\Gamma}}$.
	Thus \Cref{thm:EucSparsePartition} is tight up to the basis of the exponent and the $O(d\log d)$ factor.
	Interestingly, the lower bound of \cite{DBLP:conf/icalp/Filtser20} holds even for the case where we require the number of intersecting cluster to be bounded only in expectation and not in worst case. In this case, the $O(d\log d)$ factor from the intersection bound $\Lambda$ could be removed from \Cref{thm:EucSparsePartition}.
\end{remark}

\paragraph{Proof Overview for \Cref{thm:EucSparsePartition}}
The construction follows the ball carving approach used by Andoni and Indyk \cite{AI08} for the construction of their locality sensitive hashing (LSH). This technique was originally developed for general metrics \cite{CKR04}.
Fix diameter $\ell>0$, and let $w=\frac\ell2$.  We describe the hash function $\hash$ for the box $[w,3w]^d$. Later, it will be generalized to the entire space $\R^d$ by discretely shifting the function $\hash$.
We sample points $v_1,v_2,\dots$ from the box $[0,4w]^d$ u.a.r. which we call centers. We hash a point $x$ to first center $v_i$ at distance $w$. Formally,
$\hash(x)=v_i$ such that $\|x-v_i\|_2\le w$ and $\forall i'<i$, $\|x-v_i\|_2> w$. By the triangle inequality, points in the ball $B(x,\alpha\cdot w)$ can be hashed only to centers at distance at most $(1+\alpha)\cdot w$ from $x$. From the other hand, once there is a center $v_{i_x}$ at distance $(1-\alpha)\cdot w$ from $x$, all the previously undermined points in $B(x,\alpha\cdot w)$ will be hashed to $v_{i_x}$. 
See \Cref{fig:ShatteringExample} for illustration.
It follows that $\left|\hash\left(B(x,\alpha\cdot w)\right)\right|$ is bounded by the number of centers sampled in the ball $B(x,(1+\alpha)\cdot w)$, until a center from $B(x,(1-\alpha)\cdot w)$ is chosen. As the centers are sampled u.a.r., this is simply a geometric distribution with parameter equal to $\frac{\text{Vol}\left(B(x,(1-\alpha)\cdot w\right)}{\text{Vol}\left(B(x,(1+\alpha)\cdot w\right)}=(\frac{1-\alpha}{1+\alpha})^{d}\ge(1-\alpha)^{2d}\ge e^{-4\alpha d}$.
Using union bound (on an $\eps$-net) we conclude that w.h.p. all the balls are mapped to at most $e^{-4\alpha d}\cdot O(d\log d)$ hash cells.
In order to describe the hash function $\hash$, we need to store all the  sampled centers, and their order. W.h.p. there will be at most $2^{O(d\log d)}$ different sampled centers until all the points are clustered. We finish the argument by showing that it is enough to use a pseudorandom string \cite{Nisan92} with a seed of size $O(d^2\log d)$ to sample all centers.

\begin{figure}[t]
	\centering
	\includegraphics[width=1\textwidth]{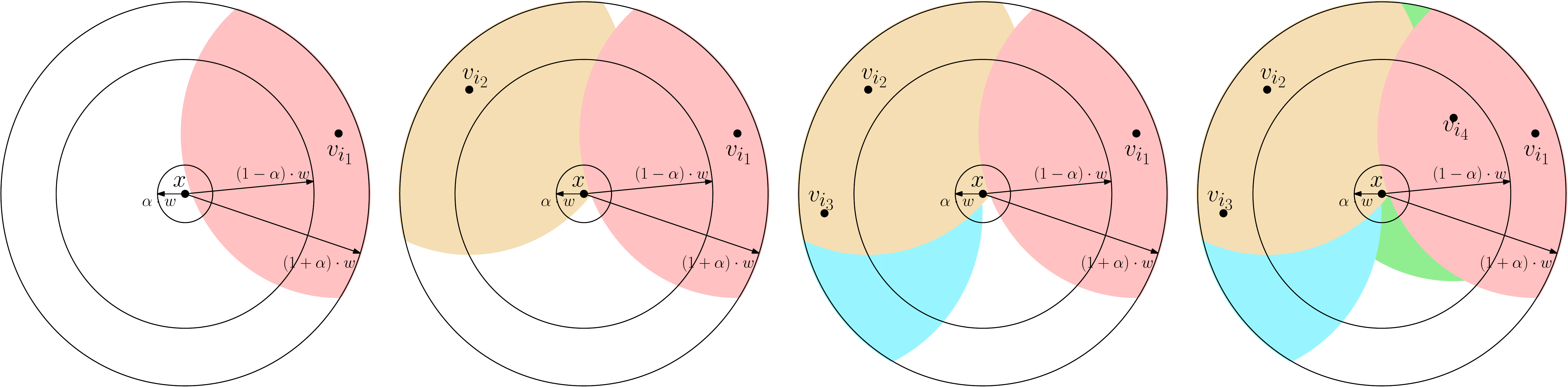}	
	\caption{Illustration of the ball carving process w.r.t. the ball $B$ of radius $\alpha\cdot w$ around $x$. Point in $B$ can be hashed only to centers from $B(x,(1+\alpha)\cdot w)$. Once there is a center $v_{i_x}\in B(x,(1-\alpha)\cdot w)$ is chosen ($v_{i_4}$ in the figure) all the previously undecided vertices in $B$ will be hashed to $v_{i_x}$.
	}
	\label{fig:ShatteringExample}
	
\end{figure}

\begin{restatable}{theorem}{DecompTimeEfficient}
    \label{thm:exist_decomp}
    There exists a $(\Gamma, d+1)$-hash $\hash : \mathbb{R}^d \to \mathbb{R}^d$
    for $\Gamma = \Theta(d \sqrt{d})$.
    Furthermore, $\hash$ is characterized solely by $d$ and the diameter bound,
    and one can evaluate $\hash(x)$ for input $x\in \mathbb{R}^d$
    in time and space $\poly(d)$.
\end{restatable}

  \paragraph{Proof Overview for \Cref{thm:exist_decomp}.}
  This $O(d^{1.5})$-gap construction first partitions the entire $\mathbb{R}^d$  into unit hypercubes.
  Then we apply an identical decomposition/partition on every unit hypercube,
  and the partition of the entire space is the collection of all the parts (across all unit hypercubes).
The partition into hypercubes immediately implies a diameter bound.
  Consider now the partition of an (arbitrary) unit hypercube $H$;
  intuitively, we partition $H$ into $d+1$ \emph{groups} of regions,
such that regions in the same group are at distance $\geq \epsilon := \ell/\Gamma$ of each other (where $\ell = O(\sqrt{d})$ when we start with unit hypercubes).
This way, a subset of diameter $< \epsilon$ can only intersect one region from each of the $d+1$ groups, which ensures consistency.
Roughly, each of the $d+1$ groups, say the $i$-th group,
corresponds to all the $i$-dimensional faces of the hypercube.
However, simply taking these $i$-dimensional faces cannot work,
since they intersect and the minimum distance is $0$.
To make them separated, we employ a sequential process,
where iteration $i=0,1,2,\ldots$ takes as our next group
the $i$-dimensional faces and their close $\ell_\infty$ neighborhoods,
excluding (i.e., removing) points that are
sufficiently close (in the $\ell_\infty$ distance) to $(i-1)$-dimensional faces.
In particular, we crucially use the following geometric fact.
\begin{fact}[See \Cref{lemma:perp_gap}]
  Consider two orthogonal $i$-dimensional subspaces $S_1$, $S_2$
  and denote their intersection by $I$.
  Denoting by $A^{+t}$ the $\ell_\infty$ neighborhood of $A$ of radius $t$,
  define $S_i' := S_i^{+\epsilon} \setminus I^{+2\epsilon}$ for $i=1,2$.
  Then $\dist(S_1', S_2') > \sqrt{2} \epsilon$.
\end{fact}
Our construction iterates over $i = 0, \ldots, d - 1$ and repeatedly uses the above fact.
Denoting by $\mathcal{S}_i$ the union of the $i$-dimensional faces,
the $i$-th group consists of the connected regions of
$\mathcal{S}_i^{+(d - i) \epsilon} \setminus \mathcal{S}_{i - 1}^{+ (d - i + 1) \epsilon}$.
Finally, the $d$-th group consists of a single region corresponding to the hypercube interior after removing the previous groups.
\subsection{A Near-optimal Tradeoff using Small Space: Proof of \Cref{thm:EucSparsePartition}}
\label{sec:proof_opt_hash}

\EucSparsePartition*
\begin{proof}
	Let $\ell > 0$ be the diameter bound of the hashing.
Let $w=\frac12 \ell$, and $\alpha=\frac2\Gamma$, note that $\alpha\in[\frac{1}{d},\frac14]$. We will construct a partition where every ball of radius $\alpha\cdot w=\frac{\ell}{\Gamma}$ intersect at most $e^{4\alpha d}\cdot O(d\log d)=e^{\frac{8 d}{\Gamma}}\cdot O(d\log d)$ clusters.

We define the
following random clustering procedure (introduced by Andoni and Indyk \cite{AI08}). 
The point set
$4w\cdot\mathbb{Z}^{d}
	=\left\{\left(4w\cdot \alpha_1,\dots,4w\cdot \alpha_d\right)\in\R^d \mid \alpha_1,\dots,\alpha_d\in\Z\right\}$
is the $d$-dimensional axis parallel lattice with side length $4w$. $G^{d,w}=\cup_{x\in4w\cdot\mathbb{Z}^{d}}B(x,w)$
is the union of all the balls with radius $w$ and center at a lattice point
$4w\cdot\mathbb{Z}^{d}$. 
Given a point $v\in\mathbb{R}^{d}$,
$v+G^{d,w}$ is simply $G^{d,w}$ shifted by $v$; we will denote
this set by $G_{v}^{d,w}$. 
We create clusters $\mathcal{C}$ of $\mathbb{R}^{d}$ as
follows: 
let $v_{1},v_{2},\dots,v_{l},\dots\in[0,4\cdot w)^{d}$ be an infinite series of random points chosen independently (i.i.d) and uniformly.
For every $v_{i}$, and
$u\in\mathbb{Z}^{d}$ there will be a cluster $C_{v_{i},u}$ containing
all yet unclustered points in the ball $B\left(v_{i}+4w\cdot u,w\right)$. In other words, $C_{v_{i},u}$ contains all the points $p$ at distance less then $w$ from the center $v_{i}+4w\cdot u$, where there is
no $v_{i'}$ with $i'<i$ and $u'\in\mathbb{Z}^{d}$ such that $p\in B\left(v_{i'}+4w\cdot u',w\right)$.
The set of created clusters is denoted $\mathcal{C}_{w}$. 
See \Cref{fig:AI08clustering2} for illustration.

\begin{figure}[t]
	\centering
	\includegraphics[width=.75\textwidth]{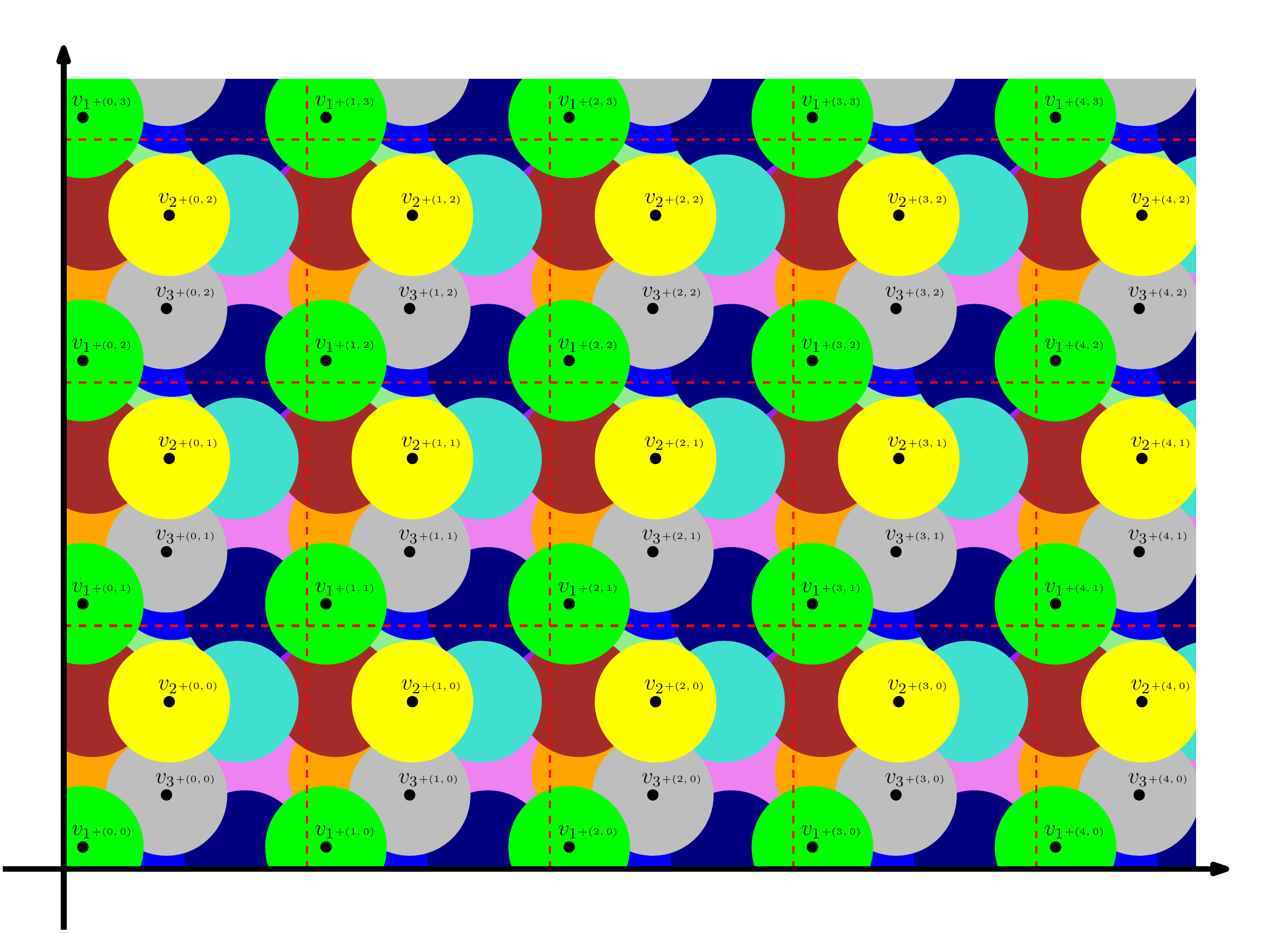}	
\caption{
		Illustration of the clustering process for $w=\frac14$ and $d=2$.
		First the point $v_1\in[0,4\cdot w)^d=[0,1)^2$ is chosen, and we add the clusters $G^{2,\frac14}_{v_1}$ (colored in green).
		The algorithm continues to create clusters in this manner until the entire space is clustered.
		}
\label{fig:AI08clustering2}

\end{figure}

Consider a vector $v\in \R^d$, and denote $B_{\alpha}=B(v,\alpha\cdot w)$.
Let $X_v$ denote the number of centers chosen in the ball  $B_{1+\alpha}=B(v,(1+\alpha)\cdot w)$, until a center chosen in $B_{1-\alpha}=B(v,(1-\alpha)\cdot w)$. Then $X_v$ is a random variable, distributed according to geometric distribution with parameter $p\ge(1-\alpha)^{2d}$.
Indeed, as the points are chosen uniformly i.i.d., the probability of choosing a center in $B_{1-\alpha}$, given that a center from $B_{1+\alpha}$  was chosen is equal to the ratio between the volumes of these two balls:
\[
p=\Pr\left[x\in B_{1-\alpha}\mid x\in B_{1+\alpha}\right]=\frac{\text{Vol}\left(B_{1-\alpha}\right)}{\text{Vol}\left(B_{1+\alpha}\right)}=(\frac{1-\alpha}{1+\alpha})^{d}\ge(1-\alpha)^{2d}\ge e^{-4\alpha d}~.
\]
Note that $X_v$ is an upper bound on the number of clusters intersecting $B_\alpha$.
Indeed, every cluster centered at a point out of $B_{1+\alpha}$ will contain no points from $B_{\alpha}$. 
Furthermore, if $x\in B_{1-\alpha}$, then all the unclustered points in $B_{\alpha}$ will join the cluster centered in $x$, and no future cluster will contain any point from $B_{\alpha}$.

Fix $\delta=\frac{\alpha}{d^{1.5}}$. Consider the box $[0,8\cdot w)^{d}$, and let $N$ be all the points of the axis parallel lattice with side length $\delta\cdot w$ inside the box. Clearly, as $\Gamma \le 2d$, we have that $\alpha=\frac2\Gamma\ge\frac1d$,
which implies $|N|=(\frac{8}{\delta})^{d}=(\frac{8d^{1.5}}{\alpha})^{d}=2^{O(d\log d)}$.
For a net point $v\in N$, the probability that $X_v$ is larger than $\frac 1p(\ln|N|+1)$ is bounded by
\[
\Pr\left[X_{v}>\frac{1}{p}(\ln|N|+1)\right]=(1-p)^{\frac{1}{p}(\ln|N|+1)}\le e^{-\ln|N|-1}=\frac{1}{|N|\cdot e}~.
\]
Denote by $\Psi$ the event  that for every lattice point $v\in N$, it holds that  $X_v\le m=\frac 1p(\ln|N|+1)\le e^{4\alpha d}\cdot O(d\log d)$.
By union bound, $\Pr[\Psi]\ge1-\frac1e$.
Suppose that $\Psi$ occurred, and consider an arbitrary point $u\in [0,8\cdot w)^{d}$. Let $v\in N$ be a point such that $\|u-v\|_2\le \sqrt{d}\cdot\delta\cdot w$.
As $B(u,(\alpha-\sqrt{d}\cdot\delta)\cdot w)\subseteq B(v,\alpha\cdot w)$ and $X_v\le m$, if follows that at most $m$ clusters intersect $B(u,(\alpha-\sqrt{d}\cdot\delta)\cdot w)$.
As this property holds for every point in $[0,8\cdot w)^{d}$, by symmetry, every ball in $\R^d$ of radius $(\alpha-\sqrt{d}\cdot\delta)\cdot w=(1-\frac{1}{ d})\cdot\alpha\cdot w$ intersects at most $m$ clusters.

Denote by $\Phi$ the event that for every point $v\in N$, there is a center chosen in $B(v,(1-\alpha)\cdot w)\supseteq B(v,\frac12\cdot w)$ among the first $s=2^{O(d\log d)}$ centers. The probability of this event not occurring for $v\in N$ is 
\begin{align*}
	\Pr\left[\left\{ v_{1},\dots,v_{s}\right\} \cap B(v,(1-\alpha)\cdot w)=\emptyset\right]
	& \le\left(1-\frac{\text{Vol}\left(B(v,\frac{1}{2}\cdot w)\right)}{\text{Vol}\left([0,4\cdot w]^{d}\right)}\right)^{s}
	=\left(1-\frac{\frac{\pi^{\frac{d}{2}}(\frac{1}{2}\cdot w)^{d}}{\frac{d}{2}!}}{(4\cdot w)^{d}}\right)^{s}\\
	& =\left(1-2^{-O(d\log d)}\right)^{s}\le e^{s\cdot2^{-O(d\log d)}}=2^{-O(d\log d)}~,
\end{align*}
for large enough $s$ and assuming without loss of generality $d$ is even.
By union bound, $\Pr[\Phi]\le|N|\cdot2^{-O(d\log d)}\le 2^{-O(d\log d)}$.

The probability of $\Psi\cap\Phi$ to occur is at least $\frac12$.
If both these events $\Psi$ and $\Phi$ occurred, then 
we created a partition described by $2^{O(d\log d)}$ centers,
where all the clusters have diameter at most $2w=\ell$,
and each ball $B(u,(\alpha-\delta)\cdot w)$ (for any $u\in \R^d$) intersects at most $\frac 1p(\ln|N|+1)\le e^{4\alpha d}\cdot O(d\log d)$ clusters.
Set $\alpha'=(1-\frac{1}{d})^{-1}\cdot\alpha\le(1+\frac{2}{d})\cdot\alpha$, and $\delta'=\frac{\alpha'}{d^{1.5}}$.
Note that $\alpha'-\sqrt{d}\cdot\delta'=\alpha'\cdot(1-\frac{1}{d})=\alpha$.
Then by the same analysis (adapting for $\alpha',\delta'$), with probability $\frac12$ both $\Psi$ and $\Phi$ occurred, and each ball $B(u,\alpha\cdot w)=B(u,(\alpha'-\sqrt{d}\cdot\delta)\cdot w)$ intersects at most $e^{4\alpha'd}\cdot O(d\log d)=e^{4\alpha d}\cdot O(d\log d)$
 clusters.

Consider an algorithm $\mathcal{A}$ that samples $s$ centers in $v_{1},v_{2},\dots,v_{l}\in[0,4\cdot w)^{d}$, and then checks whether both events
$\Psi$ and $\Phi$ occur together. The algorithm can simply go over all the points in $N$ and check them one by one.
Algorithm $\mathcal{A}$ uses a string $R$ of $2^{O(d\log d)}$ random bits (to choose the centers), while the running space is just $S=O(d\log d)$ space for the different counters.
Using Nisan's pseudorandom generator (PRG) \cite{Nisan92}, there is a seed of $O(S\cdot\log R)=O(d^2\cdot \log^2 d)$ truly random bits, that generates a pseudorandom string $R'$ of $2^{O(d\log d)}$ bits such that no algorithm with running space $S$ can distinguish between $R$ and $R'$.
In particular, when we run $\mathcal{A}$ on the pseudorandom string $R'$, it finds that $\Psi$ and $\Phi$ occur together with probability $\frac12$.
We conclude that there is a pseudorandom string $R'$ with seed of size $O(d^2\log^2 d)$ such that in the partition defined by $R'$, the diameter of each cluster is at most $2w=\ell$, and each ball $B(u,\alpha\cdot w)$ intersects at most $ e^{4\alpha d}\cdot O(d\log d)$ clusters.
\end{proof}

\subsection{A $\poly(d)$-Time $(d^{1.5}, d+1)$-Hash: Proof of \Cref{thm:exist_decomp}}
\label{sec:proof_d15_hash}

\DecompTimeEfficient*
Let $\ell > 0$ and $\epsilon := \ell / \Gamma$, where $\Gamma = \Theta(d \sqrt{d})$.
We use the following construction to obtain a $(\Gamma, d + 1)$-hash
$\hash : \mathbb{R}^d \to \mathbb{R}^d$ with diameter bound $\ell$.
Let $T := \frac{\ell}{\sqrt{d}} = \Theta(d \epsilon)$ and partition $\mathbb{R}^d$ into hypercubes of side length $T$.
Note that each of these hypercubes is of diameter $\ell$.
Let $l_i:=(d-i)\epsilon$ for $i=0,1,\ldots,d-1$. 
For a point set $S\subseteq\mathbb{R}^d$, let $N_{i}^{\infty}(S)$
be the $l_i$-neighborhood of $S$ under the $\ell_\infty$ distance, i.e.,
\[
    N_{i}^{\infty}(S):=\{x\in \mathbb{R}^d:\exists s\in S, \Vert x-s\Vert_{\infty}\le l_i\}.
\]
\begin{enumerate}
    \item For $i=0,\ldots,d$, let $\mathcal{A}_i$ be the set of $i$-dimensional faces of the hypercubes,
    and let $\mathcal{B}_i:=N^{\infty}_{i}(\mathcal{A}_i)$.
    Define $\mathcal{B}_{-1}:=\emptyset$ and $\mathcal{B}_{\le i}:=\bigcup_{j\le i}\mathcal{B}_j$.
    \item For every $i=0,\ldots,d-1$, for every $i$-dimensional face $Q$,
    let $\hatQ := N^{\infty}_{i}(Q)\setminus \mathcal{B}_{\le i-1}$ and for every $x\in \hatQ$ assign $\hash(x)=q$ where $q\in \hatQ$ is an arbitrary but fixed point.
    Observe that $\widehat{Q}  \subseteq \mathcal{B}_i$
    and thus $x \in \widehat{Q}$ will \emph{not} be assigned
    again in later iterations. (See Figure~\ref{fig:hashing2D} for an illustration.)
    \item For every $d$-dimensional face $Q$, i.e., a hypercube,
    let $\hatQ$ be the remaining part of $Q$ whose $\hash(x)$ has not been assigned, 
    and assign $\hash(x)=q$ for every $x\in \hatQ$, where $q\in \hatQ$ is arbitrary but fixed point.
\end{enumerate}

\begin{figure}[t]
	\begin{center}
		\includegraphics[scale=1]{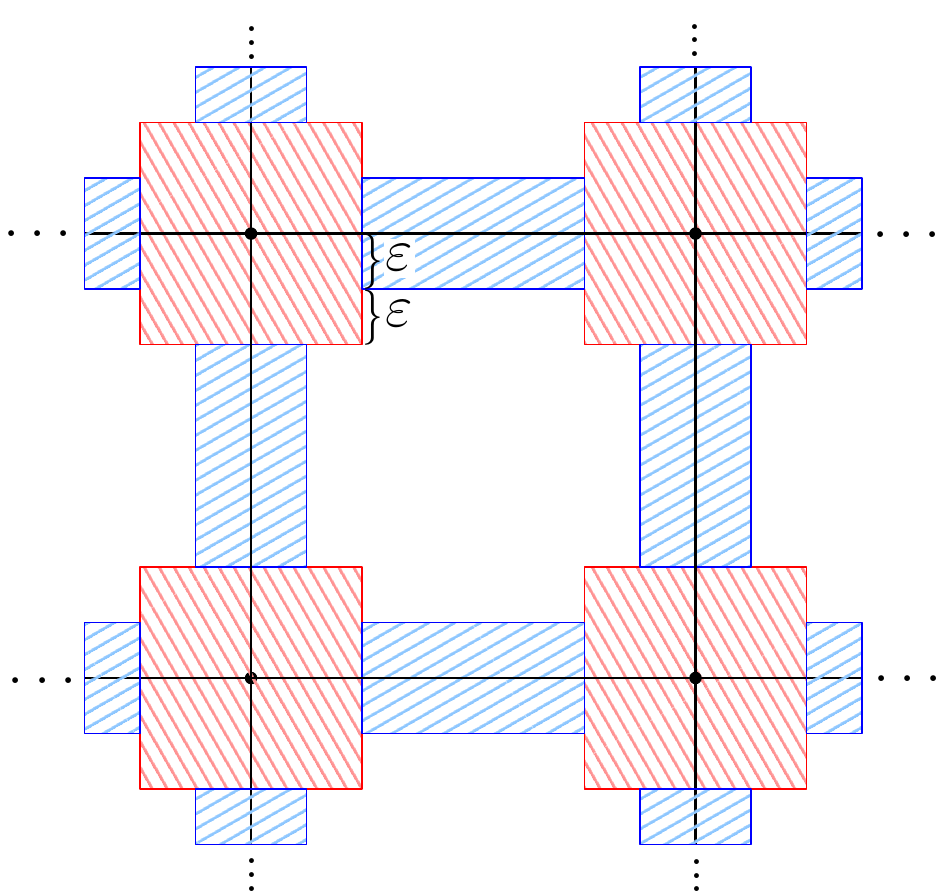}
		\caption{An illustration of the construction of a $(\Gamma, d + 1)$-hash in 2D, showing only a vicinity
			of a single square $[0,T]^2$ of the initial partition (in black).
			The $(l_0 = 2\epsilon)$-neighborhoods of vertices, which form $\mathcal{B}_0$, are depicted in red (north-west shaded areas),
			and the $(l_1 = \epsilon)$-neighborhoods of edges, that is, set $\mathcal{B}_1\setminus \mathcal{B}_0$, are depicted in blue (north-east shaded areas).}
		\label{fig:hashing2D}
	\end{center}
\end{figure}

In the following two lemmas, for any $0\le i\le d$, we show that every two $i$-dimensional faces of the hypercubes, including the $d$-dimensional interiors, are at distance larger than $\epsilon$ after we remove $\mathcal{B}_{\le i-1}$,
which is the union of the $l_i$-neighborhoods (under the $\ell_\infty$ distance) of $j$-dimensional faces for $j < i$.
This implies the property that every point set of diameter at most $\epsilon = \ell / \Gamma$ maps to at most $d+1$ buckets.

\begin{lemma}[Analysis for $i \leq d - 1$]
    \label{lemma:perp_gap_object}
    Suppose $0\le i\le d-1$ is an integer.
    Let $Q_1, Q_2$ be two distinct $i$-dimensional faces of the hypercubes.
    Let $\hatQ_1:=N^{\infty}_{i}(Q_1)\setminus \mathcal{B}_{\le i-1}$,
    and $\hatQ_2 := N^{\infty}_{i}(Q_2)\setminus \mathcal{B}_{\le i-1}$ as in the construction.
    Then $\dist(\hatQ_1, \hatQ_2)>\epsilon$.
\end{lemma}

\begin{proof}
    We assume that the intersection of $Q_1$ and $Q_2$ is non-empty, as otherwise
    \[
        \dist(\hatQ_1, \hatQ_2)\ge T - 2l_i
        \geq T - 2d\epsilon > \epsilon.
    \]

    Note that every $i$-dimensional face in the hypercube $[0, T]^d$ can be represented
    by fixing $d-i$ coordinates to be either $0$ or $T$,
    and other $i$ coordinates range over $[0, T]$.
    W.l.o.g., assume that $Q_1$ belongs to the hypercube $[0, T]^d$,
    and $Q_1=[0, T]^i\times 0^{d-i}$. 
    Let $Q=Q_1\cap Q_2$ and then $Q$ is a $j$-dimensional face of the hypercube,
    where $j < i$.
    W.l.o.g., assume $Q=[0, T]^j\times 0^{d-j}$, so $0\in Q_1\cap Q_2$.
    Since $Q \subseteq \mathcal{A}_{i-1}$, we have for all $t \in \{1, 2\}$
    \begin{equation}
        \label{eqn:subset}
        \hatQ_t = N_i^{\infty}(Q_t)\setminus \mathcal{B}_{\le i-1}
        \subseteq N_i^{\infty}(Q_t)\setminus N_{i-1}^{\infty}(Q).
    \end{equation}
    Define $I_1=\{1, \ldots, j\}, I_2=\{j+1, \ldots, i\}, I_3=\{i + 1, \ldots, d\}$.
Then by the definition of $N_i^{\infty}$,
    \begin{equation}
        \label{eqn:niq1}
        x\in N_i^{\infty}(Q_1) \iff
        \begin{cases}
                x_k\in [-l_i, T + l_i] & \forall k\in I_1\cup I_2\\
                |x_k|\leq l_i & \forall k\in I_3
        \end{cases}
    \end{equation}
    \begin{equation*}
x\in N_{i-1}^{\infty}(Q) \iff
        \begin{cases}
                x_k\in [-l_{i-1}, T + l_{i-1}] & \forall k\in I_1\\
                |x_k|\leq l_{i-1} & \forall k\in I_2\cup I_3
        \end{cases}
    \end{equation*}
    Then by combining the above with the fact that $l_i < l_{i-1}$, we have
    \begin{equation}
        \label{eqn:nintersect}
        x\in N_i^{\infty}(Q_1)\cap N_{i-1}^{\infty}(Q) \iff
        \begin{cases}
                x_k\in [-l_i, T + l_i] & \forall k\in I_1\\
                x_k\in [-l_i, l_{i-1}] & \forall k\in I_2\\
                |x_k|\leq l_i & \forall k\in I_3
        \end{cases}
    \end{equation}

    \paragraph{Case I: $Q_1$ and $Q_2$ Belong to The Same Subspace.}
    Since $Q_1$ and $Q_2$ are contained in the same linear subspace
    of dimension $i$ and they intersect,
    we can assume that $Q_2=[0, T]^j\times [-T, 0]^{i-j}\times 0^{d-i}$.
    By definition, 
    \[
        x\in N_i^{\infty}(Q_2) \iff
        \begin{cases}
                x_k\in [-l_i, T + l_i] & \forall k\in I_1\\
                x_k\in [-T - l_i, l_i] & \forall k\in I_2\\
                |x_k|\leq l_i & \forall k\in I_3
        \end{cases}
    \]
    By \eqref{eqn:niq1} and \eqref{eqn:nintersect},
    we have that for every $x \in N_i^{\infty}(Q_1) \setminus N_{i-1}^{\infty}(Q)$,
    there exists $p\in I_2$ such that $x_p>l_{i-1}$.
    Moreover, for every $y \in N_i^{\infty}(Q_2)$, we have $y_p\le l_i$. Thus
    \[
        \dist(x,y)=\| x - y \|_2
        \geq |x_p-y_p| > l_{i-1} - l_i = \epsilon
    \]
    which indicates that
    \[
        \dist\left(N_i^{\infty}(Q_1) \setminus N_{i-1}^{\infty}(Q), N_i^{\infty}(Q_2) \setminus N_{i-1}^{\infty}(Q)\right) > \epsilon.
    \]
    By combining \eqref{eqn:subset} and above, we have
    \[
        \dist(\hatQ_1, \hatQ_2) 
        \geq \dist\left(N_i^{\infty}(Q_1) \setminus N_{i-1}^{\infty}(Q), N_i^{\infty}(Q_2) \setminus N_{i-1}^{\infty}(Q)\right)
        > \epsilon
    \]
    \paragraph{Case II: $Q_1$ and $Q_2$ Belong to Different Subspaces.}
    We then consider the case where $Q_1$ and $Q_2$ are contained in different subspaces.
    Let $\mathcal{S}_1$ and $\mathcal{S}_2$ be the $i$-dimensional linear subspaces containing $Q_1$ and $Q_2$, respectively.
    By the definition of $Q_1$ and $Q_2$,
    we know that $\mathcal{S}_1$ and $\mathcal{S}_2$ are not identical,
    and each of them can be spanned by $i$ vectors from
    the orthonormal basis $\{e_1,e_2,\ldots,e_d\}$ of the Euclidean space.
    We start with the following Lemma~\ref{lemma:perp_gap}
    which proves the distance lower bound between
    the whole \emph{subspaces} (instead of bounded subsets that we need to deal with).

    \begin{lemma}
        \label{lemma:perp_gap}
        Suppose $0\leq i\leq d-1$ is an integer.
        Let $\mathcal{S}_1, \mathcal{S}_2$ be two different $i$-dimensional linear subspaces,
        where both $\mathcal{S}_1$ and $\mathcal{S}_2$ can be spanned by $i$ vectors from
        the orthonormal basis $\{e_1,e_2,\ldots,e_d\}$ of the Euclidean space.
        Let $\mathcal{I}=\mathcal{S}_1\cap \mathcal{S}_2$.
        Define $\mathcal{S}_1':=N_i^{\infty}(\mathcal{S}_1)\setminus N_{i-1}^{\infty}(\mathcal{I})$,
        and define $\mathcal{S}_2'$ similarly. Then
        $\dist(\mathcal{S}_1', \mathcal{S}_2') \geq \sqrt{2} \epsilon$.
    \end{lemma}

    \begin{proof}
        W.l.o.g., suppose $\mathcal{S}_1$ is spanned by $e_1,\ldots,e_i$
        and $\mathcal{S}_2$ is spanned by $e_1,\ldots,e_j,e_{i+1},\ldots,e_{d}$, where $0\leq j< i$, and $d=2i-j$. Then $\mathcal{I}=\mathcal{S}_1\cap\mathcal{S}_2$ is spanned by $e_1,\ldots,e_j$.
        
        Let $x=\sum_{k=1}^{d}x_ke_k$. Partition $\{1,\ldots,d\}$ into three parts as follows:
        $I_1=\{1,\ldots,j\}, I_2=\{j+1,\ldots,i\}, I_3=\{i+1,\ldots,d\}$.
        Then
        \[
            x\in N_{i-1}^{\infty}(\mathcal{I}) \iff
            |x_k|\leq l_{i-1}\epsilon, \quad\forall  k\in I_2\cup I_3.
        \]
        Similarly,
        \[
            x\in N_i^{\infty}(\mathcal{S}_1) \iff
            |x_k|\le l_i\epsilon, \quad \forall k\in I_3,
        \]
        and
        \[
            x\in N_i^{\infty}(\mathcal{S}_2) \iff
            |x_k|\le l_i\epsilon, \quad \forall k\in I_2.
        \]
        It follows that $x\in \mathcal{S}_1'=N_i^{\infty}(\mathcal{S}_1)\setminus N_{i-1}^{\infty}(\mathcal{I})$
        if and only if $|x_k|\le l_i \epsilon$ for every $k\in I_3$, 
        and $\exists p \in I_2$ such that $|x_p|>l_{i-i} \epsilon$.
        Also, $x\in \mathcal{S}_2'=N_i^{\infty}(\mathcal{S}_2)\setminus N_{i-1}^{\infty}(\mathcal{I})$
        if and only if $|x_k|\le l_i \epsilon$ for every $k\in I_2$, 
        and $\exists q\in I_3$ such that $|x_q|>l_{i-1}\epsilon$.
        
        For every $x=\sum_{k=1}^{d}x_ke_k\in \mathcal{S}_1'$ and $y=\sum_{k=1}^{d}y_ke_k\in \mathcal{S}_2'$,
        there exists $p\in I_2$ such that $|x_p|>l_{i-1} \epsilon$, and $q\in I_3$ such that $|y_q|>l_{i-1}\epsilon$.
        Note that $|y_p|\leq l_i\epsilon$, and $|x_q|\le l_i\epsilon$. Therefore, we have
        \[\dist(x,y)=\|x - y\|_2
        \geq \left\Vert(x_p-y_p)e_p+(x_q-y_q)e_q\right\Vert_2
\geq \sqrt{2}\epsilon\]
        where the first and the second inequalities are derived from the fact that $\{e_1,\ldots,e_{d}\}$ is orthonormal basis.
        This completes the proof of Lemma~\ref{lemma:perp_gap}.
    \end{proof}

    Let $\mathcal{I} := \mathcal{S}_1\cap \mathcal{S}_2$.
    In order to apply Lemma~\ref{lemma:perp_gap},
    we need the next Claim~\ref{claim:subspace} to relate the intersection of subspaces $\mathcal{I}$ to $Q$.
\begin{claim}
        \label{claim:subspace}
        For all $t \in \{1, 2\}$,
        \begin{equation}
            \label{eqn:subspace}
            N_i^{\infty}(Q_t)\setminus N_{i-1}^{\infty}(Q)
            = N_i^{\infty}(Q_t)\setminus N_{i-1}^{\infty}(\mathcal{I}).
        \end{equation}
    \end{claim}
    \begin{proof}
        Since the cases for $t = 1$ and $t = 2$ are symmetric, we only prove for $t = 1$, and $t = 2$ follows similarly.
        Recall that $Q_1=[0, T]^i \times 0^{d-i}$ and $Q=[0, T]^j \times 0^{d-j}$. Thus $\mathcal{S}_1$ is spanned by $e_1,\ldots,e_i$ and $\mathcal{I}$ is spanned by $e_1,\ldots,e_j$. By definition, 
        \begin{equation}
            \label{eqn:calI}
            x\in N_{i-1}^{\infty}(\mathcal{I}) \iff \forall k\in I_2\cup I_3,|x_k|\leq l_{i-1}
        \end{equation}
        Then by \eqref{eqn:niq1}, \eqref{eqn:calI} and the fact that $l_i < l_{i - 1}$, we have
        \[
            x \in N_i^{\infty}(Q_1)\cap N_{i-1}^{\infty}(\mathcal{I}) \iff \begin{cases}
            x_k\in [-l_i, T+l_i] & \forall k\in I_1\\
            x_k\in [-l_i, l_{i-1}] & \forall k\in I_2\\
            |x_k|\le l_i & \forall k\in I_3
            \end{cases}
        \]
        which is the same conditions for the set $N_i^{\infty}(Q_1) \cap N_{i-1}^{\infty}(Q)$ (see \eqref{eqn:nintersect}).
        Thus $N_i^\infty(Q_1) \cap N_{i - 1}^{\infty}(Q)
        = N_i^\infty(Q_1) \cap N_{i - 1}^\infty(\mathcal{I})$,
        and $N_i^{\infty}(Q_1)\setminus N_{i-1}^{\infty}(Q)=N_i^{\infty}(Q_1)\setminus N_{i-1}^{\infty}(\mathcal{I})$.
        This concludes the proof of Claim~\ref{claim:subspace}.
    \end{proof}
    Finally, by combining \eqref{eqn:subset} and Claim~\ref{claim:subspace}, we have
    \[
        \begin{aligned}
        \dist(\hatQ_1, \hatQ_2)
        &\ge \dist\left(N_i^{\infty}(Q_1)\setminus N_{i-1}^{\infty}(\mathcal{I}), N_i^{\infty}(Q_2)\setminus N_{i-1}^{\infty}(\mathcal{I})\right)\\
        &\ge \dist\left(N_i^{\infty}(\mathcal{S}_1)\setminus N_{i-1}^{\infty}(\mathcal{I}), N_i^{\infty}(\mathcal{S}_2)\setminus N_{i-1}^{\infty}(\mathcal{I})\right)\\
        &>\epsilon
        \end{aligned}
    \]
    where the last inequality is by Lemma~\ref{lemma:perp_gap}.
    This completes the proof of Lemma~\ref{lemma:perp_gap_object}.
\end{proof}
\begin{lemma}[Analysis for $i = d$]
    \label{lemma:remain_gap}
    Let $Q_{1},Q_{2}$ be two distinct hypercubes in the partition,
    and $\hatQ_{1}, \hatQ_{2}$ be the corresponding remaining parts constructed in step $3$.
    Then $\dist(\hatQ_{1}, \hatQ_{2})>\epsilon$.
\end{lemma}
\begin{proof}
    It suffices to show that for every $x \in \hatQ_1$ and $y \in \hatQ_2$,
    $\dist(x, y) > \epsilon$.
    By the construction, we know that $\dist(x, \partial Q_1) \geq l_{d - 1} = \epsilon$,
    and similarly $\dist(y, \partial Q_2) \geq \epsilon$, where for a hypercube $H$, 
    $\partial H \subseteq \mathbb{R}^d$ is the boundary of $H$.

    Let $\seg(x, y)$ be the segment with endpoints $x, y$.
    Pick $x' \in \partial Q_1 \cap \seg(x, y)$ such that $x'$ is closest to $x$,
    and pick $y' \in \partial Q_2 \cap \seg(x, y)$ such that $y'$ is closest to $y$.
    Then,
    \[
        \dist(x, y)
        \geq \dist(x, x') + \dist(y, y')
        \geq \dist(x, \partial Q_1) + \dist(y, \partial Q_2)
        \geq 2 \epsilon,
    \]
    where the first inequality follows from the fact that $Q_1$ and $Q_2$ have disjoint interior.
    This completes the proof of Lemma~\ref{lemma:remain_gap}.
\end{proof}
\begin{proof}[Proof of \Cref{thm:exist_decomp}]
    By \Cref{lemma:perp_gap_object}, \Cref{lemma:remain_gap} and the structure of the hypercubes,
    we know that for every subset $S \subseteq \mathbb{R}^d$ with $\Diam(S) \leq \epsilon$,
    $\hash(S)$ can intersect at most one $\epsilon$-neighborhood (which is defined with respect to $\ell_\infty$)
    of an $i$-dimensional face for every $0 \leq i \leq d$, since otherwise the diameter constraint of $S$ would be violated.
    Therefore, $|\hash(S)| \leq d+1$.
    To bound the diameter of the buckets, we note that
    for every $\widehat{Q}$ in the construction, we have
    $\Diam(\hatQ)< \sqrt{d}\cdot\max\{O(d\epsilon), T\}\leq \ell$.
Finally, it is immediate that one can evaluate $\hash(x)$
    at an input $ x\in \mathbb{R}^d$ in time and space $\poly(d)$.
This completes the proof of \Cref{thm:exist_decomp}.
\end{proof}

     \section{Lower Bound}

In this part, we give an $\Omega(\sqrt{n})$ space complexity lower bound for one-pass streaming algorithm for Uniform Facility Location with a small constant approximation ratio, even in insertion-only streams with $\Delta = \poly(n)$. This implies impossibility to design a streaming $1+\varepsilon$-approximation algorithm for arbitrary $\varepsilon>0$ with space polynomially depending on the dimension $d$ and polylogarithmic in $n$ and $\Delta$. With this result, the $\poly(d\cdot \log \Delta)$ space complexity of our streaming algorithms presented before is justified to be basically optimal. We first state our lower bound result, and then present the proof relying on the reduction from the Boolean Hidden Matching communication problem.

\begin{theorem}\label{thm:lb-of-ufl-in-log-space}
	For any dimension $d$ and every one-pass randomized algorithm that,
	given an insertion-only stream of points from $[\Delta]^d$ for $\Delta = 2^{O(d)}$,
	approximates Uniform Facility Location within ratio better than $1.085$,
	must use space $\exp(\poly(d\cdot \log\Delta))$.
	More precisely, the space lower bound for better than $1.085$-approximation is $\Omega(\sqrt{n})$
	in the setting with $d = \Theta(\log n)$ and $\Delta = \poly(n)$,
	where $n = |P|$ is the number of data points.
\end{theorem}

\paragraph{Boolean Hidden Matching (BHM).}
Recall that in the BHM one-way communication problem Alice gets a binary string $x\in \{0, 1\}^{2n}$, while
Bob receives a perfect matching $M$ on $[2n]$ and a binary string $w\in \{0, 1\}^n$ with one bit $w_{i,j}$ for each $(i,j)\in M$.
It is promised that either $x_i \oplus x_j = w_{i,j}$ for every pair $(i,j)\in M$ (YES instance),
or $x_i \oplus x_j \neq w_{i,j}$ for any $(i,j)\in M$ (NO instance).
The goal is for Alice to send a message to Bob so that he can distinguish the two cases with probability at least $2/3$.
The randomized one-way communication complexity of BHM is $\Omega(\sqrt{n})$ (see e.g.~\cite{GKKRdw09,VY11}).

We use a reduction from BHM to rule out a streaming approximation scheme for the Uniform Facility Location problem,
where the opening cost $\wopen$ is the same for all facilities.
We show the lower bound in a $\Theta(n)$-dimensional Euclidean space,
and then use a more general form of the Johnson–Lindenstrauss lemma from~\cite{DBLP:conf/stoc/NarayananN19} and the Kirszbraun theorem \cite{kirszbraun1934zusammenziehende} to induce the same result in a logarithmic dimension.
While even storing a single point in a $\Theta(n)$-dimensional takes $\Theta(n)$ space,
the point is to show that there is a gap between the optimal costs for YES and NO instances.
In dimension $\Theta(n)$, all input points for Facility Location will belong to $\{0, 1\}^d$, however, an optimal solution may use other points as facilities, too.
Moreover, each input point will have two or four 1s in its coordinates, while the rest of the coordinates will be set to 0; that is,
they can be interpreted as a pair or a four-tuple of coordinates set to~1.
We use $\wopen = 2$ for convenience (by scaling, the result holds for any $\wopen > 0$).

We present an instance of Uniform Facility Location problem in the $8n$-dimensional Euclidean space, with each $i \in [2n]$ independently assigned four coordinates, denoted as $i_0, i_1, i_2, i_3$, respectively.
For each Alice's bit $x_i$, we define two points at distance~2: $s^0_i$ with 1s in coordinates $i_0, i_1$, and $s^1_i$ with 1s in coordinates $i_2, i_3$.
If $x_i = 0$, then Alice adds $100n$~clients to point $s^0_i$ and no client to point $s^1_i$;
for $x_i = 1$, there are $100n$~clients at point $s^1_i$ and no client at point $s^0_i$. (Recall that clients are the data points in $P$. For simplicity, we view $P$ as a multiset in this section, i.e., allow duplicate points; however, by a small-enough perturbation of points, the lower bound holds even when points in $P$ are distinct. Note that 
out of points $s^0_i$ and $s^1_i$, exactly one is in $P$ with multiplicity $100n$, while the other one is not in $P$.) 

For each edge $(i,j)$ in the matching, Bob adds two clients:
one client at $t^0_{i,j}$, with 1s in coordinates $i_0, i_1, j_{2 - 2w_{i,j}}$, and $j_{3 - 2w_{i,j}}$,
and one client at $t^1_{i,j}$, with 1s in coordinates $i_2, i_3, j_{2w_{i,j}}$, and $j_{2w_{i,j} + 1}$
(note that these points have four coordinates set to 1); see Figure~\ref{fig:FacLoc} for an illustration.
Observe that points $t^b_{i,j}$ have disjoint sets of coordinates set to~1 and thus pairwise distances equal to $\sqrt{8}$.

\begin{figure}[!ht]
\begin{center}
\includegraphics[scale=1]{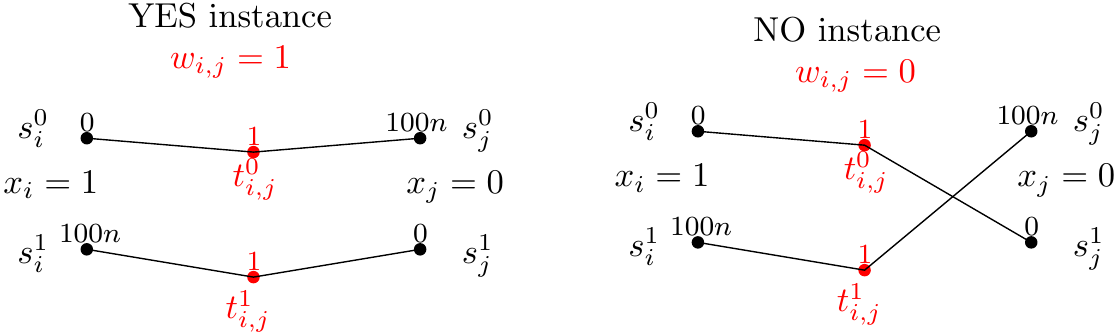}
\caption{An illustration of the reduction from BHM to an instance of Facility Location.
Alice's points are depicted in black, whereas Bob's points are in red.
The number above each point represents the number of clients in that point.
The edge depicts that the connected points are at distance $\sqrt{2}$, whereas 
black and red points not connected by an edge are at distance $\sqrt{6}$.
The distance between any two black points is always 2 and between any two red points always $\sqrt{8}$.}
\label{fig:FacLoc}
\end{center}
\end{figure}

It is crucial to observe that in either YES or NO case, the optimal solution contains a facility at $s_i^{x_i}$, for each $i = 1, \cdots, 2n$,
where the notation $s^{x_i}_i$ represents point $s^0_i$ if bit $x_i = 0$ and point $s^1_i$ otherwise;
see Lemma~\ref{lem:lb-no-case} for a detailed argument.
We show a separation between the costs in the YES and NO cases by the following lemmas.

\begin{lemma}\label{lem:lb-yes-case}
    In the YES case, the optimal cost is at most $2n \cdot (\wopen + \sqrt{2})$.
\end{lemma}

\begin{proof}
    In the optimal solution, we build a facility in all $2n$ Alice's points $s^{x_i}_i$ with many clients, and connect each client in $t^b_{i,j}$
    to the closest facility, which is at distance at most $\sqrt{2}$ as this is a YES instance.
    Then the total opening cost is $2n\cdot \wopen$  and the total connection cost equals $2n\cdot \sqrt{2}$,
    implying that the optimal cost is at most $2n\cdot (\wopen + \sqrt{2})$.
\end{proof}

\begin{lemma}\label{lem:lb-no-case}
    In the NO case, the optimal cost is $n \cdot (3\wopen + \sqrt{2})$.
\end{lemma}

\begin{proof}
    We first observe that in any optimal solution, a facility is built at $s_i^{x_i}$ for each $i = 1, \cdots, 2n$.
    Indeed, if there is no facility in $B(s_i^{x_i}, 0.1)$ for some $i$, then the connection cost for the clients at  $s_i^{x_i}$ is at least $10n$,  which means that we can decrease the total cost by building a facility at $s_i^{x_i}$.
    Thus, ball $B(s_i^{x_i}, 0.1)$ must contain a facility for any $i$.
    Furthermore, if there is a facility in $B(s_i^{x_i}, 0.1)$ but not at $s_i^{x_i}$,
    then all clients at $s_i^{x_i}$ are connected to this facility (to the closest one if there are more), thus
    by moving the facility to $s_i^{x_i}$, we decrease the total connection cost (this could increase the connection cost for Bob's clients, but there are only $2n$ of them).
    
    Next, we argue about the cost for Bob's clients.
    One possible solution is the following:
    For each edge $(i,j)$ in the matching, one of the clients sitting at $t^0_{i,j}$ and $t^1_{i,j}$ is at distance $\sqrt{2}$
    from facilities $s^{x_i}_i$ and $s^{x_j}_j$, so we connect this client to either of these two facilities,
    while we build a facility at the location of the other client, which is at distance at least $\sqrt{6}$ from any facility
    built in such a way. Thus, as $\wopen < \sqrt{6}$, it is cheaper to build a new facility for this client
    than to connect it to another facility.
    This solution has opening cost of $3n\cdot \wopen$ and connection cost of $n\cdot \sqrt{2}$, and we show that it is optimal.
    To this end, it suffices to rule out building a facility that serves more than one Bob's client.
    
    Suppose for a contradiction that the optimal solution contains a facility such that $k > 1$ of Bob's clients are connected to it.
    Such a facility cannot be placed at $s_i^{x_i}$ for some $i$
    as the connection cost is then $k\cdot \sqrt{6} > k\cdot \wopen$, i.e., it is cheaper to build one facility for each of the $k$ Bob's clients.
    Recall that these Bob's clients sit on vertices of a $(k-1)$-simplex with edge length $\sqrt{8}$.
    Since the sets of coordinates of 1s of different Bob's clients do not intersect,
    the best location for such a facility is to place it in the geometric median of the simplex, i.e.,
    the point that minimizes the sum of distances from the vertices (which is the connection cost).
    Clearly, the geometric median has coordinate $i$ equal to $1/k$ whenever there is a (single) vertex with~1
    at coordinate $i$, and to $0$ otherwise (thus, the geometric median is the average of the vertices). 
    The distance from each of the $k$ clients to this median is
    \[\sqrt{4\cdot (k-1)\cdot \left(\frac{1}{k}\right)^2 + 4\cdot \left(1 - \frac{1}{k}\right)^2}
    = 2\cdot \sqrt{\frac{1}{k} - \frac{1}{k^2} + 1 - \frac{2}{k} + \frac{1}{k^2}}
    = 2\cdot \sqrt{1 - \frac{1}{k}}\,.\]
    It follows that the connection cost for these $k$ clients plus the opening cost of the facility
    is $k\cdot 2\cdot \sqrt{1 - \frac{1}{k}} + \wopen > k\cdot \wopen$, which holds for any $\wopen \le 2$.
    Hence, it is cheaper to build $k$ facilities for these $k$ clients.
    
    Concluding, the optimal cost for a NO instance is $n\cdot (3\wopen + \sqrt{2})$.
\end{proof}

Combining Lemma~\ref{lem:lb-yes-case} and Lemma~\ref{lem:lb-no-case} and using $\wopen = 2$, a one-pass streaming algorithm for Uniform Facility Location in an $\Theta(n)$-dimensional space using space $o(\sqrt{n})$
with approximation ratio strictly better than
$(3\wopen + \sqrt{2}) / (2\wopen + \sqrt{8}) = (6 + \sqrt{2}) / (4 + \sqrt{8})\approx 1.085$ would solve the BHM problem.

In order to obtain the identical result in a logarithmic space, we introduce a more general form of Johnson-Lindenstrauss Lemma given by Narayanan and Nelson~\cite{DBLP:conf/stoc/NarayananN19}, as well as the Kirszbraun theorem given from~\cite{kirszbraun1934zusammenziehende}.

\begin{theorem}[\cite{DBLP:conf/stoc/NarayananN19}]
    \label{thm:general-JL-lemma}
    For all $\epsilon < 1$ and $x_1, \ldots, x_n \in \mathbb{R}^d$, there exists $m = O(\epsilon^{-2} \log n)$ and a (nonlinear) map $\pi : \mathbb{R}^d \to \mathbb{R}^m$ such that for all $1 \leq i \leq n$ and $u \in \mathbb{R}^d$, 
    \[
        (1 - \epsilon) \Vert u - x_i \Vert_2 \leq \Vert \pi(u) - \pi(x_i) \Vert_2 \leq (1 + \epsilon) \Vert u - x_i \Vert_2.
    \]
\end{theorem}

\begin{theorem}[\cite{kirszbraun1934zusammenziehende}]
    \label{Kirszbraun-theorem}
    For every subset $X \subset \mathbb{R}^d$ and $L$-Lipschitz map\footnote{Recall that $\varphi$ is $L$-Lipschitz if for all $x', x'' \in X$, we have $\Vert \varphi(x') - \varphi(x'') \Vert_2 \leq L \Vert x' - x''\Vert_2$.} $\varphi:X \to \mathbb{R}^m$, there exists an $L$-Lipschitz extension $\widetilde{\varphi}$ of $\varphi$ from $X$ to the entire space $\mathbb{R}^d$.
\end{theorem}

\begin{proof}[Proof of Theorem~\ref{thm:lb-of-ufl-in-log-space}]
    Given an instance of BHM, let $X = \{s_i^0, s_i^1 \mid i \in [2n]\} \cup \{t_{i,j}^0, t_{i,j}^1 \mid (i, j) \in M\}$
    be the set of data points of the corresponding Uniform Facility Location instance constructed as above.
    By Theorem~\ref{thm:general-JL-lemma}, there exists $\pi : \mathbb{R}^{8n} \to \mathbb{R}^d$ with $d = \Theta(\log n / \epsilon^2)$ and a small enough constant $\epsilon > 0$ such that for all $x \in X$ and $y \in \mathbb{R}^{8n}$,
    \[
        (1 - \epsilon) \Vert x - y \Vert_2 \leq \Vert \pi(x) - \pi(y) \Vert_2 \leq (1 + \epsilon) \Vert x - y \Vert_2.
    \]
    Thus, by Lemma~\ref{lem:lb-yes-case}, the optimal solution of Uniform Facility Location in the YES case after the dimension reduction by $\pi$ is at most $(1 + \epsilon)\cdot 2n \cdot (\wopen + \sqrt{2})$.

    It remains to show an upper bound of $n \cdot (3\wopen + \sqrt{2}) / (1 + \epsilon)$ on the optimal cost in the NO case after the dimension reduction.
    Letting $Y = \pi(X)$, the inverse map $\pi^{-1} : Y \to \mathbb{R}^{8n}$, where $\pi^{-1}(y)$ is defined as arbitrary element in $\{x \mid \pi(x) = y\}$, is $(1 + \epsilon)$-Lipschitz, since for all $y', y'' \in Y$,
    \[
        \Vert \pi^{-1}(y') - \pi^{-1}(y'') \Vert_2
        \leq \frac{1}{1 - \epsilon} \Vert y' - y'' \Vert_2
        \leq (1 + O(\epsilon)) \Vert y' - y'' \Vert_2.
    \]
    
    Using Theorem~\ref{Kirszbraun-theorem}, we extend the map $\pi^{-1} : Y \to \mathbb{R}^{8n}$ to $\widetilde{\varphi} : \mathbb{R}^d \to \mathbb{R}^{8n}$ while preserving the $(1 + \epsilon)$-Lipschitz property.
    Suppose, for the sake of contradiction, there is a solution in the NO case with cost strictly less than $n \cdot (3\wopen + \sqrt{2}) / (1 + O(\epsilon))$.
    Mapping the facilities opened in this solution back into $\mathbb{R}^{8n}$ by $\widetilde{\varphi}$, we obtain a solution in $\mathbb{R}^{8n}$ with cost strictly less than $n \cdot (3\wopen + \sqrt{2})$, which contradicts Lemma~\ref{lem:lb-no-case}.
    
    Thus a one-pass streaming algorithm for Uniform Facility Location in a $\Theta(\log n)$-dimensional space using memory of $o(\sqrt{n})$ with approximation ratio strictly better than
    \[
        \frac{3\wopen + \sqrt{2}}{(1 + O(\epsilon))^2(2\wopen + \sqrt{8})} \approx 1.085
    \]
    would solve the BHM problem, which contradicts the communication complexity of BHM.
    
    Finally, while the dimension reduction of Theorem~\ref{thm:general-JL-lemma} yields an instance where points have real coordinates,
    we can scale the instance by a $\poly(n)$ factor so that by rounding the coordinates to integers after this scaling
    we only change the optimal cost by an arbitrarily small factor in both YES and NO cases.
    Therefore, the lower bound holds even if the point set is restricted to $[\Delta]^d$ for $\Delta = \poly(n)$.
\end{proof}

     \section{Future Directions}
\label{sec:future}

\paragraph{Improved Hash Functions.}
Even though the tradeoff of parameters in \Cref{thm:EucSparsePartition} is already nearly-tight and it uses only $\poly(d)$ space to evaluate the hash value,
the running time for hash-value evaluation is still exponential.
It is thus an interesting question to find a more time-efficient construction, ideally one that can be evaluated in $\poly(d)$ time.

\paragraph{New One-Pass Approach.}
Since the gap bound determines the approximation ratio of the one-pass
algorithm of Section~\ref{s:1p}, our framework cannot be directly used for a better than 
$O(d / \log d)$-approximation in one pass using $\poly(d \cdot \log \Delta)$ space.
The main open problem is thus to design a one-pass streaming $O(1)$-approximation $\poly(d \log \Delta)$-space algorithm
for Uniform Facility Location in high-dimensional Euclidean spaces.

\paragraph{Multiple Passes.}
It is also interesting to explore the power of multiple passes.
In particular, is it possible to achieve $(1 + \epsilon)$-approximation
using, e.g., $\poly(d\cdot \log \Delta)$ passes?
Our lower bound works only for one pass,
and it would be interesting to strengthen it to $O(1)$-passes
using $\poly(d \cdot \log \Delta)$ space.

    \addcontentsline{toc}{section}{References}
    \bibliographystyle{alphaurl}
    \begin{small}	
        \bibliography{ref}
    \end{small}

    \begin{appendices}
        \section{Proof of Lemma~\ref{lemma:2dl0}}
\label{sec:proof_lemma_2dl0}

We restate the lemma for convenience.

\lemmaTwoLevelEllZero*

We show that this lemma can be proven using $\ell_0$-samplers together with
the simple but powerful fact that $\ell_0$-samplers can be obtained by linear sketching.
We state a lemma about the existence of such an $\ell_0$-sampler.

\begin{lemma}[$\ell_0$-Sampler; see e.g.~{\cite{CormodeF14}}]
    \label{lemma:ell_0-sampler}
    There is an algorithm that for integer $n\ge 1$, and every multiset $S$ of elements of $[n]$ with multiplicities bounded by $\exp(\poly(\log n))$ that is
    presented as a dynamic stream, succeeds with probability at least $1-1/\poly(n)$ and, conditioned on it succeeding,
    returns an element $j\in S$ together with its multiplicity $m_j$ in $S$ such that for every $s\in S$ it holds that $\Pr[j = s] = 1 / |S|$, where $|S|$ is the number of distinct elements in $S$.
    The algorithm has space and both update and query times
    bounded by $\poly(\log n)$,
    and its memory contents is a linear sketch of the frequency vector $\mathbf{m}$ of $S$.
\end{lemma}

\begin{proof}[Proof sketch of Lemma~\ref{lemma:2dl0}.]
    The idea is to use the $\ell_0$-sampler of Lemma~\ref{lemma:ell_0-sampler} in two levels, one for rows and one for columns of $M$,
    with the column level being encoded in the multiplicity of the row level
    (a similar trick has been used e.g.~in~\cite{FIS08}).
    In more detail, we instantiate an $\ell_0$-sampler for rows of $M$ such that
    the (binary representation of) multiplicity of any row $i$ additionally encodes the memory contents
    of the $\ell_0$-sampler for the columns of row $i$ (the true multiplicity of a row is also stored, which is needed to return the row sum).
    As an $\ell_0$-sampler takes $\poly(\log n)$, the multiplicity of any row
    will be bounded by $\exp(\poly(\log n))$, as desired.
    Since the memory contents of the column $\ell_0$-sampler is a linear sketch of row $i$, 
    its representation is an all 0s vector if and only if row $i$ is an all 0s vector.
    
    The randomness (namely, hash functions) for column $\ell_0$-samplers of all rows will be the same. Using this and the linearity of $\ell_0$-samplers, an update to a column $\ell_0$-sampler in row $i$ can be turned into an update of the multiplicity of row $i$ in the row $\ell_0$-sampler.
    Thus, upon an update $a$ to $M_{i,j}$, we first compute the update to the column $\ell_0$-sampler from $j$ and $a$ (this update is independent of $i$). Then we view this update as a change of the multiplicity of row $i$, and feed it to the entry for row $i$ in the row $\ell_0$-sampler.

    To sample a desired cell of $M$, we query the row $\ell_0$-sampler to obtain
    a u.a.r.\ non-zero row $i$, together with the multiplicity of row $i$, which encodes
    the memory contents of the column $\ell_0$-sampler for row $i$.
    We then sample from this column $\ell_0$-sampler to obtain a u.a.r.~column $j$ such that $M_{i,j}\neq 0$.
    We get the desired row sum from the row $\ell_0$-sampler.
\end{proof}

\section{Lower Bounds Based on Indexing}
\label{sec:LBs-index}

In this appendix, we give details of our lower bound arguments that are based on the hardness of indexing in one-way communication. Recall that the one-way communication 
problem of INDEX is defined as follows: Alice holds a set of bits $x_1, \dots, x_k$,
and Bob has an index $j\in [k]$. The goal is for Alice to send a message
to Bob so that Bob can determine $x_j$;
crucially, Bob cannot send any information to Alice.
It is well-known that the message of Alice needs to have $\Omega(k)$ bits for Bob
to succeed with constant probability~\cite{KNR99} (see also~\cite{KushilevitzNisan97,JKS08}).

\paragraph{Lower Bound for $\kMedian$.}
We show that an algorithm with space $o(k)$ that determines whether the $\kMedian$ cost
is 0 or 1 would solve the INDEX problem.
This holds even in the one-dimensional case and also for related types of clustering, such as $k$-\textsc{Means}.
Indeed, given such an algorithm, Alice inserts $k$ points, one point $2\cdot i + x_i$ for each of bit $x_i$. Then she sends the memory contents of this algorithm to Bob 
who inserts point $2\cdot j$. Thus the instance has $k+1$ distinct points if $x_i = 1$
and otherwise, it has $k$ distinct points. It follows that the $\kMedian$ cost
equals 1 iff $x_i = 1$ and otherwise, it equals 0.

\paragraph{Lower Bound for Answering Queries about The $r_p$ Values.}
In Section~\ref{s:1p} we discussed a straightforward one-pass implementation of our approach
in arbitrary-order streams that would be similar to the implementation in two passes
or in random-order streams.
However, as we only get the sampled points at the end of the stream,
we would need to estimate the $r_p$ value for each of the sampled points \emph{a posteriori},
that is, after processing the stream.

We now argue that an algorithm giving a $c$-approximation of $r_p$ for any query point $p$
would solve the INDEX problem. Indeed, given such an algorithm,
we set $\wopen = 1$, and for every $i = 1, \ldots, n$,
Alice inserts $2$ points at position $i$ if $x_i = 0$ and $\lceil 2c+1\rceil$
points at position $i$ otherwise. After that, for a point $p$ at position $i$, 
we have that $r_p = 0.5$ if $x_i = 0$ and $r_p\le 1/(2c)$ if $x_i = 1$.
Then Bob queries the $r_p$ value at position $j$.
Since the returned value is a $c$-approximation of $r_p$, Bob can distinguish between $r_p = 0.5$ and $r_p\le 1/(2c)$ and hence, determine the value of $x_j$.

 \end{appendices}
\end{document}